\documentclass[11pt]{article}

\usepackage[margin=0.88in,headheight=24pt]{geometry}
\usepackage[T1]{fontenc}
\usepackage{lmodern}
\usepackage{microtype}
\usepackage{graphicx}
\usepackage{multirow}
\usepackage{amsmath,amssymb,amsfonts}
\usepackage{amsthm}
\usepackage{mathrsfs}
\usepackage[title]{appendix}
\usepackage[table]{xcolor}
\usepackage{textcomp}
\usepackage{manyfoot}
\usepackage{booktabs,longtable,array,tabularx}
\usepackage{algorithm}
\usepackage{algorithmicx}
\usepackage{algpseudocode}
\usepackage{listings}
\usepackage{hyperref}
\usepackage{enumitem}
\usepackage{caption}
\usepackage{float}
\usepackage{subcaption}
\usepackage[numbers]{natbib}
\usepackage{titlesec}
\usepackage{fancyhdr}
\usepackage[most]{tcolorbox}
\newtcolorbox{revblock}{enhanced,breakable,colback=yellow!24,colframe=yellow!24,boxrule=0pt,arc=0pt,left=1mm,right=1mm,top=0.6mm,bottom=0.6mm,before skip=2pt,after skip=2pt}

\usepackage{tikz}
\usetikzlibrary{arrows.meta,positioning,calc}
\usepackage{lastpage}
\usepackage{etoolbox}
\usetikzlibrary{fit}

\definecolor{REEBLRAMaroon}{HTML}{6F1726}
\definecolor{REEBLRAMaroonDark}{HTML}{4F0F1A}
\definecolor{REEBLRAGreen}{HTML}{0F6B3D}
\definecolor{REEBLRAGreenDark}{HTML}{084B2B}
\definecolor{REEBLRAGold}{HTML}{D6A326}
\definecolor{REEBLRAGoldSoft}{HTML}{FFF3C7}
\definecolor{REEBLRABlue}{HTML}{1F5F99}
\definecolor{REEBLRABlueSoft}{HTML}{EAF4FF}
\definecolor{REEBLRACream}{HTML}{FFFAF0}
\definecolor{REEBLRAInk}{HTML}{24302A}
\definecolor{REEBLRAMuted}{HTML}{64706A}
\definecolor{REEBLRALightGreen}{HTML}{EDF8F1}
\definecolor{REEBLRALightMaroon}{HTML}{FAEEF1}

\hypersetup{
  colorlinks=true,
  linkcolor=REEBLRAMaroon,
  urlcolor=REEBLRABlue,
  citecolor=REEBLRAGreenDark,
  pdfauthor={Jomar F. Rabajante et al.},
  pdftitle={Learning Epidemic Threshold Equations from Simulation: REEBLRA, an Interpretable Bayesian Classification Framework for Reproduction Numbers.},
  pdfsubject={Methodological article introducing REEBLRA; the accompanying Shiny application is an implementation and demonstration}
}

\color{REEBLRAInk}
\setlist[itemize]{topsep=0.25em,itemsep=0.15em,leftmargin=1.55em}
\setlist[enumerate]{topsep=0.25em,itemsep=0.18em,leftmargin=1.75em}

\titleformat{\section}
  {\Large\bfseries\color{REEBLRAMaroon}}
  {\thesection}{0.65em}{}
  [\vspace{-0.55em}\color{REEBLRAGold}\rule{\textwidth}{1.1pt}]
\titlespacing*{\section}{0pt}{1.55em}{0.70em}
\titleformat{name=\section,numberless}
  {\Large\bfseries\color{REEBLRAMaroon}}{}{0pt}{}
  [\vspace{-0.55em}\color{REEBLRAGold}\rule{\textwidth}{1.1pt}]
\titleformat{\subsection}
  {\large\bfseries\color{REEBLRAGreenDark}}
  {\thesubsection}{0.60em}{}
\titleformat{\subsubsection}
  {\normalsize\bfseries\color{REEBLRAMaroonDark}}
  {\thesubsubsection}{0.55em}{}

\renewcommand{\headrulewidth}{0.7pt}
\renewcommand{\headrule}{\hbox to\headwidth{\color{REEBLRAMaroon}\leaders\hrule height \headrulewidth\hfill}}
\renewcommand{\footrulewidth}{0.35pt}
\renewcommand{\footrule}{\hbox to\headwidth{\color{REEBLRAGreen}\leaders\hrule height \footrulewidth\hfill}}

\providecommand{\backmatter}{}

\newtcolorbox{abstractbox}{
  enhanced,breakable,arc=2mm,boxrule=0.8pt,
  colback=white,colframe=white,
  borderline west={3pt}{0pt}{white},
  title=Abstract,fonttitle=\bfseries\large,
  colbacktitle=white,coltitle=black
}
\renewenvironment{abstract}{\begin{abstractbox}\small}{\end{abstractbox}}

\newtcolorbox{methodbox}[1][]{
  enhanced,arc=2mm,boxrule=0.7pt,
  colback=REEBLRALightGreen,colframe=REEBLRAGreen,
  borderline west={3pt}{0pt}{REEBLRAGreen},
  title=#1,fonttitle=\bfseries,coltitle=white,colbacktitle=REEBLRAGreen
}
\newtcolorbox{cautionbox}[1][]{
  enhanced,breakable,arc=2mm,boxrule=0.7pt,
  colback=REEBLRALightMaroon,colframe=REEBLRAMaroon,
  borderline west={3pt}{0pt}{REEBLRAMaroon},
  title=#1,fonttitle=\bfseries,coltitle=white,colbacktitle=REEBLRAMaroon
}

\lstdefinestyle{rcode}{
  language=R,
  basicstyle=\ttfamily\small,
  breaklines=true,
  columns=fullflexible,
  frame=single,
  rulecolor=\color{REEBLRAGreen},
  backgroundcolor=\color{REEBLRACream},
  showstringspaces=false,
  commentstyle=\color{REEBLRAGreenDark},
  keywordstyle=\color{REEBLRAMaroon}\bfseries,
  stringstyle=\color{REEBLRABlue}
}

\newcommand{\maybeincludegraphics}[2][]{%
\IfFileExists{#2}{\includegraphics[#1]{#2}}{%
\fbox{\begin{minipage}[c][0.25\textheight][c]{0.88\linewidth}\centering\textbf{Figure file not found:} #2\\[0.5em]
The figure environment, caption, and label are retained. Place the image file in the LaTeX working directory to replace this placeholder.\end{minipage}}}}

\newcommand{\tablehead}{\rowcolor{REEBLRAMaroon!92}}
\newcommand{\HeaderCell}[1]{\textcolor{white}{\bfseries #1}}

\tikzset{
  compartment/.style={draw=REEBLRAMaroon,very thick,rounded corners=2pt,minimum width=15mm,minimum height=9mm,fill=REEBLRACream,font=\bfseries},
  flow/.style={-{Latex[length=2.5mm]},very thick,draw=REEBLRAGreenDark},
  loss/.style={-{Latex[length=2.2mm]},thick,draw=REEBLRAMuted},
  flowlabel/.style={font=\small,fill=white,inner sep=1pt,text=REEBLRAInk}
}

\theoremstyle{definition}

\newtheorem{remark}{Remark}
\newtheorem{proposition}{Proposition}
\newtheorem{corollary}{Corollary}

\begin{document}

\begin{center}
\begin{tcolorbox}[
  enhanced,arc=3mm,boxrule=0pt,
  colback=white,colframe=white,
  borderline south={1.1 pt}{0pt}{REEBLRAGold},
  left=5mm,right=5mm,top=4mm,bottom=4mm
]
\centering
{\color{REEBLRAMaroonDark}\fontsize{17.5}{21}\selectfont\bfseries
Learning Epidemic Threshold Equations\\ from Simulation: REEBLRA, an Interpretable Bayesian Classification Framework for Reproduction Numbers\par}
\end{tcolorbox}

\vspace{0.45em}
{\large\bfseries Jomar F. Rabajante$^{1,*}$, Norvin P. Bansilan$^{1}$,\\
Maria Czarina T. Lagura-Arreza$^{1}$, Destiny SM. Lutero$^{1}$, and Emerson R. Rico$^{1}$\par}
\vspace{0.35em}
{$^{1}$Institute of Mathematical Sciences, University of the Philippines Los Ba\~nos, Laguna 4031, Philippines\par}
{\small $^{*}$Corresponding author: \href{mailto:jfrabajante@up.edu.ph}{jfrabajante@up.edu.ph}\par}
\end{center}

\vspace{0.35em}

\begin{abstract}
The basic reproduction number, \(R_0\), is central to infectious-disease modelling, but deriving an analytic expression can become difficult for high-dimensional, nonlinear, or structurally complex models. We introduce REEBLRA (\(R_0\) Equation Estimation using Bayesian Logistic Regression Algorithm), an interpretable simulation-to-threshold framework for learning explicit symbolic surrogate equations that identify parameter combinations associated with epidemic invasion or non-invasion. Given $\theta$ as vector of model parameters, rather than globally regressing the magnitude of \(R_0\), REEBLRA learns \(\Pr\{R_0(\boldsymbol\theta)>c\}\), directly targeting epidemiologically relevant thresholds such as \(R_0=1\). REEBLRA yields two complementary outputs: a binary threshold equation that distinguishes \(R_0>1\) from \(R_0<1\), and a multiclass bracket model that uses successive threshold equations to approximate the numerical magnitude of \(R_0\), interpreted as the expected number of secondary infections generated by a typical infectious individual in a fully susceptible population. This makes the framework useful for rapid parameter screening, sensitivity interpretation, threshold classification, exploratory analysis of complex simulators, and approximate reproduction-number estimation when direct analytic derivation is cumbersome. Bayesian logistic regression with standardized polynomial-like feature bases and proper weakly informative priors provides regularization, posterior uncertainty, and finite estimates under complete or quasi-complete separation. We further establish a mechanistic--surrogate concordance principle relating the {full fitted-score gradient} to the local gradient of \(\log R_0\) {under explicit regularity and tangency assumptions}, providing a basis for conditional interpretation of local threshold directions. Across seven {analytic-label benchmarks associated with ODE models}, including a 15-compartment five-group SIR system, five-fold binary AUC was 0.999--1.000 and ordered-bracket accuracy was 0.951--0.998. REEBLRA complements, rather than replaces, next-generation-matrix and other mechanistic derivations, providing an interpretable bridge between simulation, machine learning, threshold analysis, and reproduction-number modelling. An accompanying Shiny application supports practical implementation.
\end{abstract}



\begin{methodbox}[Accompanying Shiny App]
An interactive implementation of the REEBLRA workflow, including data generation, model fitting, diagnostics, visualization, and downloadable results, is available at:
\url{https://jfrabajanteuplb.shinyapps.io/REEBLRA}
\end{methodbox}

\newpage
\section{Introduction}\label{sec:introduction}

In mathematical epidemiology, threshold quantities are used to summarize whether an infectious disease can invade or persist in a population. The basic reproduction number, denoted by $R_0$, is commonly defined as the expected number of secondary infections produced by a typical infectious individual in a wholly susceptible population \citep{diekmann1990,vandendriessche2002,heesterbeek1996}. Under standard assumptions for many deterministic compartmental models, the condition $R_0<1$ is associated with local disease extinction, whereas $R_0>1$ is associated with initial epidemic growth or disease invasion. Although the interpretation of $R_0$ depends on model structure and assumptions, it remains one of the most widely used summary quantities in infectious disease modelling, from classical compartmental theory to modern structured epidemic models \citep{kermack1927,hethcote2000}.

The next generation matrix (NGM) method is a standard approach for deriving $R_0$ from compartmental models \citep{diekmann1990,vandendriessche2002,diekmann2010}. It provides a mathematically principled way to identify new-infection and transition terms, construct next-generation operators, and interpret the reproduction number as a spectral radius. However, for large, nonlinear, multi-host, age-structured, vector-borne, intervention-rich, or highly parameterized models, symbolic derivation may become cumbersome. Alternative approaches, including graph-theoretic methods, early-growth estimation, plug-and-play inference, and computer-simulation-based threshold calculations, have therefore been explored as complements to analytic derivation \citep{dejong1992,heffernan2005,breban2007,king2008,tennant2018}. Simulation-based analyses are especially useful when latent stages, heterogeneous mixing, or distributed infectious periods alter the relationship between parameters and epidemic growth \citep{wearing2005,lloyd2001}.

The interpretation and numerical estimation of $R_0$ require careful attention to model assumptions and data-generating mechanisms. $R_0$ is a context-dependent quantity whose value is shaped by population structure, contact patterns, model formulation, and epidemiological assumptions, rather than an invariant biological property of a pathogen \citep{Delamater2019}. When $R_0$ is inferred from early epidemic growth, the relationship between the observed growth rate and the reproduction number depends strongly on the generation-interval distribution \citep{WallingaLipsitch2007}. Epidemic trajectories including early epidemic phase may deviate from simple exponential growth because of spatial structure, behavioral changes, population heterogeneity, and other mechanisms \citep{Chowell2016}. From a computational perspective, Obadia et al.~\citep{Obadia2012} developed the $R_0$ package as a practical toolbox for estimating reproduction numbers from epidemic data using several established approaches. Together, these studies support the need to distinguish mechanistically derived reproduction numbers from growth-based or data-driven estimates and to report explicitly the assumptions, fitting windows, conversion formulas, and parameter domains underlying any numerical $R_0$ approximation.

We propose REEBLRA, a simulation-to-threshold method for learning interpretable epidemic threshold equations. Its central design choice is to estimate the geometry of $R_0(\boldsymbol\theta)=c$ rather than to reconstruct the full-valued response surface $R_0(\boldsymbol\theta)$. Parameter sets are sampled, the dynamical model is simulated or otherwise evaluated, and each run is labelled according to whether a mechanistic, effective, or numerical reproduction-number summary crosses a scientifically chosen cutoff. Bayesian logistic regression then learns the corresponding threshold-exceedance surface. Ordered successive thresholds retain coarse information about reproduction-number magnitude without requiring a single global direct-regression equation.

The accompanying REEBLRA Shiny application is an implementation and demonstration of the method. It supports built-in and custom ODE simulation, uploaded parameter--time series tables, early-growth targets, Bayesian and ordinary binary logistic equations, successive ordered-bracket equations, coefficient heatmaps, diagnostics, sensitivity summaries, and publication-oriented exports. Numerical integration and model evaluation follow established R workflows based on \texttt{deSolve}, classification diagnostics, ROC analysis, and Bayesian computation \citep{soetaert2010,kuhn2008,robin2011,carpenter2017,gabry2019}.

Machine learning has already been used to construct fast epidemic surrogates that predict outputs including $R_0$, and recent work has used machine learning to infer reproductive-number summaries from simulated or laboratory-derived signals \citep{pereira2021,kurul2025,khan2026}. Most directly relevant, the recently published Math-model Informed Neural Network (MINN) framework represents the linearized infection subsystem with a neural network to estimate $R_0$ across homogeneous and heterogeneous infectious-disease models \citep{lenz2027}. Symbolic-regression studies have also recovered governing epidemic ODEs from trajectory data \citep{brum2026}. These approaches establish that machine learning can accelerate epidemic-model analysis, but they learn a numerical reproduction number, simulator response, or governing dynamics. REEBLRA targets a different mathematical object: the {parameter-space threshold surface} $R_0(\boldsymbol\theta)=c$, represented by a directly inspectable logit equation and an ordered family of successive threshold equations. Its additional mechanistic--surrogate concordance test asks whether that learned boundary points in the same local direction as the analytic reproduction-number surface.

Threshold/level-set learning has precedents in adaptive computer experiments and excursion-set estimation \citep{picheny2010,bect2012}; likewise, cumulative-threshold and partially non-proportional ordinal regressions are established approaches \citep{mccullagh1980,peterson1990}. The application-specific contribution here is the combination of explicit epidemic-threshold equations, an ordered-bracket representation, and a mechanistic comparison of the {full fitted-score gradient} with the reproduction-number gradient. 

The distinctive contribution is therefore twofold (Table~\ref{tab:novelty-positioning}). First, REEBLRA provides a transparent simulation-calibrated bridge between mechanistic ODE modelling and interpretable threshold learning. Second, we introduce a local mechanistic--surrogate concordance analysis that connects the directions of full fitted-score derivatives to the gradient geometry of the original reproduction-number surface, with direct coefficient interpretation reserved for the raw linear feature basis. This converts coefficient interpretation from an informal ``positive parameter/negative parameter'' reading into a falsifiable model-checking principle. When a rigorous analytic reproduction number is available, REEBLRA can be tested against it; when derivation is difficult or unavailable, the same workflow can be used exploratorily to identify candidate threshold directions that require subsequent mathematical verification. The learned equation is never presented as a substitute for a next-generation-matrix derivation.

\begin{table*}[!htbp]
\centering
\caption{Positioning of REEBLRA relative to representative analytic and machine-learning approaches. The table distinguishes the mathematical object being learned or derived; it is not a performance ranking.}
\label{tab:novelty-positioning}
\small
\begin{tabularx}{\textwidth}{p{0.19\textwidth}p{0.24\textwidth}Xp{0.25\textwidth}}
\toprule
Approach & Primary object & Typical output & Relation to REEBLRA \\
\midrule
Next-generation matrix \citep{diekmann1990,vandendriessche2002,diekmann2010} & Mechanistic infection/transition operator & Exact or analytic $R_0$ under stated assumptions & Gold-standard mechanistic comparator when derivation is tractable. \\
Neural epidemic surrogates \citep{pereira2021,kurul2025} & Simulator response or trajectory map & Fast predictions of states, epidemiological outputs, or $R_0$ & REEBLRA instead targets an inspectable parameter-space threshold boundary rather than a general forward emulator. \\
Math-model-informed $R_0$ neural network (MINN) \citep{lenz2027} & Linearized infection subsystem represented by a neural network & Numerical $R_0$ across homogeneous and heterogeneous mechanistic models & Closest recent ML comparator: REEBLRA instead returns explicit threshold equations and tests their coefficient geometry against the mechanistic $R_0/R_e$ surface. \\
Symbolic regression \citep{brum2026} & Governing dynamical equations & Recovered ODE expressions from trajectories & REEBLRA assumes the simulator/mechanism is supplied and learns the threshold surface induced by that mechanism. \\
Data-to-$R$ machine learning \citep{khan2026} & Reproduction-number summary from empirical signals & Estimated transmission summary from polymerase chain reaction (PCR) cycle threshold (Ct) or related data & REEBLRA maps model parameters to threshold exceedance rather than mapping surveillance measurements directly to $R$. \\
Ordinal regression and level-set surrogates \citep{mccullagh1980,peterson1990,picheny2010,bect2012} & Ordered exceedance probabilities or simulator level sets & Ordinal class probabilities or targeted level-set estimates & Pre-existing statistical foundations; REEBLRA specializes these ideas to epidemic thresholds with explicit mechanistic-gradient checks. \\
REEBLRA & Parameter-space surface $R_0(\boldsymbol\theta)=c$ & Binary Bayesian logit boundary, ordered successive-threshold equations, validation, and mechanistic-gradient concordance & Adds a directly testable link between surrogate coefficient geometry and the analytic $R_0/R_e$ surface. Can be extended to heterogeneous, stochastic, and large systems.\\
\bottomrule
\end{tabularx}
\end{table*}

\section{Conceptual framing}\label{sec:framing}

A standard model-analysis workflow in mathematical epidemiology may include proving non-negativity and boundedness of solutions, determining equilibria, deriving $R_0$, analyzing stability, fitting parameters, and performing sensitivity analysis. In practice, sensitivity analysis often requires many numerical simulations across a parameter space. REEBLRA uses these simulation runs for an additional purpose: to estimate an interpretable probability equation for threshold exceedance.

REEBLRA is intentionally designed to be an interpretable Bayesian classification framework that learns reproduction-number relationships directly from simulated epidemic models and produces two complementary outputs:

\begin{enumerate}
    \item \textit{Binary threshold model.} REEBLRA learns an explicit threshold equation that classifies parameter combinations according to whether $R_0>1$ or $R_0<1$, thereby identifying parameter combinations associated with epidemic growth or decline.

    \item \textit{Sequential ordered multiclass bracket model.} REEBLRA learns a sequence of threshold equations that places $R_0$ into ordered numerical ranges, providing an interpretable approximation of the magnitude of $R_0$---that is, the expected number of secondary infections generated by a typical infectious individual in an otherwise fully susceptible population.
\end{enumerate}

Thus, REEBLRA can be used both to recover the epidemiologically important $R_0=1$ threshold and to approximate the numerical scale of $R_0$ through ordered reproduction-number brackets.

Let $\boldsymbol{\theta}=(\theta_1,\ldots,\theta_p)$ denote the vector of model parameters. Suppose that for each sampled parameter set $\boldsymbol{\theta}^{(i)}$, a model is simulated and assigned an outcome $y^{(i)}$. The outcome may be based on an analytic reproduction number, an effective reproduction number, a numerical growth-rate estimate, or an observed case-increase criterion. For a threshold $c$, the binary target is
\[
y_i(c)=\begin{cases}
1, & R(\boldsymbol{\theta}^{(i)})>c,\\
0, & R(\boldsymbol{\theta}^{(i)})\le c,
\end{cases}
\]
where $R$ may refer to analytic $R_0$, effective $R_e(0)$, or numerical \texttt{R0estimate}, depending on the selected threshold source. REEBLRA then uses logistic regression to estimate the probability ($p$) that a reproduction-number summary exceeds a specified threshold. Although the fitted probability is not itself \(R_0\), it provides information about the relationship between model parameters and \(R_0\), particularly the likelihood of exceeding the chosen cutoff. The resulting equation serves as a smooth, interpretable classification surrogate for identifying parameter combinations associated with reproduction numbers above or below the threshold.

For a fixed parameter vector in a deterministic model with a well-defined analytic $R_0$, the true event $\mathbf 1\{R_0(\boldsymbol\theta)>c\}$ is binary; a fitted logistic probability describes the classifier's model-based uncertainty or soft classification, not by itself the physical probability of an outbreak. A Bayesian posterior for the logistic coefficients quantifies uncertainty conditional on the classifier, its prior, and the training design; it does not automatically propagate uncertainty in the epidemic mechanism or the estimated target $R$.

The distinction is important. A logistic regression equation should not be described as the exact $R_0$ equation unless it is explicitly derived to match a known mechanistic or analytic framework. More precisely, the equation estimates a threshold-exceedance probability, and the decision boundary $\hat{\eta}(\boldsymbol{\theta})=0$ approximates the threshold surface $R_0=c$ within the simulated, custom-simulated, generated, or uploaded parameter domain.

\begin{remark}[Interpretation of the symbolic surrogate]
The symbolic equation produced by REEBLRA is a fitted logit-score equation for threshold exceedance rather than the analytic reproduction-number formula itself. Nevertheless, the structure, signs, and relative magnitudes of its coefficients can provide insights into the underlying $R_0$ relationship, and, when parameter values are specified, the binary and multiclass models can be used to assess threshold exceedance and approximate the corresponding $R_0$ value range. In the Shiny app implementation, the default predictor is a first-order logit polynomial in the selected parameters after standardization. If users supply transformed variables or future versions add interaction or polynomial terms, the same framework can represent richer surrogate threshold surfaces. In all cases, the learned equation should be described as a simulation-calibrated symbolic surrogate rather than a replacement for a next-generation-matrix derivation.
\end{remark}

\subsection{Why threshold classification rather than direct reproduction-number regression?}\label{subsec:why-classification}

REEBLRA treats the problem as classification because the primary scientific question is often threshold-based: is $R_0$ above or below a value such as one, and how does the parameter combination move the system across that boundary? Direct regression of $R_0(\boldsymbol\theta)$ is a valid alternative when accurate reproduction-number magnitude is the principal target, but it solves a harder and different approximation problem. A direct surrogate must represent the response over the full sampled range, including tails, strong curvature, multiplicative effects, and unknown parameter interactions. If the regression family is misspecified, magnitude errors can therefore become large in regions that may be irrelevant to the threshold decision.

For a threshold $c$, by contrast, classification only requires the sign of $R_0(\boldsymbol\theta)-c$. A direct estimate $\widehat R_0$ can have substantial absolute error and still give the correct threshold decision whenever
\[
\operatorname{sign}\{\widehat R_0(\boldsymbol\theta)-c\}
=
\operatorname{sign}\{R_0(\boldsymbol\theta)-c\}.
\]
Conversely, even a small magnitude error close to $c$ can flip the class. Thus the statistically relevant difficulty is concentrated near the decision boundary. This distinction is consistent with statistical learning theory for plug-in classifiers: classification risk depends on how accurately the decision boundary and its margin are recovered, rather than on uniformly accurate estimation of the entire regression function \citep{audibert2007}. REEBLRA exploits this lower-resolution target deliberately; it does {not} claim that classification universally dominates direct regression.

This design has three practical consequences. First, a nonlinear reproduction-number surface may still admit a simple, accurate local threshold boundary. Second, threshold classification is naturally aligned with invasion, persistence, intervention, and screening decisions. Third, information about magnitude is not discarded completely: the successive-threshold construction estimates ordered brackets such as $R_0\le1$, $1<R_0\le1.5$, and higher ranges. When exact numerical $R_0$ is required, direct regression or mechanistic calculation should be used alongside REEBLRA rather than replaced by it.

\subsection{Mechanistic--surrogate concordance principle}\label{subsec:concordance-principle}

Let $\boldsymbol\theta\in\Theta\subset(0,\infty)^p$ denote the original parameter vector (the positivity restriction is only needed when logarithmic derivatives are used). Let $R(\boldsymbol\theta)>0$ be a continuously differentiable mechanistic reproduction-number summary and define
\[
 g_c(\boldsymbol\theta)=\log R(\boldsymbol\theta)-\log c,
 \qquad c>0.
\]
For a differentiable declared feature map $\boldsymbol\phi:\Theta\to\mathbb R^q$, let $\boldsymbol\mu_x$ and $S_x=\operatorname{diag}(s_{x1},\ldots,s_{xq})$ be the feature means and positive standard deviations from the training set. The fitted score is
\[
 \eta_c(\boldsymbol\theta)=a_c+\mathbf b_c^{\mathsf T}S_x^{-1}
 [\boldsymbol\phi(\boldsymbol\theta)-\boldsymbol\mu_x].
\]
Write $J_{\phi}(\boldsymbol\theta)=D_{\boldsymbol\theta}\boldsymbol\phi(\boldsymbol\theta)$ for its $q\times p$ Jacobian. By the chain rule,
\[
 \nabla_{\boldsymbol\theta}\eta_c(\boldsymbol\theta)=
 J_{\phi}(\boldsymbol\theta)^{\mathsf T}S_x^{-1}\mathbf b_c.
\]

\begin{proposition}[Local mechanistic--surrogate concordance]\label{prop:concordance}
Suppose $R$ and $\boldsymbol\phi$ are continuously differentiable near $\boldsymbol\theta_*$, $g_c(\boldsymbol\theta_*)=\eta_c(\boldsymbol\theta_*)=0$, and both gradients are nonzero there. If the two regular level surfaces share the same tangent hyperplane at $\boldsymbol\theta_*$ and their normals point toward increasing threshold exceedance, then
\[
\frac{J_{\phi}(\boldsymbol\theta_*)^{\mathsf T}S_x^{-1}\mathbf b_c}
{\|J_{\phi}(\boldsymbol\theta_*)^{\mathsf T}S_x^{-1}\mathbf b_c\|}
=
\frac{\nabla_{\boldsymbol\theta}\log R(\boldsymbol\theta_*)}
{\|\nabla_{\boldsymbol\theta}\log R(\boldsymbol\theta_*)\|}.
\]
The equality is {conditional} on exact local tangency and matching orientation; its empirical analogue is only approximate when fitted and mechanistic boundaries differ. No coefficient-gradient equality is asserted for arbitrary transformed or interacting features.
\end{proposition}
\begin{proof}
At a regular level set, the gradient is normal to the tangent hyperplane. The two nonzero gradients must therefore be scalar multiples of each other, and the shared threshold orientation makes this scalar positive. Applying the chain rule to $\eta_c$ yields the stated equality after normalization. The argument is a level-set geometry identity, not a proof that a fitted classifier must recover the mechanistic normal.
\end{proof}

\begin{corollary}[Multiplicative effects for raw linear features]\label{cor:multiplicative-signs}
If $R(\boldsymbol\theta)=K\prod_{j=1}^p\theta_j^{\alpha_j}$ with $\theta_j>0$, then
\[
\frac{\partial\log R}{\partial\theta_j}=\frac{\alpha_j}{\theta_j},
\qquad
\frac{\partial\log R}{\partial\log\theta_j}=\alpha_j.
\]
For the raw linear feature map $\boldsymbol\phi(\boldsymbol\theta)=\boldsymbol\theta$, the proposition implies, under its tangency hypotheses, $\operatorname{sign}(b_{cj})=\operatorname{sign}(\alpha_j)$ for nonzero components and relative standardized coefficient magnitudes proportional to $|\alpha_j|s_{xj}/\theta_j^*$. For a nonlinear feature map, the same interpretation applies to the {complete score derivative}, not to a single main-effect coefficient. The signs and local magnitudes may differ across thresholds as the boundary points and fitting domains change.
\end{corollary}

In the numerical tables below, the reported $D_j$ and $G_j$ compare median {componentwise log-parameter derivatives} near the unit threshold. They provide a sign-concordance check, but are not a test of exact boundary tangency. A stronger {pointwise} orientation diagnostic at an evaluated parameter vector is
\[
\frac{\nabla_{\boldsymbol\theta}\eta_c(\boldsymbol\theta)^{\mathsf T}
\nabla_{\boldsymbol\theta}\log R(\boldsymbol\theta)}
{\|\nabla_{\boldsymbol\theta}\eta_c(\boldsymbol\theta)\|\,
\|\nabla_{\boldsymbol\theta}\log R(\boldsymbol\theta)\|},
\]
provided both gradients are nonzero. We can report its distribution at well-supported points near $R=c$ alongside boundary-distance error. In models without an analytic $R$ or a trusted numerical gradient, concordance can only be proposed as a hypothesis pending independent validation.

\section{REEBLRA workflow and Shiny application}\label{sec:app-workflow}

\subsection{Workflow overview}

REEBLRA is an interpretable Bayesian classification framework that learns reproduction-number relationships directly from simulated epidemic models and produces two complementary outputs. First, its binary model learns an explicit threshold equation that classifies parameter combinations according to whether $R_0>1$ or $R_0<1$, thereby identifying parameter combinations associated with epidemic growth or decline. Second, its multiclass bracket model learns a sequence of threshold equations that places $R_0$ into ordered numerical ranges, providing an interpretable approximation of the magnitude of $R_0$---that is, the expected number of secondary infections generated by a typical infectious individual in an otherwise fully susceptible population. Thus, REEBLRA can be used both to recover the epidemiologically important $R_0=1$ threshold and to approximate the numerical scale of $R_0$ through ordered reproduction-number brackets.

\begin{figure}[!htbp]
    \centering
    \maybeincludegraphics[width=0.95\linewidth]{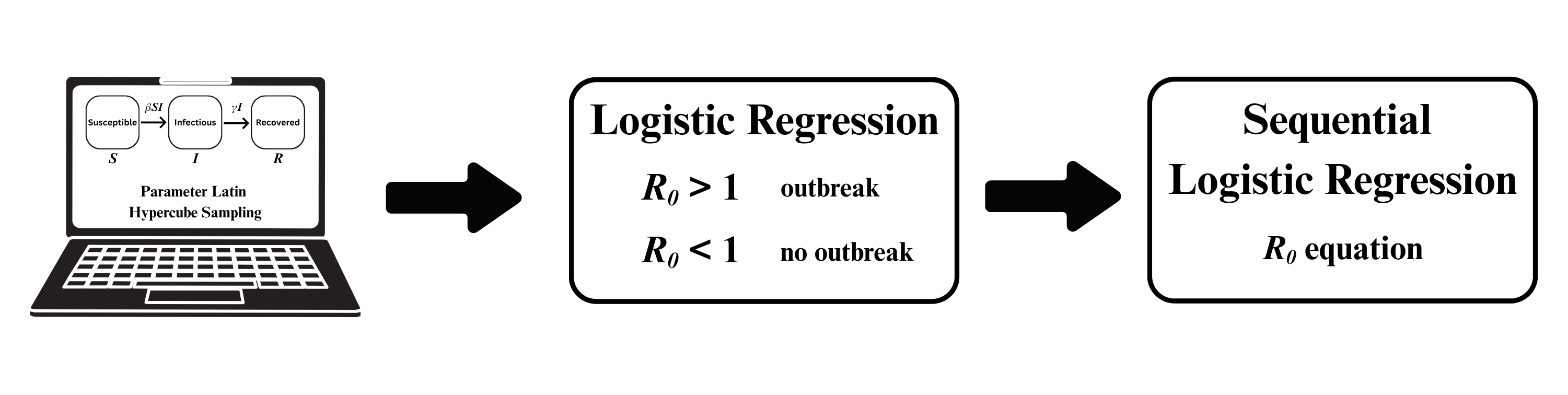}
    \caption{REEBLRA simulation-to-threshold workflow. The implementation supports built-in ODE simulation, custom ODE simulation with optional auxiliary equations, and direct upload of parameter--time series tables. These inputs feed binary Bayesian threshold analysis, successive ordered-bracket models, validation, collinearity checking, sensitivity summaries, and export modules. REEBLRA provides two complementary outputs: a {binary threshold equation} that distinguishes $R_0>1$ from $R_0<1$, and a {multiclass bracket model} that uses successive threshold equations to approximate the numerical magnitude of $R_0$, interpreted as the expected number of secondary infections generated by a typical infectious individual in an otherwise fully susceptible population.}
    \label{fig:flowchart}
\end{figure}

The app follows three operational stages (Figure~\ref{fig:flowchart}). First, the user selects or constructs the data source (Table~\ref{tab:data-modes}). Then reproduction-number target or threshold outcome is defined. Latin hypercube parameter sampling follows the space-filling design principle of McKay et al. \citep{mckay1979}; related simulation sensitivity workflows are reviewed by Marino et al. \citep{marino2008}. Second, the binary Bayesian model is fitted (Section~\ref{sec:binary}). Third, the sequential ordered multiclass bracket model is fitted (Section~\ref{sec:multiclass}). After the operational stages, equations, diagnostics, metrics, plots, trial predictions, and downloadable outputs are presented and can be evaluated. The learned equations are logit-polynomial surrogates in a declared feature basis (linear, log-linear, log-quadratic, or compact log-cubic); they approximate threshold surfaces over the sampled or uploaded domain rather than symbolically deriving a mechanistic $R_0$.

\begin{table}[!htbp]
\centering
\caption{REEBLRA Shiny data modes and their threshold summaries.}
\label{tab:data-modes}
\begin{tabularx}{\textwidth}{>{\bfseries\raggedright\arraybackslash}p{0.18\textwidth}>{\raggedright\arraybackslash}X>{\raggedright\arraybackslash}X}
\toprule
\tablehead \HeaderCell{Mode} & \HeaderCell{Input or simulation} & \HeaderCell{Available threshold summaries} \\
\midrule
Built-in ODE & Closed SIR, open SIR, open SEIR, vaccinated SEIR, Age-structured SIR (two groups), 15-compartment five-group SIR, and vector-borne models & Case increase; analytic $R_0$ or $R_e(0)$; early-growth numerical \texttt{R0estimate} \\
Custom ODE & User equations, optional auxiliary equations, initial values, parameter ranges, selected nonnegative states, and an early-growth target & Case increase and early-growth numerical \texttt{R0estimate}; analytic NGM $R_0$ is not derived automatically \\
Uploaded table & One \texttt{R0estimate} column and up to 50 selected numeric parameter columns & \texttt{R0estimate}>cutoff and ordered \texttt{R0estimate} brackets \\
CSV generator & Time in the first column, one trajectory per remaining column, and one matching parameter row per trajectory & Estimated growth rate, diagnostics, and generated \texttt{R0estimate} table for direct use in REEBLRA \\
\bottomrule
\end{tabularx}
\end{table}

Algorithm~\ref{alg:reeblr} summarizes the basic workflow of REEBLRA. Readers interested in exploring the framework further can use the accompanying interactive Shiny app. In REEBLRA, `trial prediction' refers to the interactive what-if prediction feature of the Shiny app. It allows users to enter a new combination of epidemiological parameter values and evaluate the previously fitted REEBLRA equations.

\begin{algorithm}[!htbp]
\caption{REEBLRA simulation-to-threshold workflow}\label{alg:reeblr}
\begin{algorithmic}[1]
\State Choose built-in ODE, custom ODE, uploaded table, or CSV-generator workflow.
\If{built-in ODE simulation is selected}
  \State Set parameter ranges, initial conditions, time grid, susceptible threshold, and case-increase tolerance.
  \State Generate $n$ parameter sets using LHS and simulate each system with \texttt{deSolve::ode} using \texttt{lsoda}.
  \State Compute analytic $R_0$ or $R_e(0)$ where implemented and estimate early-growth \texttt{R0estimate}.
\ElsIf{custom ODE simulation is selected}
  \State Parse state equations, optional auxiliary equations, initial values, parameter ranges, selected nonnegative states, and target output.
  \State Generate LHS parameter sets, simulate each custom system, fit early growth, and convert $r$ to \texttt{R0estimate}.
\Else
  \State Read or generate a complete parameter--\texttt{R0estimate} table.
\EndIf
\State Define the binary threshold event and/or ordered $R_0$ bracket source.
\State Choose a feature basis before validation; log-based bases require positive predictors.
\State Construct feature terms within each training fold and standardize them using training-data means and standard deviations.
\State Fit the binary Bayesian logistic model; also fit an ordinary logistic model on the same full-data formula, for comparison.
\State Fit successive ordinary binary logistic models for the ordered $R_0$ thresholds.
\State Compute apparent and optional held-out $k$-fold predictions and metrics.
\State Report equations, coefficients, trial predictions, ROC/AUC, class distribution, coefficient heatmap, collinearity, sensitivity summaries, warnings, and exports.
\end{algorithmic}
\end{algorithm}

\section{Binary Bayesian threshold REEBLRA model}\label{sec:binary}

For a selected cutoff $c$, let $y_i(c)=1$ when the selected threshold event is true and $0$ otherwise. REEBLRA first constructs a declared feature map
\[
\boldsymbol{x}_i=\boldsymbol{\phi}(\boldsymbol{\theta}_i),
\]
where $\boldsymbol{\theta}$ contains the raw model parameters. The available bases are linear raw parameters, log-linear, log-quadratic, compact log-cubic, and an automatic model-aware choice. Log-based bases require positive predictors. The fitted feature terms are standardized within the training data,
\[
z_{ij}=\frac{x_{ij}-\bar x_j}{s_j},
\]
and REEBLRA fits
\[
\Pr\{y_i(c)=1\mid\boldsymbol z_i\}
=\operatorname{logit}^{-1}\!\left(\alpha+\sum_{j=1}^{q}b_jz_{ij}\right).
\]
The standardized score and probability are
\[
\eta_z=\alpha+\sum_{j=1}^{q}b_jz_j,
\qquad
p=\frac{1}{1+\exp(-\eta_z)}.
\]
At the default decision cutoff, $p>0.5$ is equivalent to $\eta_z>0$.

Undoing standardization gives the unstandardized fitted-feature score
\[
\eta_x
=\left(\alpha-\sum_{j=1}^{q}\frac{b_j\bar x_j}{s_j}\right)
+\sum_{j=1}^{q}\frac{b_j}{s_j}x_j.
\]
For the linear basis, $x_j=\theta_j$ and this is directly the original-parameter-scale equation. For transformed bases, $x_j$ remains a log, power, or interaction feature. In the Shiny app, trial prediction nevertheless accepts raw ODE or uploaded parameter values because the app applies the declared feature map and training scaling automatically before evaluating the fitted score.

The Bayesian model is implemented using \texttt{rstanarm::stan\_glm}. The supplied app uses autoscaled normal priors with location 0 and scale 2 for both coefficients and the intercept. Users may choose MCMC sampling or mean-field variational Bayes. MCMC supports convergence diagnostics; variational Bayes is faster but approximate and does not provide MCMC $\widehat R_0$ or effective-sample-size diagnostics. The reported equation is fitted to all complete observations. For $k\ge2$, each validation fold reconstructs the selected feature map and standardization from the training data and then applies those training transformations to its held-out observations.

In the Shiny app, when all simulations fall on one side of the selected threshold, the app stops fitting and reports the observed counts, directing the user to inspect threshold positives and negatives, widen parameter ranges, increase simulation size, alter the cutoff, or change the threshold source. Algorithm~\ref{alg:binary-reeblr} shows the flow of the binary model.

\begin{algorithm}[!htbp]
\caption{Binary REEBLRA fitting for a selected threshold}\label{alg:binary-reeblr}
\begin{algorithmic}[1]
\Require Complete parameter table, binary threshold target, cutoff, Bayesian settings, and validation-fold count.
\State Verify number of complete rows and both target classes are present.
\State Construct the declared predictor feature map, standardize its terms, and fit the all-data Bayesian logistic model.
\State Undo standardization to report the unstandardized fitted-feature equation; this is an original-parameter equation only for a linear basis.
\State Fit the ordinary logistic comparator using the same full-data formula.
\State If $k\ge2$, reconstruct the feature map and standardization within each training fold, refit the Bayesian model, and predict its held-out fold.
\State Compute binary metrics, ROC data, posterior coefficient summaries, MCMC diagnostics when available, and coefficient-based sensitivity summaries.
\Ensure Bayesian and ordinary equations, trial prediction, validation metrics, plots, warnings, and exports.
\end{algorithmic}
\end{algorithm}

\subsection*{Ordinary logistic comparison}

The binary output in the Shiny app also displays an ordinary, non-Bayesian logistic regression fitted to the same complete rows, standardized fitted features, and full-data formula. It reports
\[
\eta^{\mathrm{GLM}}=\alpha^{\mathrm{GLM}}+\sum_j b_j^{\mathrm{GLM}}z_j,
\qquad
p^{\mathrm{GLM}}=\{1+\exp(-\eta^{\mathrm{GLM}})\}^{-1},
\]
on the standardized scale and after undoing standardization. Under a linear basis the latter is the original-parameter scale; under transformed bases it remains on the fitted-feature scale. This comparator helps users see whether posterior-mean Bayesian equations and maximum-likelihood equations have similar direction and scale. It is a descriptive comparison, not an independent validation set and not a replacement for posterior diagnostics or cross-validation.

\section{Sequential ordered multiclass bracket REEBLRA model}\label{sec:multiclass}

The multiclass module estimates ordered brackets rather than unrelated nominal categories. In the Shiny app, for step size $h\in\{0.25,0.5,1,2,5\}$, thresholds are generated from 1 through 10 and brackets take the form
\[
<1,\quad 1-(1+h),\quad (1+h)-(1+2h),\quad\ldots,\quad >10.
\]
At each threshold $c_k$, a logistic model of the declared type estimates
\[
q_k(\boldsymbol z)
=\Pr(R_0>c_k\mid\boldsymbol z)
=\operatorname{logit}^{-1}\!\left(\alpha_k+\sum_{j=1}^{q}b_{kj}z_j\right),
\]
using the same declared feature basis and training-fold standardization as the binary module. Note that users can also explore individual cutoffs through the Bayesian binary module, but this does not constitute a jointly fitted Bayesian multiclass model.

If all training observations lie on one side of a threshold, a finite logistic equation is not estimable. The app records the corresponding constant exceedance probability as 0 or 1 and reports the counts above and at-or-below that threshold. Predicted exceedance probabilities are post-processed to be non-increasing:
\[
q_{k+1}\leftarrow\min(q_k,q_{k+1}).
\]
Bracket probabilities are then obtained by subtraction,
\begin{align*}
\Pr(R_0\le c_1)&=1-q_1,\\
\Pr(c_{k-1}<R_0\le c_k)&=q_{k-1}-q_k,\\
\Pr(R_0>c_K)&=q_K,
\end{align*}
negative numerical artifacts are set to zero, and each probability row is renormalized. The predicted bracket is the class with the largest bracket probability.

The operation above projects {predicted cumulative probabilities} into a non-increasing sequence for each parameter vector; it does not constrain the independently fitted coefficient vectors or guarantee that their raw logistic threshold surfaces are nested everywhere. A diagnostic should count pre-correction violations, for example
\[
\sum_{k=1}^{K-1}
\max\{0,q_{k+1}^{\rm raw}(\boldsymbol\theta)-q_k^{\rm raw}(\boldsymbol\theta)\}.
\]
We can implement a comparison with a jointly fitted ordinal model \citep{mccullagh1980,peterson1990}.

Algorithm~\ref{alg:multi-reeblr} shows the flow of the multiclass model. 

\subsection*{Coefficient matrix across thresholds}

For fitted threshold equations, the Shiny app forms a coefficient matrix and heatmap whose rows are successive $R_0$ thresholds and whose columns are model parameters. To visualize relative importance within threshold $k$, coefficients are normalized separately by row:
\[
h_{kj}=\frac{b_{kj}}{\max_{\ell}|b_{k\ell}|},
\qquad -1\le h_{kj}\le1.
\]
The heatmap therefore shows the direction and relative magnitude of parameters within each threshold equation. Because every row has its own denominator, color intensity should not be interpreted as an absolute coefficient comparison between different $R_0$ levels. 

\begin{algorithm}[!htbp]
\caption{Sequential ordered multiclass bracket REEBLRA fitting}\label{alg:multi-reeblr}
\begin{algorithmic}[1]
\Require Complete parameter table, numeric $R_0$-like target, step size $h\in\{0.25,0.5,1,2,5\}$, and validation-fold count.
\State Create thresholds from 1 through 10 and their associated ordered brackets.
\State Construct the declared predictor feature map and standardize its fitted terms.
\For{each threshold $c_k$}
  \If{observations occur on both sides of $c_k$}
    \State Fit ordinary logistic regression for $\Pr(R_0>c_k)$.
  \Else
    \State Record the constant empirical exceedance probability 0 or 1 and mark the finite equation as not estimable.
  \EndIf
\EndFor
\State Enforce non-increasing cumulative exceedance probabilities.
\State Difference and renormalize them to obtain bracket probabilities.
\State Predict the maximum-probability bracket and repeat feature construction, standardization, threshold fitting, and prediction independently within each validation fold when requested.
\State Compute overall, macro, weighted, per-class, and one-vs-rest metrics.
\State Build the raw and row-normalized coefficient matrices across thresholds.
\Ensure Ordered equations, predicted brackets, coefficient heatmap, class distribution, AUC summaries, confusion matrix, trial prediction, and exports.
\end{algorithmic}
\end{algorithm}

\section{Model evaluation metrics}\label{sec:evaluation}

For binary evaluation with 0.5 classification cutoff, the app reports accuracy, balanced accuracy, AUC, sensitivity (recall), specificity, precision, F1 score, negative predictive value (NPV), Cohen's $\kappa$, and Brier score (reported as the mean squared error of the predicted probability). ROC AUC measures ranking discrimination and is threshold independent, whereas F1 and balanced accuracy are sensitive to the decision rule and class imbalance \citep{fawcett2006,saito2015,steyerberg2010}. The Brier score retains probability information and therefore complements classification-only summaries. Apparent metrics are calculated from the final all-data model. When $k\ge2$, held-out predictions from stratified folds are combined to calculate cross-validation metrics; model interpretation should emphasize the held-out estimates when enough observations occur in both classes.

The ordered-bracket output reports overall accuracy, overall balanced accuracy (macro recall), and Cohen's $\kappa$. For each bracket, it reports sensitivity/recall, specificity, precision, F1 score, NPV, one-vs-rest balanced accuracy, one-vs-rest AUC, and observed support, together with macro and support-weighted summaries. One-vs-rest AUC is appropriate as a class-specific discrimination diagnostic, while multiclass summaries should be interpreted with the observed bracket distribution and confusion matrix \citep{hand2001}. Sparse brackets can yield unstable coefficients or apparently strong aggregate scores; therefore class support, per-class F1/recall, and macro summaries are retained rather than relying on weighted averages alone. Calibration and discrimination are distinct properties of predictive models, so high AUC should not be read as evidence that a fixed probability cutoff is well calibrated for every class distribution \citep{vancalster2019}.

The Shiny app also reports variance inflation factors (VIFs), a parameter correlation matrix, and a correlation heatmap. Collinearity may leave prediction useful while making individual coefficient signs and magnitudes unstable, a well-known distinction between predictive adequacy and coefficient interpretability \citep{dormann2013}. 

Figure~\ref{fig:cross-model-metrics} presents selected metrics for the ODE model demonstrations.

\section{ODE model demonstrations}\label{sec:models}

The REEBLRA application includes seven built-in ODE models, including a 15-compartment five-group SIR benchmark. The results in this section are {analytic-label parameter-space classification benchmarks}: their labels are generated by evaluating each model's known analytic $R_0$ or $R_e(0)$ expression. The examples verify recovery of known threshold surfaces on specified sampling domains. Note that the Shiny app also supports simulation-derived targets.

REEBLRA allows the analyst to choose and declare the predictor feature transformation before fitting. For example, original parameter values can enter linearly after standardization, or positive parameters can be represented by a log-linear, log-quadratic, or compact log-cubic basis with powers and interactions. The choice affects the complexity, interpretability, and predictive performance of the learned threshold surfaces. The main-text ODE results below use the stated model-aware binary bases and a compact log-cubic basis for sequential brackets with increment $h=5$. To illustrate alternative settings across {all seven} ODE examples, Supplementary Information Section~\ref{app:alternate-equations} additionally gives (i) the complete binary and sequential REEBLRA equations and validation metrics for standardized {linear} original parameters, retaining $h=5$ (Section~\ref{app:linear_h5}); and (ii) the complete sequential equations and metrics for the same compact log-cubic transformations as the main text, but with a finer $h=0.5$ (Section~\ref{app:log_cubic_h0p5}). The finer $h$ creates more classes and a stricter exact-bracket criterion, so its classification accuracy should not be compared directly with the $h=5$ accuracy without accounting for bracket resolution.

For the simple linear-basis SIR examples, an individual binary coefficient can be read directly as the fitted threshold direction of its raw parameter, holding the others fixed. For transformed bases, however, a parameter can enter through a main term, powers, and interactions. We therefore report both the main fitted coefficient and the local full-score quantity
\[
D_j=\operatorname{median}_{0.8\le R_0\le1.2}
\left\{\frac{\partial\eta}{\partial\log\theta_j}\right\},
\]
computed from the complete binary REEBLRA score near the unit threshold. $D_j$ includes all interaction and curvature terms and is used as an interpretability diagnostic; it is not a causal effect. This distinction is important whenever one standalone polynomial coefficient does not represent the net effect of the original epidemiological parameter. For the restored concordance tables, we also report the corresponding analytic direction
\[
G_j=\operatorname{median}_{0.8\le R_0\le1.2}
\left\{\frac{\partial\log R_0}{\partial\log\theta_j}\right\},
\]
so that the analytic numerator/denominator or next-generation-matrix role can be compared directly with the sign of the learned REEBLRA score.

For each ODE demonstration, five independent parameter sets are also evaluated explicitly. The tables compare the analytic reproduction number with the binary REEBLRA unit-threshold classification and the ordered multiclass REEBLRA bracket, without inserting the analytic formula into the feature basis. In the equations below, equations with more than five feature terms are displayed using the five largest absolute coefficients followed by an ellipsis. Complete equations and metrics can be generated from the supplementary Python code or from the Shiny app.

\noindent\textbf{Note}: The binary and first ordered-threshold equations both estimate the event \(R_0>1\), but are independently fitted using different Latin-hypercube samples and, where specified, different predictor feature bases. Consequently, their coefficients are not expected to coincide but both recover closely agreeing invasion boundaries. If both use exactly the same observations, predictor design, feature transformations, standardization, estimator, penalty/prior treatment, and convergence tolerances, their fitted equations should agree up to numerical tolerance.

\begin{figure}[!htbp]
\centering
\includegraphics[width=0.98\textwidth]{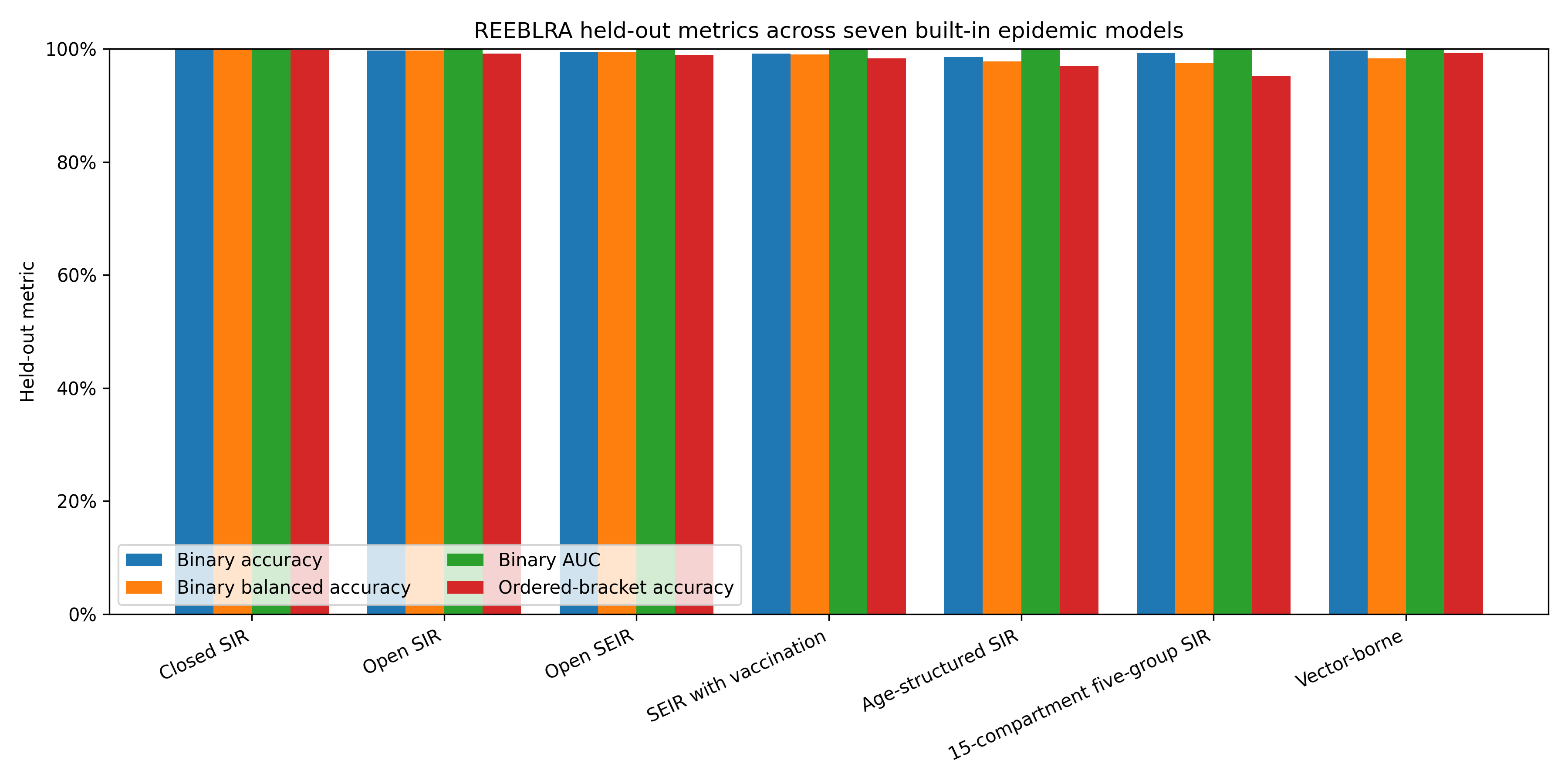}
\caption{Selected REEBLRA metrics across seven {analytic-label} epidemic-model benchmarks. For the actual values, refer to Table~\ref{tab:cross-model-v4}. These are held-out performance estimates under the sampled analytic-label domains, not validation on real-world or growth-derived threshold labels.}
\label{fig:cross-model-metrics}
\end{figure}

\subsection{Closed SIR model}\label{subsec:closed-sir}

\begin{figure}[!htbp]
\centering
\begin{tikzpicture}[node distance=25mm]
\node[compartment] (S) {$S$};
\node[compartment,right=of S] (I) {$I$};
\node[compartment,right=of I] (R) {$R$};
\draw[flow] (S) -- node[flowlabel,above] {$\beta SI/N$} (I);
\draw[flow] (I) -- node[flowlabel,above] {$\gamma I$} (R);
\end{tikzpicture}
\caption{Closed SIR compartmental structure used by the built-in simulator. The cumulative-incidence state $\mathcal{C}$ integrates the infection flow $\beta SI/N$.}
\label{fig:SIR}
\end{figure}
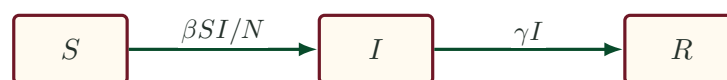

The closed SIR model (Figure~\ref{fig:SIR}) is
\begin{align}
\frac{dS}{dt} &= -\beta\frac{SI}{N},\\
\frac{dI}{dt} &= \beta\frac{SI}{N}-\gamma I,\\
\frac{dR}{dt} &= \gamma I,\\
\frac{dC}{dt} &= \beta\frac{SI}{N},
\end{align}
where $N=S+I+R$ and $\mathcal{C}$ is cumulative incidence. The analytic reproduction number is
\[
R_0=\frac{\beta}{\gamma}.
\]
For $R_0>1$, the analytic threshold is $\beta-\gamma>0$. The stronger comparison follows from $\log R_0=\log\beta-\log\gamma$: the standardized mechanistic gradient is positive in $\beta$ and negative in $\gamma$, with equal unit log-elasticities before sampling-scale adjustment.

\paragraph{REEBLRA results.}
The binary demonstration used the declared model-aware basis and five-fold held-out validation. The unstandardized score below is the posterior-mode (MAP) estimate at the center of the Gaussian Laplace approximation. 
\[\resizebox{0.98\linewidth}{!}{$\displaystyle \eta_{R_0>1}=-0.0144 - 60.1290\,\gamma + 60.0748\,\beta.$}\]

The ordered-bracket demonstration used the generic compact log-cubic basis with step $h=5$, five-fold held-out validation, and independent $N(0,2^2)$ prior and MAP/Laplace fitting routine. The displayed scores are MAP plug-in logit equations, not integrated posterior-predictive probabilities. Representative expressions were:

\begin{center}
\resizebox{0.98\linewidth}{!}{$\displaystyle \eta_{R_0>1} = 0.1414 + 27.6737\,\log(\beta) - 27.1390\,\log(\gamma) + 3.8169\,\log(\gamma)^2 - 3.5491\,\log(\beta)^2 + 0.7792\,\log(\beta)^3 + \cdots$}\\[2pt]
\resizebox{0.98\linewidth}{!}{$\displaystyle \eta_{R_0>6} = -8.4491 - 9.4359\,\log(\beta)\,\log(\gamma) + 2.4399\,\log(\gamma) + 1.7741\,\log(\gamma)^2 - 1.2630\,\log(\gamma)^3 - 0.1887\,\log(\beta) + \cdots$}\\[2pt]
\resizebox{0.98\linewidth}{!}{$\displaystyle \eta_{R_0>10} = -7.0246 - 6.1480\,\log(\beta)\,\log(\gamma) + 2.8277\,\log(\gamma) - 1.0110\,\log(\gamma)^3 + 0.5673\,\log(\beta) + 0.3047\,\log(\beta)^2 + \cdots$}
\end{center}

\begin{table}[ht]
\centering
\caption{Mechanistic--REEBLRA coefficient concordance for the Closed SIR demonstration. The analytic role identifies whether each parameter enters the $R_0$ expression through a numerator, denominator, both, or the next-generation matrix (NGM). $G_j$ is the median analytic log-gradient $\partial\log R_0/\partial\log\theta_j$ and $D_j$ is the median derivative of the complete REEBLRA binary score over $0.8\le R_0\le1.2$. For this linear basis, $b_j$ is the unstandardized raw-parameter coefficient. Concordance is assessed from the sign of the full-score direction $D_j$ relative to $G_j$.}
\label{tab:closed-sir-concordance}
\scriptsize
\setlength{\tabcolsep}{3pt}
\begin{tabularx}{\linewidth}{l >{\raggedright\arraybackslash}X r r r c}
\toprule
Parameter & Analytic $R_0$ role & Analytic $G_j$ & REEBLRA $b_j$ & Full-score $D_j$ & Concordant? \\
\midrule
$\beta$ & Numerator & +1.000 & +60.075 & +39.950 & Yes \\
$\gamma$ & Denominator & -1.000 & -60.129 & -40.864 & Yes \\
\bottomrule
\end{tabularx}
\end{table}

\begin{figure}[!htbp]
\centering
\includegraphics[width=1\textwidth]{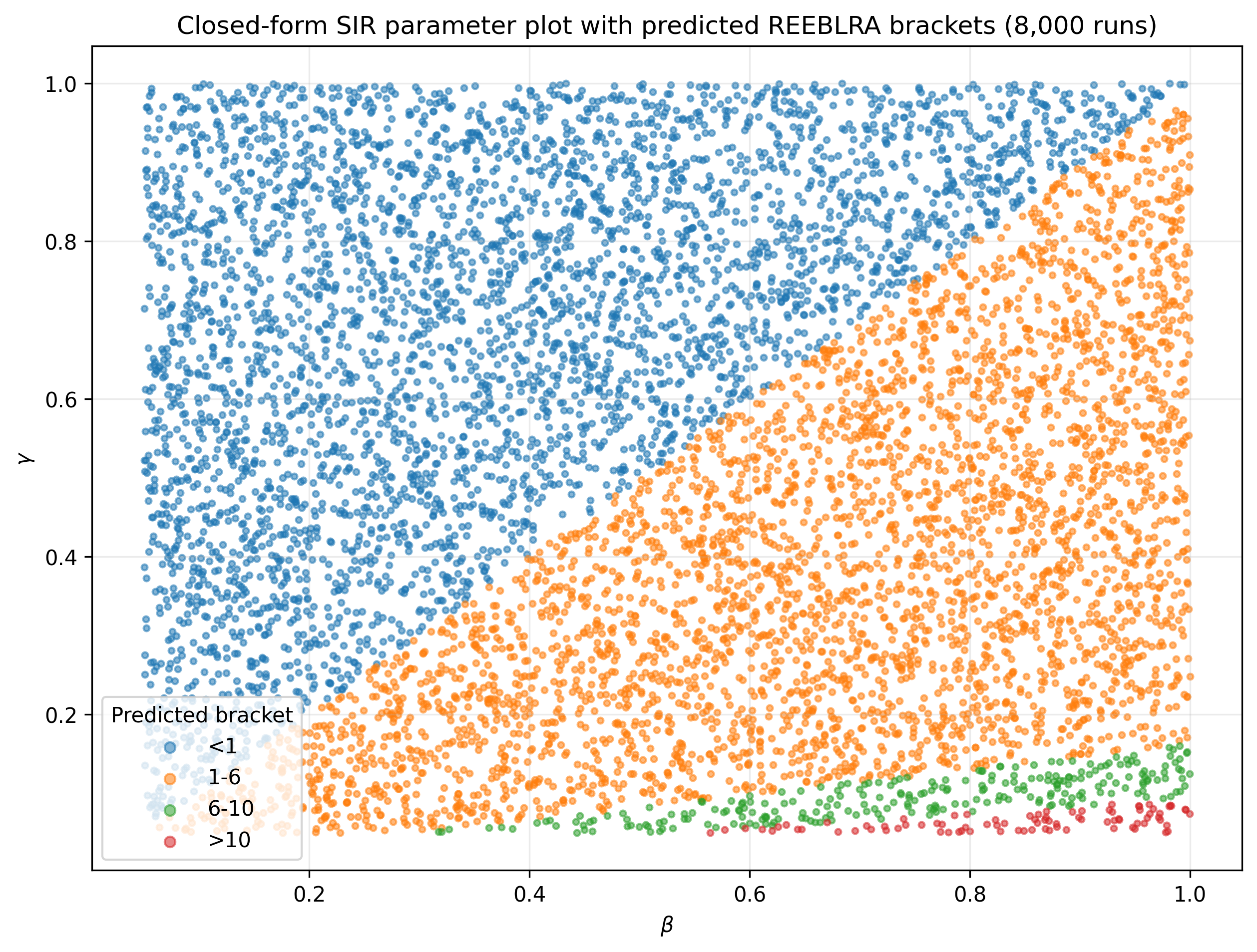}
\caption{Parameter-space plot for the closed-form SIR demonstration. Among the 8,000 parameter combinations, REEBLRA classified 3,987 (49.84\%) as $R_0<1$, 3,675 (45.94\%) as $1\leq R_0\leq 6$, 252 (3.15\%) as $6<R_0\leq 10$, and 86 (1.08\%) as $R_0>10$.}
\label{fig:param-space}
\end{figure}

\begin{table}[ht]
\centering
\caption{Five independent parameter examples for the Closed SIR demonstration. The analytic $R_0$ is compared with the binary REEBLRA prediction and the ordered-bracket REEBLRA prediction. The binary entry reports the predicted side of the unit threshold with $P(R_0>1)$ in parentheses; the ordered entry reports the modal bracket with its normalized class probability. Cases are selected from an independent deterministic Latin-hypercube design (seed 2600) by proximity to prespecified analytic targets $0.7,2,5,8,12$; prediction results were not used for case selection.}
\label{tab:closed-sir-five-examples}
\scriptsize
\setlength{\tabcolsep}{3pt}
\begin{tabularx}{\linewidth}{c >{\raggedright\arraybackslash}X c c c}
\toprule
Case & Parameter values & Analytic $R_0$ (bracket) & Binary REEBLRA & Ordered REEBLRA \\
\midrule
E1 & $\beta=0.6669,\allowbreak\,\gamma=0.9526$ & 0.700 ($<1$) & $\leq1$ ($P_{>1}=0.000$) & $<1$ ($p=1.000$) \\
E2 & $\beta=0.9858,\allowbreak\,\gamma=0.4928$ & 2.000 ($1\text{--}6$) & $>1$ ($P_{>1}=1.000$) & $1\text{--}6$ ($p=1.000$) \\
E3 & $\beta=0.6791,\allowbreak\,\gamma=0.1358$ & 5.000 ($1\text{--}6$) & $>1$ ($P_{>1}=1.000$) & $1\text{--}6$ ($p=0.967$) \\
E4 & $\beta=0.6308,\allowbreak\,\gamma=0.07904$ & 7.980 ($6\text{--}10$) & $>1$ ($P_{>1}=1.000$) & $6\text{--}10$ ($p=0.972$) \\
E5 & $\beta=0.8246,\allowbreak\,\gamma=0.0686$ & 12.021 ($>10$) & $>1$ ($P_{>1}=1.000$) & $>10$ ($p=0.960$) \\
\bottomrule
\end{tabularx}
\end{table}

This first example is a deliberately transparent sanity check. Without giving REEBLRA the analytic formula $R_0=\beta/\gamma$, the binary equation recovers a positive transmission coefficient and a negative removal coefficient (Table~\ref{tab:closed-sir-concordance}). More importantly, the fitted raw-scale magnitudes are almost identical, $|b_\beta|=60.075$ and $|b_\gamma|=60.129$. That near symmetry mirrors the exact threshold $\beta=\gamma$ and the equal-and-opposite unit log-elasticities of $\log R_0$. Thus REEBLRA does more than classify this toy model correctly: its coefficient geometry reproduces the mechanistic balance that defines the invasion boundary. For users, this provides an immediate diagnostic benchmark--if a fitted surrogate for a known SIR system did not recover these opposite, nearly equal directions over a balanced parameter domain, the sampling design, feature basis, or fit should be questioned before applying the method to a more complex model. The ordered-threshold equations then show how the same two parameters continue to move the system across progressively higher reproduction-number cutoffs (Figure~\ref{fig:param-space}).

Table~\ref{tab:closed-sir-five-examples} shows that REEBLRA assigns these five examples to the correct {analytic-$R_0$ brackets}; it does not recover their exact numerical values. A smaller bracket width $h$ can increase nominal resolution, but is not guaranteed to improve predictive accuracy because narrow brackets may have limited training support. 

\subsection{Open SIR model}

The open SIR model with demographic turnover is
\begin{align}
\frac{dS}{dt} &= \mu N-\beta\frac{SI}{N}-\mu S,\\
\frac{dI}{dt} &= \beta\frac{SI}{N}-\gamma I-\mu I,\\
\frac{dR}{dt} &= \gamma I-\mu R,\\
\frac{dC}{dt} &= \beta\frac{SI}{N}.
\end{align}
The analytic reproduction number is
\[
R_0=\frac{\beta}{\gamma+\mu}.
\]
The mechanistic log-threshold is $\log R_0=\log\beta-\log(\gamma+\mu)$. Hence $\partial\log R_0/\partial\beta>0$ while the derivatives with respect to both $\gamma$ and $\mu$ are negative.

\paragraph{REEBLRA results.}
The binary demonstration used the declared model-aware basis and five-fold held-out validation. The unstandardized score below is the posterior-mode (MAP) estimate at the center of the Gaussian Laplace approximation. 
\[\resizebox{0.98\linewidth}{!}{$\displaystyle \eta_{R_0>1}=-0.3650 + 54.7156\,\beta - 54.6899\,\gamma - 51.8621\,\mu.$}\]

The ordered-bracket demonstration used the generic compact log-cubic basis with step $h=5$, five-fold held-out validation, and independent $N(0,2^2)$ prior and MAP/Laplace fitting routine. The displayed scores are MAP plug-in logit equations, not integrated posterior-predictive probabilities. Representative expressions were:

\begin{center}
\resizebox{0.98\linewidth}{!}{$\displaystyle \eta_{R_0>1} = -9.4498 + 19.9294\,\log(\beta) + 7.8800\,\log(\gamma)\,\log(\mu) - 6.7794\,\log(\gamma) - 4.2481\,\log(\mu) - 3.5240\,\log(\beta)\,\log(\mu) + \cdots$}\\[2pt]
\resizebox{0.98\linewidth}{!}{$\displaystyle \eta_{R_0>6} = -14.4591 + 3.8716\,\log(\gamma)\,\log(\mu) - 3.4726\,\log(\beta)\,\log(\gamma) + 3.2166\,\log(\gamma) + 1.6060\,\log(\beta)^2 - 1.5573\,\log(\beta)\,\log(\mu) + \cdots$}\\[2pt]
\resizebox{0.98\linewidth}{!}{$\displaystyle \eta_{R_0>10} = -8.5646 + 4.2431\,\log(\gamma) - 2.5256\,\log(\beta)\,\log(\gamma) + 2.3848\,\log(\gamma)\,\log(\mu) - 1.4532\,\log(\beta) - 1.0112\,\log(\beta)\,\log(\mu) + \cdots$}
\end{center}

\begin{table}[ht]
\centering
\caption{Mechanistic--REEBLRA coefficient concordance for the Open SIR demonstration. The analytic role identifies whether each parameter enters the $R_0$ expression through a numerator, denominator, both, or the next-generation matrix (NGM). $G_j$ is the median analytic log-gradient $\partial\log R_0/\partial\log\theta_j$ and $D_j$ is the median derivative of the complete REEBLRA binary score over $0.8\le R_0\le1.2$. For this linear basis, $b_j$ is the unstandardized raw-parameter coefficient. Concordance is assessed from the sign of the full-score direction $D_j$ relative to $G_j$.}
\label{tab:open-sir-concordance}
\scriptsize
\setlength{\tabcolsep}{3pt}
\begin{tabularx}{\linewidth}{l >{\raggedright\arraybackslash}X r r r c}
\toprule
Parameter & Analytic $R_0$ role & Analytic $G_j$ & REEBLRA $b_j$ & Full-score $D_j$ & Concordant? \\
\midrule
$\beta$ & Numerator & +1.000 & +54.716 & +28.548 & Yes \\
$\gamma$ & Denominator: $(\gamma+\mu)$ & -0.723 & -54.690 & -21.794 & Yes \\
$\mu$ & Denominator: $(\gamma+\mu)$ & -0.277 & -51.862 & -7.761 & Yes \\
\bottomrule
\end{tabularx}
\end{table}

\begin{table}[ht]
\centering
\caption{Five independent parameter examples for the Open SIR demonstration. The analytic $R_0$ is compared with the binary REEBLRA prediction and the ordered-bracket REEBLRA prediction. The binary entry reports the predicted side of the unit threshold with $P(R_0>1)$ in parentheses; the ordered entry reports the modal bracket with its normalized class probability. Cases are selected from an independent deterministic Latin-hypercube design (seed 2601) by proximity to prespecified analytic targets $0.7,2,5,8,12$; prediction results were not used for case selection.}
\label{tab:open-sir-five-examples}
\scriptsize
\setlength{\tabcolsep}{3pt}
\begin{tabularx}{\linewidth}{c >{\raggedright\arraybackslash}X c c c}
\toprule
Case & Parameter values & Analytic $R_0$ (bracket) & Binary REEBLRA & Ordered REEBLRA \\
\midrule
E1 & $\beta=0.3592,\allowbreak\,\gamma=0.4541,\allowbreak\,\mu=0.05893$ & 0.700 ($<1$) & $\leq1$ ($P_{>1}=0.000$) & $<1$ ($p=1.000$) \\
E2 & $\beta=0.2401,\allowbreak\,\gamma=0.1142,\allowbreak\,\mu=0.005828$ & 2.000 ($1\text{--}6$) & $>1$ ($P_{>1}=0.998$) & $1\text{--}6$ ($p=1.000$) \\
E3 & $\beta=1.1,\allowbreak\,\gamma=0.1372,\allowbreak\,\mu=0.08269$ & 5.002 ($1\text{--}6$) & $>1$ ($P_{>1}=1.000$) & $1\text{--}6$ ($p=0.954$) \\
E4 & $\beta=0.9599,\allowbreak\,\gamma=0.08933,\allowbreak\,\mu=0.03058$ & 8.005 ($6\text{--}10$) & $>1$ ($P_{>1}=1.000$) & $6\text{--}10$ ($p=0.895$) \\
E5 & $\beta=1.182,\allowbreak\,\gamma=0.05958,\allowbreak\,\mu=0.03929$ & 11.958 ($>10$) & $>1$ ($P_{>1}=1.000$) & $>10$ ($p=0.692$) \\
\bottomrule
\end{tabularx}
\end{table}

The open-SIR result shows how REEBLRA can expose a mechanistic grouping that is not visible from prediction accuracy alone. Transmission remains strongly positive ($b_\beta=54.716$), while recovery and demographic turnover are both negative ($b_\gamma=-54.690$ and $b_\mu=-51.862$), consistent with their additive appearance in the same denominator $\gamma+\mu$ (Table~\ref{tab:open-sir-concordance}). The similar negative magnitudes indicate that, over the sampled threshold region, increasing either loss process moves the system away from invasion in a comparably strong direction. This does not mean that $\gamma$ and $\mu$ are biologically interchangeable; rather, it shows that the learned threshold surface has detected their shared mathematical role in removing infectious individuals. For model users, such coefficient patterns can reveal which parameters operate as a common threshold mechanism and can suggest useful reparameterizations such as an aggregate removal rate when interpretation or identifiability is difficult.

Table~\ref{tab:open-sir-five-examples} {illustrates recovery of $R_0$ brackets} rather than pointwise numerical values.

\subsection{Open SEIR model}\label{subsec:open-seir}

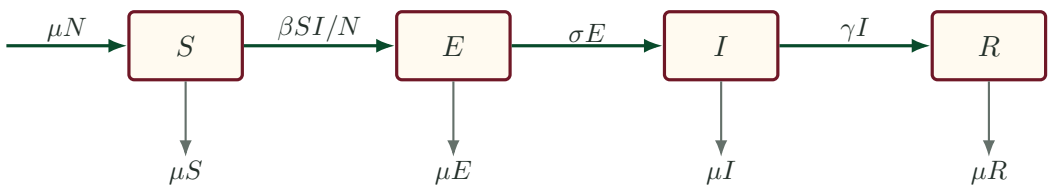
\begin{figure}[!htbp]
\centering
\begin{tikzpicture}[node distance=20mm]
\node[compartment] (S) {$S$};
\node[compartment,right=of S] (E) {$E$};
\node[compartment,right=of E] (I) {$I$};
\node[compartment,right=of I] (R) {$R$};
\coordinate[left=16mm of S] (birth);
\draw[flow] (birth) -- node[flowlabel,above] {$\mu N$} (S);
\draw[flow] (S) -- node[flowlabel,above] {$\beta SI/N$} (E);
\draw[flow] (E) -- node[flowlabel,above] {$\sigma E$} (I);
\draw[flow] (I) -- node[flowlabel,above] {$\gamma I$} (R);
\foreach \X/\lab in {S/\mu S,E/\mu E,I/\mu I,R/\mu R}{
  \draw[loss] (\X.south) -- ++(0,-10mm) node[flowlabel,below] {$\lab$};
}
\end{tikzpicture}
\caption{Open SEIR compartmental structure with demographic turnover. The cumulative-incidence state integrates $\beta SI/N$.}
\label{fig:SEIR}
\end{figure}

The open SEIR model (Figure~\ref{fig:SEIR}) is
\begin{align}
\frac{dS}{dt} &= \mu N-\beta\frac{SI}{N}-\mu S,\\
\frac{dE}{dt} &= \beta\frac{SI}{N}-\sigma E-\mu E,\\
\frac{dI}{dt} &= \sigma E-\gamma I-\mu I,\\
\frac{dR}{dt} &= \gamma I-\mu R,\\
\frac{dC}{dt} &= \beta\frac{SI}{N}.
\end{align}
Its analytic reproduction number is
\[
R_0=\frac{\beta\sigma}{(\sigma+\mu)(\gamma+\mu)}.
\]
The log form
\[
\log R_0=\log\beta+\log\sigma-\log(\sigma+\mu)-\log(\gamma+\mu)
\]
shows why sign interpretation must use derivatives rather than a numerator/denominator mnemonic alone. In particular,
\[
\frac{\partial\log R_0}{\partial\sigma}=\frac{1}{\sigma}-\frac{1}{\sigma+\mu}>0,
\qquad
\frac{\partial\log R_0}{\partial\mu}=-\frac{1}{\sigma+\mu}-\frac{1}{\gamma+\mu}<0.
\]
Thus $\beta$ and $\sigma$ are threshold-increasing, whereas $\gamma$ and $\mu$ are threshold-decreasing.

\paragraph{REEBLRA results.}
The binary demonstration used the declared model-aware basis and five-fold held-out validation. The unstandardized score below is the posterior-mode (MAP) estimate at the center of the Gaussian Laplace approximation. 
\[\resizebox{0.98\linewidth}{!}{$\displaystyle \eta_{R_0>1}=-14.5183 + 14.9421\,\log(\beta) + 7.9717\,\log(\sigma) + 5.7728\,\log(\gamma)\,\log(\mu) - 5.7618\,\log(\gamma) - 5.6778\,\log(\mu) + \cdots.$}\]

The ordered-bracket demonstration used the generic compact log-cubic basis with step $h=5$, five-fold held-out validation, and independent $N(0,2^2)$ prior and MAP/Laplace fitting routine. The displayed scores are MAP plug-in logit equations, not integrated posterior-predictive probabilities. Representative expressions were:

\begin{center}
\resizebox{0.98\linewidth}{!}{$\displaystyle \eta_{R_0>1} = -19.5161 + 17.1601\,\log(\beta) - 7.6621\,\log(\mu) + 7.0568\,\log(\gamma)\,\log(\mu) + 7.0479\,\log(\sigma) - 6.7967\,\log(\gamma) + \cdots$}\\[2pt]
\resizebox{0.98\linewidth}{!}{$\displaystyle \eta_{R_0>6} = -21.2621 + 4.2635\,\log(\gamma)\,\log(\mu) - 3.2188\,\log(\beta)\,\log(\gamma) - 2.2119\,\log(\beta)\,\log(\mu) + 2.1605\,\log(\beta)^2 + 1.9233\,\log(\gamma) + \cdots$}\\[2pt]
\resizebox{0.98\linewidth}{!}{$\displaystyle \eta_{R_0>10} = -12.2716 + 4.0578\,\log(\gamma) - 2.8369\,\log(\beta)\,\log(\gamma) + 2.7234\,\log(\gamma)\,\log(\mu) - 1.3846\,\log(\beta) + 1.3353\,\log(\beta)^2 + \cdots$}
\end{center}

\begin{table}[ht]
\centering
\caption{Mechanistic--REEBLRA coefficient concordance for the Open SEIR demonstration. The analytic role identifies whether each parameter enters the $R_0$ expression through a numerator, denominator, both, or the next-generation matrix (NGM). $G_j$ is the median analytic log-gradient $\partial\log R_0/\partial\log\theta_j$ and $D_j$ is the median derivative of the complete REEBLRA binary score over $0.8\le R_0\le1.2$. For this compact log-cubic basis, $b_j$ is the unstandardized first-order $\log\theta_j$ coefficient; interactions and curvature terms can make its isolated sign differ from the net parameter direction. Concordance is assessed from the sign of the full-score direction $D_j$ relative to $G_j$.}
\label{tab:open-seir-concordance}
\scriptsize
\setlength{\tabcolsep}{3pt}
\begin{tabularx}{\linewidth}{l >{\raggedright\arraybackslash}X r r r c}
\toprule
Parameter & Analytic $R_0$ role & Analytic $G_j$ & REEBLRA $b_j$ & Full-score $D_j$ & Concordant? \\
\midrule
$\beta$ & Numerator & +1.000 & +14.942 & +22.117 & Yes \\
$\sigma$ & Numerator and denominator: $\sigma/(\sigma+\mu)$ & +0.191 & +7.972 & +4.545 & Yes \\
$\gamma$ & Denominator: $(\gamma+\mu)$ & -0.766 & -5.762 & -16.352 & Yes \\
$\mu$ & Denominator: $(\sigma+\mu)(\gamma+\mu)$ & -0.463 & -5.678 & -9.580 & Yes \\
\bottomrule
\end{tabularx}
\end{table}

\begin{table}[ht]
\centering
\caption{Five independent parameter examples for the Open SEIR demonstration. The analytic $R_0$ is compared with the binary REEBLRA prediction and the ordered-bracket REEBLRA prediction. The binary entry reports the predicted side of the unit threshold with $P(R_0>1)$ in parentheses; the ordered entry reports the modal bracket with its normalized class probability. Cases are selected from an independent deterministic Latin-hypercube design (seed 2602) by proximity to prespecified analytic targets $0.7,2,5,8,12$; prediction results were not used for case selection.}
\label{tab:open-seir-five-examples}
\scriptsize
\setlength{\tabcolsep}{3pt}
\begin{tabularx}{\linewidth}{c >{\raggedright\arraybackslash}X c c c}
\toprule
Case & Parameter values & Analytic $R_0$ (bracket) & Binary REEBLRA & Ordered REEBLRA \\
\midrule
E1 & $\beta=0.2015,\allowbreak\,\sigma=0.2111,\allowbreak\,\gamma=0.2345,\allowbreak\,\mu=0.02401$ & 0.700 ($<1$) & $\leq1$ ($P_{>1}=0.000$) & $<1$ ($p=1.000$) \\
E2 & $\beta=1.214,\allowbreak\,\sigma=0.5854,\allowbreak\,\gamma=0.3574,\allowbreak\,\mu=0.1355$ & 2.000 ($1\text{--}6$) & $>1$ ($P_{>1}=1.000$) & $1\text{--}6$ ($p=1.000$) \\
E3 & $\beta=1.44,\allowbreak\,\sigma=0.9619,\allowbreak\,\gamma=0.2338,\allowbreak\,\mu=0.04209$ & 5.000 ($1\text{--}6$) & $>1$ ($P_{>1}=1.000$) & $1\text{--}6$ ($p=0.947$) \\
E4 & $\beta=1.492,\allowbreak\,\sigma=0.489,\allowbreak\,\gamma=0.1284,\allowbreak\,\mu=0.04318$ & 7.993 ($6\text{--}10$) & $>1$ ($P_{>1}=1.000$) & $6\text{--}10$ ($p=0.895$) \\
E5 & $\beta=1.131,\allowbreak\,\sigma=0.9985,\allowbreak\,\gamma=0.07675,\allowbreak\,\mu=0.01663$ & 11.917 ($>10$) & $>1$ ($P_{>1}=1.000$) & $>10$ ($p=0.843$) \\
\bottomrule
\end{tabularx}
\end{table}

The open-SEIR example demonstrates why REEBLRA is useful when the threshold surface is nonlinear even though its qualitative mechanism is known. The full binary score gives positive local directions for $\beta$ and $\sigma$ and negative directions for $\gamma$ and $\mu$, matching the analytic derivatives (Table~\ref{tab:open-seir-concordance}). The fitted main effect of $\beta$ is the largest positive term, while $\sigma$ is weaker because faster exposed-to-infectious progression increases transmission potential but also appears inside $\sigma+\mu$. Conversely, demographic turnover acts through both denominator factors and has a substantial negative full-score effect. The interaction terms are informative rather than decorative: they show that the learned decision boundary needs joint parameter combinations to reproduce the curvature of the analytic threshold. For users working with models whose exact $R_0$ is cumbersome, this pattern can distinguish a parameter with a consistently monotone role from one whose contribution is conditional on other rates, while still retaining a compact interpretable classifier.

Table~\ref{tab:open-seir-five-examples} {illustrates recovery of $R_0$ brackets} rather than pointwise numerical values.

\subsection{SEIR model with vaccination}

The vaccinated SEIR model with demographic turnover is
\begin{align}
\frac{dS}{dt} &= \mu N-\beta\frac{SI}{N}-vS-\mu S,\\
\frac{dV}{dt} &= vS-\mu V,\\
\frac{dE}{dt} &= \beta\frac{SI}{N}-\sigma E-\mu E,\\
\frac{dI}{dt} &= \sigma E-\gamma I-\mu I,\\
\frac{dR}{dt} &= \gamma I-\mu R,\\
\frac{dC}{dt} &= \beta\frac{SI}{N}.
\end{align}
The app reports the control or effective reproduction number under continuous vaccination as
\[
R_e(0)=\frac{\mu}{\mu+v}\cdot\frac{\beta\sigma}{(\sigma+\mu)(\gamma+\mu)}.
\]
The susceptible-fraction multiplier $\mu/(\mu+v)$ is included because the disease-free susceptible fraction under continuous vaccination is $S^*/N=\mu/(\mu+v)$. The basic reproduction number without vaccination excludes this multiplier. Here the effect of $\mu$ cannot be inferred from a simple numerator/denominator label because it occurs in both. Specifically,
\[
\frac{\partial\log R_e(0)}{\partial\mu}
=\frac{1}{\mu}-\frac{1}{\mu+v}-\frac{1}{\sigma+\mu}-\frac{1}{\gamma+\mu},
\]
so its net direction depends on the local parameter regime.

\paragraph{REEBLRA results.}
The binary demonstration used the declared model-aware basis and five-fold held-out validation. The unstandardized score below is the posterior-mode (MAP) estimate at the center of the Gaussian Laplace approximation. 
\[\resizebox{0.98\linewidth}{!}{$\displaystyle \eta_{R_e(0)>1}=-16.0078 + 10.2025\,\log(\beta) + 4.0229\,\log(\sigma) - 3.8673\,\log(\gamma) + 3.7181\,\log(\mu)\,\log(v) + 3.5497\,\log(\gamma)\,\log(\mu) + \cdots.$}\]

The ordered-bracket demonstration used the generic compact log-cubic basis with step $h=5$, five-fold held-out validation, and independent $N(0,2^2)$ prior and MAP/Laplace fitting routine. The displayed scores are MAP plug-in logit equations, not integrated posterior-predictive probabilities. Representative expressions were:

\begin{center}
\resizebox{0.98\linewidth}{!}{$\displaystyle \eta_{R_e(0)>1} = -17.4662 + 10.3197\,\log(\beta) + 4.5834\,\log(\mu)\,\log(v) + 4.5611\,\log(\sigma) + 4.2399\,\log(\gamma)\,\log(\mu) - 4.2182\,\log(\gamma) + \cdots$}\\[2pt]
\resizebox{0.98\linewidth}{!}{$\displaystyle \eta_{R_e(0)>6} = -28.7462 + 3.6188\,\log(\beta)^2 + 3.5368\,\log(\gamma)\,\log(\mu) + 2.9613\,\log(\mu)\,\log(v) - 2.3085\,\log(\beta)\,\log(\gamma) + 1.9915\,\log(\beta)\,\log(\sigma) + \cdots$}\\[2pt]
\resizebox{0.98\linewidth}{!}{$\displaystyle \eta_{R_e(0)>10} = -19.2916 + 3.7470\,\log(\gamma) - 2.3729\,\log(\beta) + 2.0502\,\log(\beta)^2 + 1.9924\,\log(\gamma)\,\log(\mu) - 1.9262\,\log(\beta)\,\log(\gamma) + \cdots$}
\end{center}

\begin{table}[ht]
\centering
\caption{Mechanistic--REEBLRA coefficient concordance for the SEIR with vaccination demonstration. The analytic role identifies whether each parameter enters the $R_e(0)$ expression through a numerator, denominator, both, or the next-generation matrix (NGM). $G_j$ is the median analytic log-gradient $\partial\log R_e(0)/\partial\log\theta_j$ and $D_j$ is the median derivative of the complete REEBLRA binary score over $0.8\le R_e(0)\le1.2$. For this compact log-cubic basis, $b_j$ is the unstandardized first-order $\log\theta_j$ coefficient; interactions and curvature terms can make its isolated sign differ from the net parameter direction. Concordance is assessed from the sign of the full-score direction $D_j$ relative to $G_j$.}
\label{tab:seir-with-vaccination-concordance}
\scriptsize
\setlength{\tabcolsep}{3pt}
\begin{tabularx}{\linewidth}{l >{\raggedright\arraybackslash}X r r r c}
\toprule
Parameter & Analytic $R_e(0)$ role & Analytic $G_j$ & REEBLRA $b_j$ & Full-score $D_j$ & Concordant? \\
\midrule
$\beta$ & Numerator & +1.000 & +10.203 & +16.804 & Yes \\
$\sigma$ & Numerator and denominator: $\sigma/(\sigma+\mu)$ & +0.262 & +4.023 & +4.675 & Yes \\
$\gamma$ & Denominator: $(\gamma+\mu)$ & -0.678 & -3.867 & -10.892 & Yes \\
$\mu$ & Numerator and denominators & -0.354 & +1.357 & -5.349 & Yes \\
$v$ & Denominator: $(\mu+v)$ & -0.326 & +0.151 & -4.816 & Yes \\
\bottomrule
\end{tabularx}
\end{table}

\begin{table}[ht]
\centering
\caption{Five independent parameter examples for the SEIR with vaccination demonstration. The analytic $R_e(0)$ is compared with the binary REEBLRA prediction and the ordered-bracket REEBLRA prediction. The binary entry reports the predicted side of the unit threshold with $P(R_e(0)>1)$ in parentheses; the ordered entry reports the modal bracket with its normalized class probability. Cases are selected from an independent deterministic Latin-hypercube design (seed 2603) by proximity to prespecified analytic targets $0.7,2,5,8,12$; prediction results were not used for case selection.}
\label{tab:seir-with-vaccination-five-examples}
\scriptsize
\setlength{\tabcolsep}{3pt}
\begin{tabularx}{\linewidth}{c >{\raggedright\arraybackslash}X c c c}
\toprule
Case & Parameter values & Analytic $R_e$ (bracket) & Binary REEBLRA & Ordered REEBLRA \\
\midrule
E1 & $\beta=0.1994,\allowbreak\,\sigma=0.2458,\allowbreak\,\gamma=0.06908,\allowbreak\,\mu=0.03407,\allowbreak\,v=0.04856$ & 0.700 ($<1$) & $\leq1$ ($P_{>1}=0.002$) & $<1$ ($p=0.999$) \\
E2 & $\beta=3.947,\allowbreak\,\sigma=0.6551,\allowbreak\,\gamma=0.3801,\allowbreak\,\mu=0.02933,\allowbreak\,v=0.106$ & 2.000 ($1\text{--}6$) & $>1$ ($P_{>1}=1.000$) & $1\text{--}6$ ($p=1.000$) \\
E3 & $\beta=2.623,\allowbreak\,\sigma=0.6904,\allowbreak\,\gamma=0.1032,\allowbreak\,\mu=0.1549,\allowbreak\,v=0.1023$ & 4.998 ($1\text{--}6$) & $>1$ ($P_{>1}=1.000$) & $1\text{--}6$ ($p=0.969$) \\
E4 & $\beta=2.208,\allowbreak\,\sigma=0.69,\allowbreak\,\gamma=0.09018,\allowbreak\,\mu=0.0204,\allowbreak\,v=0.02902$ & 8.004 ($6\text{--}10$) & $>1$ ($P_{>1}=1.000$) & $6\text{--}10$ ($p=0.902$) \\
E5 & $\beta=3.796,\allowbreak\,\sigma=0.7814,\allowbreak\,\gamma=0.03307,\allowbreak\,\mu=0.1389,\allowbreak\,v=0.07906$ & 11.943 ($>10$) & $>1$ ($P_{>1}=1.000$) & $>10$ ($p=0.841$) \\
\bottomrule
\end{tabularx}
\end{table}

This example is particularly important for interpretation because a main polynomial coefficient is not necessarily the net parameter effect. The isolated $\log v$ main coefficient is slightly positive, but the complete fitted score gives $D_v=-4.82$ near $R_e(0)=1$, correctly recovering the threshold-reducing role of vaccination once interactions are included. Likewise, $\mu$ appears in both favorable and unfavorable algebraic locations in $R_e(0)$; its local full-score derivative is negative over the sampled boundary even though its standalone main coefficient is positive. This is exactly the situation for which reporting only a coefficient table could be misleading. REEBLRA therefore gives users two levels of insight: individual fitted feature terms describe how the surrogate constructs the boundary, while the derivative of the full score describes the net local direction of the original epidemiological parameter. For intervention models, that distinction can identify when the effect of a control parameter is robustly threshold-lowering and when its apparent role depends on interactions or the operating regime (Table~\ref{tab:seir-with-vaccination-concordance}).

Table~\ref{tab:seir-with-vaccination-five-examples} {illustrates recovery of $R_0$ brackets} rather than pointwise numerical values.

\subsection{Age-structured SIR model (two groups)}

The age-structured SIR model is
\begin{align}
\lambda_1 &= \beta\left(c_{11}\frac{I_1}{N_1}+c_{12}\frac{I_2}{N_2}\right),\\
\lambda_2 &= \beta\left(c_{21}\frac{I_1}{N_1}+c_{22}\frac{I_2}{N_2}\right),\\
\frac{dS_1}{dt} &= -\lambda_1 S_1, &
\frac{dI_1}{dt} &= \lambda_1 S_1-\gamma_1 I_1, &
\frac{dR_1}{dt} &= \gamma_1 I_1,\\
\frac{dS_2}{dt} &= -\lambda_2 S_2, &
\frac{dI_2}{dt} &= \lambda_2 S_2-\gamma_2 I_2, &
\frac{dR_2}{dt} &= \gamma_2 I_2,\\
\frac{dC}{dt} &= \lambda_1S_1+\lambda_2S_2.
\end{align}
The analytic reproduction number is the spectral radius of
\[
K=\begin{pmatrix}
\dfrac{\beta c_{11}}{\gamma_1} & \dfrac{\beta c_{12}}{\gamma_2}\\[0.7em]
\dfrac{\beta c_{21}}{\gamma_1} & \dfrac{\beta c_{22}}{\gamma_2}
\end{pmatrix},
\qquad
R_0=\rho(K),
\]
assuming equally sized, fully susceptible groups at the disease-free state. In the app, total infectious and susceptible columns are constructed as $I_{\mathrm{total}}=I_1+I_2$ and $S_{\mathrm{total}}=S_1+S_2$ for growth-rate fitting.

\paragraph{REEBLRA results.}
The binary demonstration used the declared model-aware basis and five-fold held-out validation. The unstandardized score below is the posterior-mode (MAP) estimate at the center of the Gaussian Laplace approximation. 
\[\resizebox{0.98\linewidth}{!}{$\displaystyle \eta_{R_0>1}=6.1983 + 8.2880\,\log(\beta) - 6.1155\,\log(\gamma_1) - 6.0797\,\log(\gamma_2) + 4.7133\,\log(\gamma_2)\,\log(c_{11}) + 4.4751\,\log(\gamma_1)\,\log(c_{22}) + \cdots.$}\]

The ordered-bracket demonstration used the generic compact log-cubic basis with step $h=5$, five-fold held-out validation, and independent $N(0,2^2)$ prior and MAP/Laplace fitting routine. The displayed scores are MAP plug-in logit equations, not integrated posterior-predictive probabilities. Representative expressions were:

\begin{center}
\resizebox{0.98\linewidth}{!}{$\displaystyle \eta_{R_0>1} = 7.6707 + 11.2516\,\log(\beta) - 7.5780\,\log(\gamma_2) - 6.8347\,\log(\gamma_1) + 6.8132\,\log(\gamma_1)\,\log(c_{22}) + 6.7955\,\log(\gamma_2)\,\log(c_{11}) + \cdots$}\\[2pt]
\resizebox{0.98\linewidth}{!}{$\displaystyle \eta_{R_0>6} = -17.6514 + 6.4183\,\log(c_{11}) + 5.9754\,\log(\beta) - 5.9168\,\log(\gamma_1) + 5.6031\,\log(c_{22}) - 5.5716\,\log(\gamma_2) + \cdots$}\\[2pt]
\resizebox{0.98\linewidth}{!}{$\displaystyle \eta_{R_0>10} = -11.1985 + 6.0485\,\log(c_{11}) + 4.1984\,\log(c_{22}) - 3.5351\,\log(\beta)\,\log(\gamma_1) - 3.1675\,\log(c_{11})\,\log(c_{22}) - 2.9984\,\log(\beta)\,\log(\gamma_2) + \cdots$}
\end{center}

\begin{table}[ht]
\centering
\caption{Mechanistic--REEBLRA coefficient concordance for the Age-structured SIR demonstration. The analytic role identifies whether each parameter enters the $R_0$ expression through a numerator, denominator, both, or the next-generation matrix (NGM). $G_j$ is the median analytic log-gradient $\partial\log R_0/\partial\log\theta_j$ and $D_j$ is the median derivative of the complete REEBLRA binary score over $0.8\le R_0\le1.2$. For this compact log-cubic basis, $b_j$ is the unstandardized first-order $\log\theta_j$ coefficient; interactions and curvature terms can make its isolated sign differ from the net parameter direction. Concordance is assessed from the sign of the full-score direction $D_j$ relative to $G_j$.}
\label{tab:age-structured-sir-concordance}
\scriptsize
\setlength{\tabcolsep}{3pt}
\begin{tabularx}{\linewidth}{l >{\raggedright\arraybackslash}X r r r c}
\toprule
Parameter & Analytic $R_0$ role & Analytic $G_j$ & REEBLRA $b_j$ & Full-score $D_j$ & Concordant? \\
\midrule
$\beta$ & NGM numerator multiplier & +1.000 & +8.288 & +16.796 & Yes \\
$\gamma_1$ & NGM denominator (column 1) & -0.467 & -6.115 & -7.317 & Yes \\
$\gamma_2$ & NGM denominator (column 2) & -0.533 & -6.080 & -8.005 & Yes \\
$c_{11}$ & NGM numerator/contact & +0.245 & -1.569 & +4.462 & Yes \\
$c_{12}$ & NGM numerator/contact & +0.150 & +3.155 & +3.000 & Yes \\
$c_{21}$ & NGM numerator/contact & +0.150 & +2.806 & +2.798 & Yes \\
$c_{22}$ & NGM numerator/contact & +0.302 & +1.059 & +5.574 & Yes \\
\bottomrule
\end{tabularx}
\end{table}

\begin{table}[ht]
\centering
\caption{Five independent parameter examples for the Age-structured SIR demonstration. The analytic $R_0$ is compared with the binary REEBLRA prediction and the ordered-bracket REEBLRA prediction. The binary entry reports the predicted side of the unit threshold with $P(R_0>1)$ in parentheses; the ordered entry reports the modal bracket with its normalized class probability. Cases are selected from an independent deterministic Latin-hypercube design (seed 2604) by proximity to prespecified analytic targets $0.7,2,5,8,12$; prediction results were not used for case selection.}
\label{tab:age-structured-sir-five-examples}
\scriptsize
\setlength{\tabcolsep}{3pt}
\begin{tabularx}{\linewidth}{c >{\raggedright\arraybackslash}X c c c}
\toprule
Case & Parameter values & Analytic $R_0$ (bracket) & Binary REEBLRA & Ordered REEBLRA \\
\midrule
E1 & $\beta=0.1596,\allowbreak\,\gamma_1=0.4024,\allowbreak\,\gamma_2=0.308,\allowbreak\,c_{11}=1.328,\allowbreak\,c_{12}=0.2807,\allowbreak\,c_{21}=0.6058,\allowbreak\,c_{22}=0.9627$ & 0.700 ($<1$) & $\leq1$ ($P_{>1}=0.013$) & $<1$ ($p=0.994$) \\
E2 & $\beta=0.132,\allowbreak\,\gamma_1=0.132,\allowbreak\,\gamma_2=0.1189,\allowbreak\,c_{11}=0.9655,\allowbreak\,c_{12}=0.9897,\allowbreak\,c_{21}=0.3708,\allowbreak\,c_{22}=1.448$ & 2.000 ($1\text{--}6$) & $>1$ ($P_{>1}=1.000$) & $1\text{--}6$ ($p=1.000$) \\
E3 & $\beta=0.2803,\allowbreak\,\gamma_1=0.4069,\allowbreak\,\gamma_2=0.08217,\allowbreak\,c_{11}=0.5081,\allowbreak\,c_{12}=0.08871,\allowbreak\,c_{21}=0.4798,\allowbreak\,c_{22}=1.459$ & 5.000 ($1\text{--}6$) & $>1$ ($P_{>1}=1.000$) & $1\text{--}6$ ($p=0.830$) \\
E4 & $\beta=0.7881,\allowbreak\,\gamma_1=0.07343,\allowbreak\,\gamma_2=0.2996,\allowbreak\,c_{11}=0.6871,\allowbreak\,c_{12}=0.1643,\allowbreak\,c_{21}=0.6821,\allowbreak\,c_{22}=1.125$ & 8.003 ($6\text{--}10$) & $>1$ ($P_{>1}=1.000$) & $6\text{--}10$ ($p=0.873$) \\
E5 & $\beta=0.5928,\allowbreak\,\gamma_1=0.4841,\allowbreak\,\gamma_2=0.06423,\allowbreak\,c_{11}=1.142,\allowbreak\,c_{12}=0.6892,\allowbreak\,c_{21}=0.1383,\allowbreak\,c_{22}=1.288$ & 11.991 ($>10$) & $>1$ ($P_{>1}=1.000$) & $>10$ ($p=0.739$) \\
\bottomrule
\end{tabularx}
\end{table}

The spectral-radius example extends the interpretation beyond scalar numerator--denominator formulas. The full REEBLRA score is positive in $\beta$ and all four contact coefficients, and negative in both recovery rates near the threshold, consistent with Perron--Frobenius monotonicity of the nonnegative next-generation matrix (Table \ref{tab:age-structured-sir-concordance}). Some isolated main coefficients, such as that of $\log c_{11}$, need not have the same sign because the compact log-cubic basis also contains contact--contact, contact--recovery, and curvature terms; the positive full-score derivative restores the correct parameter-level direction. This gives users a practical way to interrogate structured models: REEBLRA can indicate whether transmission intensity, recovery, or particular mixing pathways are driving movement toward the invasion boundary without requiring the spectral-radius expression to be symbolically simplified. Differences in $D_j$ across contact terms can be used as local, domain-specific screening signals, but they should not be interpreted as universal causal rankings because they depend on the fixed matrix, sampled ranges, and the current threshold.

Table~\ref{tab:age-structured-sir-five-examples} {illustrates recovery of $R_0$ brackets} rather than pointwise numerical values.

\subsection{Fifteen-compartment five-group SIR model}\label{subsec:five-group}

To test whether REEBLRA remains useful beyond small textbook systems, we added a five-group SIR model with 15 epidemiological compartments, $(S_i,I_i,R_i)$ for $i=1,\ldots,5$ (Figure~\ref{fig:five-group}). The numerical helper states for cumulative incidence are not counted among the 15 epidemiological compartments. For each group,
\begin{align}
\lambda_i &= \beta\sum_{j=1}^{5} C_{ij}\frac{I_j}{N_j},\\
\frac{dS_i}{dt} &= -\lambda_iS_i, &
\frac{dI_i}{dt} &= \lambda_iS_i-\gamma_iI_i, &
\frac{dR_i}{dt} &= \gamma_iI_i.
\end{align}
The fixed mixing matrix is
\[
C=\begin{pmatrix}
3.0&1.2&0.6&0.3&0.2\\
1.0&3.2&1.1&0.5&0.3\\
0.5&1.0&2.8&1.0&0.5\\
0.3&0.5&1.0&2.6&0.9\\
0.2&0.3&0.5&0.9&2.4
\end{pmatrix}.
\]
With fully susceptible groups, the next-generation matrix is
\[
K_{ij}=\frac{\beta C_{ij}}{\gamma_j},\qquad R_0=\rho(K).
\]
The built-in parameter vector is therefore only six-dimensional, $(\beta,\gamma_1,\ldots,\gamma_5)$, even though the dynamic state dimension is 15. This separation is useful for testing whether state-space complexity alone defeats the surrogate threshold approach.

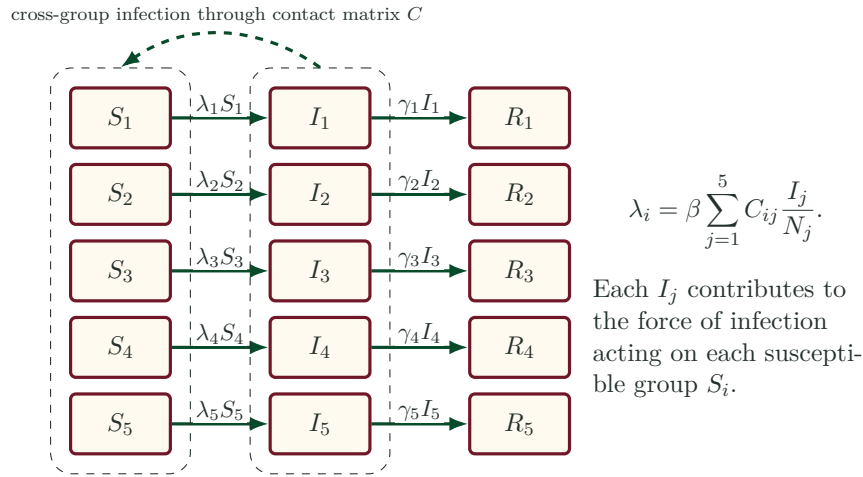
\begin{figure}[!htbp]
\centering
\begin{tikzpicture}[
    node distance=8mm and 14mm,
    scale=0.88,
    transform shape
]

% Five group-specific SIR triplets
\foreach \i in {1,...,5}{
  \node[compartment] (S\i) at (0,-1.15*\i) {$S_{\i}$};
  \node[compartment] (I\i) at (3.0,-1.15*\i) {$I_{\i}$};
  \node[compartment] (R\i) at (6.0,-1.15*\i) {$R_{\i}$};

  \draw[flow] (S\i) --
    node[flowlabel,above] {$\lambda_{\i}S_{\i}$}
    (I\i);

  \draw[flow] (I\i) --
    node[flowlabel,above] {$\gamma_{\i}I_{\i}$}
    (R\i);
}

% Dashed enclosures for susceptible and infectious groups
\node[
    draw,
    dashed,
    rounded corners=8pt,
    inner xsep=7pt,
    inner ysep=7pt,
    fit=(S1)(S5)
] (Sgroup) {};

\node[
    draw,
    dashed,
    rounded corners=8pt,
    inner xsep=7pt,
    inner ysep=7pt,
    fit=(I1)(I5)
] (Igroup) {};

% Aggregate cross-group infection link
\draw[
    flow,
    dashed,
    bend right=35
]
(Igroup.north)
to node[
    above,
    align=center,
    font=\scriptsize
] {
    cross-group infection
    through contact matrix $C$
}
(Sgroup.north);

% Force-of-infection definition
\node[
    align=left,
    text width=4.0cm
] at (9.1,-3.45) {
\[
\lambda_i
=
\beta
\sum_{j=1}^{5}
C_{ij}\frac{I_j}{N_j}.
\]
Each $I_j$ contributes to the
force of infection acting on
each susceptible group $S_i$.
};

\end{tikzpicture}

\caption{Compact representation of the 15-compartment five-group SIR benchmark.
Each row represents one group-specific SIR triplet. The dashed enclosures and
aggregate dashed arrow indicate cross-group infection: infectious individuals
from all groups contribute to the force of infection acting on susceptible
individuals in all groups through the fixed contact matrix $C$.}
\label{fig:five-group}
\end{figure}

\paragraph{REEBLRA results.}
The binary demonstration used the declared model-aware basis and five-fold held-out validation. The unstandardized score below is the posterior-mode (MAP) estimate at the center of the Gaussian Laplace approximation. 
\[\resizebox{0.98\linewidth}{!}{$\displaystyle \eta_{R_0>1}=0.9042 - 2.8558\,\log(\gamma_2) - 2.8100\,\log(\gamma_1) - 2.5984\,\log(\gamma_5) - 2.0845\,\log(\gamma_3) - 2.0739\,\log(\gamma_2)\,\log(\gamma_5) + \cdots.$}\]

The ordered-bracket demonstration used the generic compact log-cubic basis with step $h=5$, five-fold held-out validation, and independent $N(0,2^2)$ prior and MAP/Laplace fitting routine. The displayed scores are MAP plug-in logit equations, not integrated posterior-predictive probabilities. Representative expressions were:

\begin{center}
\resizebox{0.98\linewidth}{!}{$\displaystyle \eta_{R_0>1} = 4.0181 - 3.4039\,\log(\gamma_3) - 2.4754\,\log(\gamma_2)\,\log(\gamma_4) - 2.1734\,\log(\gamma_5) - 2.0688\,\log(\gamma_2) - 1.9267\,\log(\gamma_4) + \cdots$}\\[2pt]
\resizebox{0.98\linewidth}{!}{$\displaystyle \eta_{R_0>6} = -16.6150 - 7.7408\,\log(\gamma_2) - 6.2590\,\log(\gamma_4) - 6.2427\,\log(\gamma_3) - 5.7667\,\log(\gamma_1) - 4.5164\,\log(\beta)\,\log(\gamma_2) + \cdots$}\\[2pt]
\resizebox{0.98\linewidth}{!}{$\displaystyle \eta_{R_0>10} = -6.5899 - 4.3635\,\log(\gamma_1) - 3.8354\,\log(\beta)\,\log(\gamma_3) - 3.6487\,\log(\beta)\,\log(\gamma_2) - 3.5186\,\log(\gamma_5) - 3.3570\,\log(\gamma_3) + \cdots$}
\end{center}

\begin{table}[ht]
\centering
\caption{Mechanistic--REEBLRA coefficient concordance for the 15-compartment five-group SIR demonstration. The analytic role identifies whether each parameter enters the $R_0$ expression through a numerator, denominator, both, or the next-generation matrix (NGM). $G_j$ is the median analytic log-gradient $\partial\log R_0/\partial\log\theta_j$ and $D_j$ is the median derivative of the complete REEBLRA binary score over $0.8\le R_0\le1.2$. For this compact log-cubic basis, $b_j$ is the unstandardized first-order $\log\theta_j$ coefficient; interactions and curvature terms can make its isolated sign differ from the net parameter direction. Concordance is assessed from the sign of the full-score direction $D_j$ relative to $G_j$.}
\label{tab:15-compartment-five-group-sir-concordance}
\scriptsize
\setlength{\tabcolsep}{3pt}
\begin{tabularx}{\linewidth}{l >{\raggedright\arraybackslash}X r r r c}
\toprule
Parameter & Analytic $R_0$ role & Analytic $G_j$ & REEBLRA $b_j$ & Full-score $D_j$ & Concordant? \\
\midrule
$\beta$ & NGM numerator multiplier & +1.000 & +0.048 & +12.710 & Yes \\
$\gamma_1$ & NGM denominator (column 1) & -0.106 & -2.810 & -2.352 & Yes \\
$\gamma_2$ & NGM denominator (column 2) & -0.168 & -2.856 & -1.894 & Yes \\
$\gamma_3$ & NGM denominator (column 3) & -0.118 & -2.085 & -2.143 & Yes \\
$\gamma_4$ & NGM denominator (column 4) & -0.078 & -1.279 & -0.913 & Yes \\
$\gamma_5$ & NGM denominator (column 5) & -0.045 & -2.598 & -1.067 & Yes \\
\bottomrule
\end{tabularx}
\end{table}

\begin{table}[ht]
\centering
\caption{Five independent parameter examples for the 15-compartment five-group SIR demonstration. The analytic $R_0$ is compared with the binary REEBLRA prediction and the ordered-bracket REEBLRA prediction. The binary entry reports the predicted side of the unit threshold with $P(R_0>1)$ in parentheses; the ordered entry reports the modal bracket with its normalized class probability. Cases are selected from an independent deterministic Latin-hypercube design (seed 2605) by proximity to prespecified analytic targets $0.7,2,5,8,12$; prediction results were not used for case selection.}
\label{tab:15-compartment-five-group-sir-five-examples}
\scriptsize
\setlength{\tabcolsep}{3pt}
\begin{tabularx}{\linewidth}{c >{\raggedright\arraybackslash}X c c c}
\toprule
Case & Parameter values & Analytic $R_0$ (bracket) & Binary REEBLRA & Ordered REEBLRA \\
\midrule
E1 & $\beta=0.04168,\allowbreak\,\gamma_1=0.4096,\allowbreak\,\gamma_2=0.4303,\allowbreak\,\gamma_3=0.4377,\allowbreak\,\gamma_4=0.2124,\allowbreak\,\gamma_5=0.32$ & 0.700 ($<1$) & $\leq1$ ($P_{>1}=0.033$) & $<1$ ($p=0.992$) \\
E2 & $\beta=0.09297,\allowbreak\,\gamma_1=0.2553,\allowbreak\,\gamma_2=0.2499,\allowbreak\,\gamma_3=0.2875,\allowbreak\,\gamma_4=0.211,\allowbreak\,\gamma_5=0.3328$ & 2.000 ($1\text{--}6$) & $>1$ ($P_{>1}=0.998$) & $1\text{--}6$ ($p=0.999$) \\
E3 & $\beta=0.1023,\allowbreak\,\gamma_1=0.1678,\allowbreak\,\gamma_2=0.1676,\allowbreak\,\gamma_3=0.283,\allowbreak\,\gamma_4=0.05951,\allowbreak\,\gamma_5=0.2702$ & 5.001 ($1\text{--}6$) & $>1$ ($P_{>1}=1.000$) & $1\text{--}6$ ($p=1.000$) \\
E4 & $\beta=0.233,\allowbreak\,\gamma_1=0.1805,\allowbreak\,\gamma_2=0.4199,\allowbreak\,\gamma_3=0.09769,\allowbreak\,\gamma_4=0.3524,\allowbreak\,\gamma_5=0.1475$ & 8.000 ($6\text{--}10$) & $>1$ ($P_{>1}=1.000$) & $6\text{--}10$ ($p=0.988$) \\
E5 & $\beta=0.2597,\allowbreak\,\gamma_1=0.3044,\allowbreak\,\gamma_2=0.169,\allowbreak\,\gamma_3=0.2414,\allowbreak\,\gamma_4=0.4289,\allowbreak\,\gamma_5=0.05476$ & 12.003 ($>10$) & $>1$ ($P_{>1}=1.000$) & $>10$ ($p=0.650$) \\
\bottomrule
\end{tabularx}
\end{table}

The 15-compartment example separates state-space complexity from parameter-space complexity. Although the dynamical system contains 15 epidemiological states, REEBLRA only needs to learn a six-parameter threshold surface because the mixing matrix is fixed (Table~\ref{tab:15-compartment-five-group-sir-concordance}). The complete fitted score recovers a positive local effect of the common transmission multiplier $\beta$ and negative local effects for all five recovery rates, in agreement with the monotonicity of $\rho(K)$. The very small isolated main coefficient of $\log\beta$ is not evidence that transmission is unimportant: once the many interaction and curvature terms are included, its full local effect is strongly positive ($D_\beta=12.71$). This is a concrete illustration of why parameter interpretation in an expanded feature basis must use the full score rather than one main coefficient. For users, the group-specific recovery derivatives provide a compact summary of which infectious-duration changes most strongly move the system across the threshold under the fixed contact structure; these are useful screening signals for model refinement or intervention exploration, but they remain local to the sampled parameter domain.

Table~\ref{tab:15-compartment-five-group-sir-five-examples} {illustrates recovery of $R_0$ brackets} rather than pointwise numerical values.

\subsection{Vector-borne model}\label{subsec:vector}

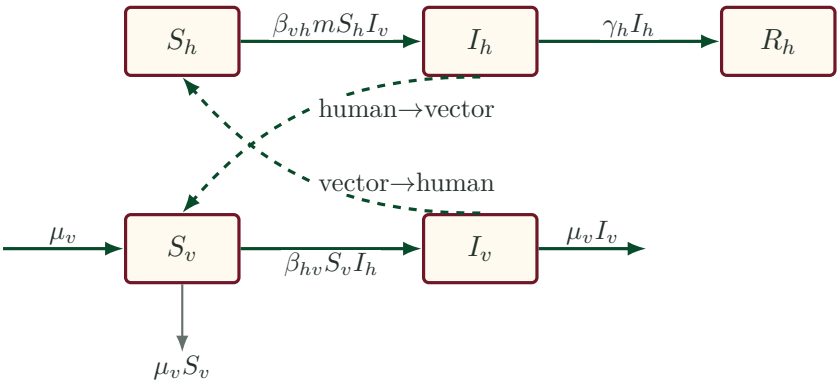
\begin{figure}[!htbp]
\centering
\begin{tikzpicture}[node distance=20mm and 24mm]
\node[compartment] (Sh) {$S_h$};
\node[compartment,right=of Sh] (Ih) {$I_h$};
\node[compartment,right=of Ih] (Rh) {$R_h$};
\node[compartment,below=18mm of Sh] (Sv) {$S_v$};
\node[compartment,right=of Sv] (Iv) {$I_v$};
\draw[flow] (Sh) -- node[flowlabel,above] {$\beta_{vh}mS_hI_v$} (Ih);
\draw[flow] (Ih) -- node[flowlabel,above] {$\gamma_h I_h$} (Rh);
\draw[flow] (Sv) -- node[flowlabel,below] {$\beta_{hv}S_vI_h$} (Iv);
\draw[flow] (Iv) -- ++(22mm,0) node[flowlabel,above,pos=0.5] {$\mu_v I_v$};
\draw[flow] ($(Sv.west)+(-16mm,0)$) -- node[flowlabel,above] {$\mu_v$} (Sv);
\draw[loss] (Sv.south) -- ++(0,-9mm) node[flowlabel,below] {$\mu_vS_v$};
\draw[flow,dashed] (Iv.north) to[bend left=25] node[flowlabel,right] {vector$\to$human} (Sh.south);
\draw[flow,dashed] (Ih.south) to[bend right=25] node[flowlabel,right] {human$\to$vector} (Sv.north);
\end{tikzpicture}
\caption{Human--vector transmission structure used by the built-in vector-borne simulator. Transmission between populations is represented by the dashed cross-population links.}
\label{fig:vector}
\end{figure}

The vector-borne model (Figure~\ref{fig:vector}) is
\begin{align}
\frac{dS_h}{dt} &= -\beta_{vh}mS_hI_v,\\
\frac{dI_h}{dt} &= \beta_{vh}mS_hI_v-\gamma_hI_h,\\
\frac{dR_h}{dt} &= \gamma_hI_h,\\
\frac{dS_v}{dt} &= \mu_v(1-S_v)-\beta_{hv}S_vI_h,\\
\frac{dI_v}{dt} &= \beta_{hv}S_vI_h-\mu_vI_v,\\
\frac{dC_h}{dt} &= \beta_{vh}mS_hI_v,\\
\frac{dC_v}{dt} &= \beta_{hv}S_vI_h.
\end{align}
Here $m$ is the vector-to-human ratio. The analytic reproduction number used by the app is
\[
R_0=\sqrt{\frac{\beta_{hv}\beta_{vh}m}{\gamma_h\mu_v}}.
\]
For growth-rate fitting, the major-host infectious column is $I_h$ and the major susceptible column is $S_h$.

\paragraph{REEBLRA results.}
The binary demonstration used the declared model-aware basis and five-fold held-out validation. The unstandardized score below is the posterior-mode (MAP) estimate at the center of the Gaussian Laplace approximation. 
\[\resizebox{0.98\linewidth}{!}{$\displaystyle \eta_{R_0>1}=0.3110 + 5.7502\,\log(\beta_{hv}) - 5.1796\,\log(\beta_{hv})\,\log(\beta_{vh}) - 3.9833\,\log(\mu_v) + 3.6503\,\log(m) + 3.4138\,\log(\beta_{vh}) + \cdots.$}\]

The ordered-bracket demonstration used the generic compact log-cubic basis with step $h=5$, five-fold held-out validation, and independent $N(0,2^2)$ prior and MAP/Laplace fitting routine. The displayed scores are MAP plug-in logit equations, not integrated posterior-predictive probabilities. Representative expressions were:

\begin{center}
\resizebox{0.98\linewidth}{!}{$\displaystyle \eta_{R_0>1} = 0.8594 - 5.7008\,\log(\beta_{hv})\,\log(\beta_{vh}) + 4.9146\,\log(\beta_{hv}) - 3.9926\,\log(\gamma_h) - 3.9242\,\log(\mu_v) + 3.8614\,\log(\beta_{vh}) + \cdots$}\\[2pt]
\resizebox{0.98\linewidth}{!}{$\displaystyle \eta_{R_0>6} = -46.9710 - 5.6409\,\log(\mu_v) - 5.5344\,\log(\gamma_h) - 5.3459\,\log(\gamma_h)\,\log(m) + 5.3098\,\log(\beta_{vh}) - 5.1692\,\log(\mu_v)\,\log(m) + \cdots$}\\[2pt]
\resizebox{0.98\linewidth}{!}{$\displaystyle \eta_{R_0>10} = -33.0565 + 3.7660\,\log(\beta_{hv})\,\log(\beta_{vh}) + 3.6142\,\log(\gamma_h)\,\log(\mu_v) - 3.5880\,\log(\beta_{vh})\,\log(\mu_v) - 3.5467\,\log(\beta_{hv})\,\log(\mu_v) - 3.5082\,\log(\beta_{hv})\,\log(\gamma_h) + \cdots$}
\end{center}

\begin{table}[ht]
\centering
\caption{Mechanistic--REEBLRA coefficient concordance for the Vector-borne demonstration. The analytic role identifies whether each parameter enters the $R_0$ expression through a numerator, denominator, both, or the next-generation matrix (NGM). $G_j$ is the median analytic log-gradient $\partial\log R_0/\partial\log\theta_j$ and $D_j$ is the median derivative of the complete REEBLRA binary score over $0.8\le R_0\le1.2$. For this compact log-cubic basis, $b_j$ is the unstandardized first-order $\log\theta_j$ coefficient; interactions and curvature terms can make its isolated sign differ from the net parameter direction. Concordance is assessed from the sign of the full-score direction $D_j$ relative to $G_j$.}
\label{tab:vector-borne-concordance}
\scriptsize
\setlength{\tabcolsep}{3pt}
\begin{tabularx}{\linewidth}{l >{\raggedright\arraybackslash}X r r r c}
\toprule
Parameter & Analytic $R_0$ role & Analytic $G_j$ & REEBLRA $b_j$ & Full-score $D_j$ & Concordant? \\
\midrule
$\beta_{hv}$ & Numerator (inside square root) & +0.500 & +5.750 & +10.042 & Yes \\
$\beta_{vh}$ & Numerator (inside square root) & +0.500 & +3.414 & +9.882 & Yes \\
$\gamma_h$ & Denominator (inside square root) & -0.500 & -2.934 & -9.234 & Yes \\
$\mu_v$ & Denominator (inside square root) & -0.500 & -3.983 & -9.974 & Yes \\
$m$ & Numerator (inside square root) & +0.500 & +3.650 & +9.724 & Yes \\
\bottomrule
\end{tabularx}
\end{table}

\begin{table}[ht]
\centering
\caption{Five independent parameter examples for the Vector-borne demonstration. The analytic $R_0$ is compared with the binary REEBLRA prediction and the ordered-bracket REEBLRA prediction. The binary entry reports the predicted side of the unit threshold with $P(R_0>1)$ in parentheses; the ordered entry reports the modal bracket with its normalized class probability. Cases are selected from an independent deterministic Latin-hypercube design (seed 2606) by proximity to prespecified analytic targets $0.7,2,5,8,12$; prediction results were not used for case selection.}
\label{tab:vector-borne-five-examples}
\scriptsize
\setlength{\tabcolsep}{3pt}
\begin{tabularx}{\linewidth}{c >{\raggedright\arraybackslash}X c c c}
\toprule
Case & Parameter values & Analytic $R_0$ (bracket) & Binary REEBLRA & Ordered REEBLRA \\
\midrule
E1 & $\beta_{hv}=1.36,\allowbreak\,\beta_{vh}=0.07201,\allowbreak\,\gamma_h=0.5546,\allowbreak\,\mu_v=0.4127,\allowbreak\,m=1.147$ & 0.701 ($<1$) & $\leq1$ ($P_{>1}=0.003$) & $<1$ ($p=0.999$) \\
E2 & $\beta_{hv}=0.6068,\allowbreak\,\beta_{vh}=1.111,\allowbreak\,\gamma_h=0.5159,\allowbreak\,\mu_v=0.3649,\allowbreak\,m=1.117$ & 2.000 ($1\text{--}6$) & $>1$ ($P_{>1}=0.998$) & $1\text{--}6$ ($p=0.999$) \\
E3 & $\beta_{hv}=0.1364,\allowbreak\,\beta_{vh}=1.425,\allowbreak\,\gamma_h=0.07617,\allowbreak\,\mu_v=0.3019,\allowbreak\,m=2.959$ & 5.000 ($1\text{--}6$) & $>1$ ($P_{>1}=1.000$) & $1\text{--}6$ ($p=0.999$) \\
E4 & $\beta_{hv}=0.6209,\allowbreak\,\beta_{vh}=0.8853,\allowbreak\,\gamma_h=0.4866,\allowbreak\,\mu_v=0.05874,\allowbreak\,m=3.328$ & 8.000 ($6\text{--}10$) & $>1$ ($P_{>1}=1.000$) & $6\text{--}10$ ($p=0.998$) \\
E5 & $\beta_{hv}=1.238,\allowbreak\,\beta_{vh}=1.117,\allowbreak\,\gamma_h=0.1694,\allowbreak\,\mu_v=0.1865,\allowbreak\,m=3.291$ & 12.002 ($>10$) & $>1$ ($P_{>1}=1.000$) & $>10$ ($p=0.997$) \\
\bottomrule
\end{tabularx}
\end{table}

Taking logarithms gives
\[
\log R_0=\tfrac12\left(\log\beta_{hv}+\log\beta_{vh}+\log m-\log\gamma_h-\log\mu_v\right).
\]
Thus the exact mechanistic log-elasticities have equal magnitude $1/2$: the two transmission rates and $m$ are positive directions, whereas $\gamma_h$ and $\mu_v$ are negative directions. REEBLRA recovers the same five-way pattern. More strikingly, the full-score derivatives near $R_0=1$ are also of comparable magnitude (approximately 9--10 in absolute value), reproducing the mechanistic symmetry even though the fitted feature expansion contains nonlinear interactions (Table~\ref{tab:vector-borne-concordance}). This provides an interpretable diagnostic for multi-host systems: the method can reveal that host-to-vector transmission, vector-to-host transmission, vector abundance, human removal, and vector removal enter the threshold with broadly balanced log-scale influence, while still allowing interaction terms to improve classification. For users, the result offers a compact threshold map of a two-population transmission cycle and a direct way to check whether a learned surrogate respects known biological directionality.

Table~\ref{tab:vector-borne-five-examples} {illustrates recovery of $R_0$ brackets} rather than pointwise numerical values.

\section{Discussion}\label{sec:discussion}

REEBLRA learns inspectable {threshold-score} equations from sampled parameter--response pairs, including analytic reproduction-number evaluations, model simulations, or uploaded numerical summaries. These equations approximate events $R>c$; they are not analytic expressions for the magnitude of $R_0$. 

Its principal advantage is interpretability: each binary equation specifies how fitted features alter the log-odds of crossing a chosen threshold. The successive multiclass construction extends this idea across ordered $R_0$ levels without treating brackets as unrelated nominal classes. The coefficient matrix and its row-normalized heatmap (Figure~\ref{fig:multiclass_r0_coefficients}) further reveal whether the relative direction and prominence of parameters change as the threshold rises. Moreover, the workflow also brings together operations that are often conducted separately: LHS parameter exploration, growth-based threshold summaries, ODE simulation, Bayesian fitting, sensitivity summaries, and diagnostics. 

\begin{figure}[!htbp]
    \centering
    \includegraphics[width=1\textwidth]
    {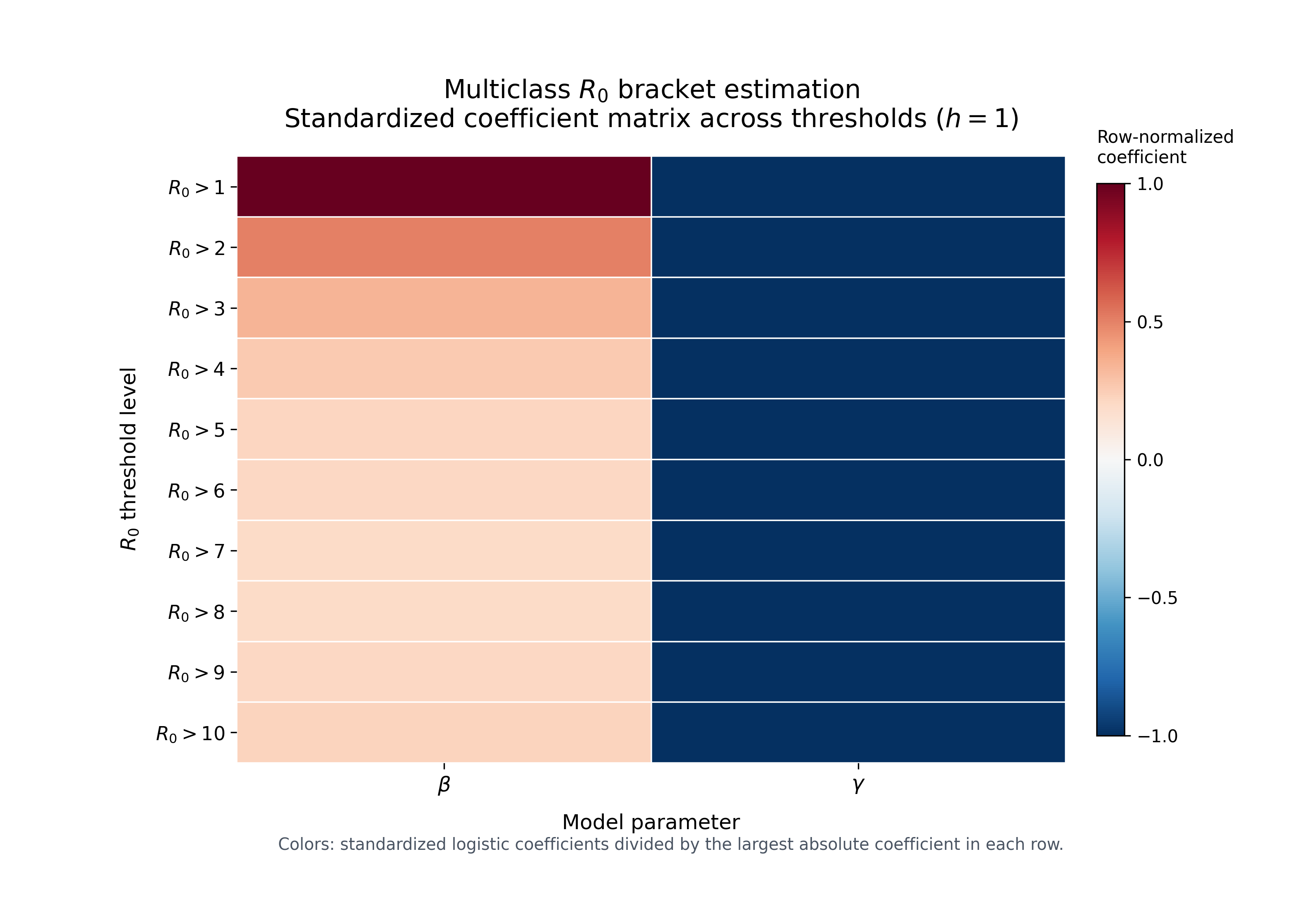}
    \caption{\textbf{Multiclass $R_0$ coefficient patterns across thresholds ($h=1$).} This independent heatmap uses the MAP/Laplace routine with linear $\beta$--$\gamma$ terms. Rows represent successive $R_0$ thresholds, while columns represent the transmission rate ($\beta$) and recovery/removal rate ($\gamma$) in the closed SIR model. Cell colors indicate the direction and relative prominence of standardized logistic coefficients, normalized independently within each row by the largest absolute coefficient, with red indicating positive and blue indicating negative associations with threshold exceedance. Across all thresholds, $\beta$ remains positively associated with $R_0$, whereas $\gamma$ remains negatively associated, consistent with $R_0=\beta/\gamma$. As the threshold increases, the relative prominence of $\beta$ generally decreases, while the negative coefficient of $\gamma$ becomes increasingly dominant, although minor variations occur at higher thresholds. Biologically, this suggests that, within the simulated parameter space, lower recovery/removal rates become relatively more prominent in distinguishing high-$R_0$ regimes, whereas increased transmission remains an important positive contributor. These changes reflect the threshold-dependent behavior of the fitted classifier rather than changes in the underlying SIR mechanism or absolute parameter sensitivities. Row-wise normalization precludes direct comparison of absolute coefficient magnitudes across thresholds.}
    \label{fig:multiclass_r0_coefficients}
\end{figure}

REEBLRA remains a surrogate-modelling framework. The fitted score $\eta(\boldsymbol\theta)$ approximates a threshold surface inside the sampled domain; it is not the spectral radius of a next-generation operator. The term ``polynomial'' in the Shiny app title refers to the default first-order logit polynomial. Without explicitly supplied transformations or interactions, the present fitted boundary is linear in standardized parameter space. Strong nonlinearity, parameter interactions, discontinuous interventions, multiple attractors, stochastic dynamics, or broad extrapolation may therefore require richer surrogate families or a different analysis.

Growth-rate conversion requires equal care. The SIR-like, SEIR-like, and generation-time equations encode different assumptions and should not be interpreted as exact mechanistic $R_0$ conversion. Results can change with the target signal, fitting window, model-time unit, removal or latency rate, initial conditions, and simulation horizon. Growth-fit confidence intervals and app warnings should be reported rather than treating every finite \texttt{R0estimate} as equally reliable.

Coefficient interpretation also depends on identifiability and class support. High VIF or strong parameter correlations can destabilize signs and magnitudes. One-class binary outcomes are not estimable, and one-sided multiclass thresholds correspond only to constant exceedance probabilities. Rare brackets may produce quasi-separation, very large coefficients, unstable sensitivity or precision, and uncertain one-vs-rest AUC. The row-normalized heatmap is useful for within-threshold patterns but does not preserve absolute magnitudes between rows. Apparent metrics should therefore be accompanied by held-out metrics whenever feasible, together with class counts, raw coefficients, and collinearity diagnostics.

Across all seven built-in models, high held-out performance persisted after feature construction was confined to the training folds. The mechanistic formulas then provide an external sign check: transmission-enhancing directions should raise the threshold score, while removal-enhancing directions should lower it. Under transformed bases, the correct object for magnitude interpretation is the combined local derivative of the fitted score with respect to an original parameter, not any single polynomial coefficient. This preserves the geometric concordance idea without confusing feature-term coefficients with literal algebraic exponents in the exact $R_0$ formula.

Class imbalance and sparse ordered brackets remain important. A high discrimination measure can coexist with weaker threshold-dependent classification diagnostics. REEBLRA should therefore be used with parameter ranges that deliberately explore both sides of scientifically relevant thresholds, and its high-level summaries should be read together with the full diagnostic report when making substantive claims. This follows broader recommendations for transparent prediction-model performance assessment \citep{steyerberg2010,saito2015,vancalster2019}.

The high held-out AUC and bracket accuracy are conditional on the chosen parameter domains, known analytic labels, supported brackets, feature basis, and cross-validation split. They do not establish an advantage over direct-$R_0$ regression, ordinal regression, or flexible classification. Additional tests should quantify errors near $R=1$, assess calibration and bracket support, examine sample-efficiency curves and feature-basis sensitivity, and measure the frequency and magnitude of threshold-order violations before projection.

The Bayesian formulation provides posterior uncertainty for coefficients and sensitivity summaries, while proper weakly informative priors regularize estimates under sparse data, correlated predictors, and complete or quasi-complete separation. Under separation, classical logistic maximum-likelihood estimates may fail to exist as finite values \citep{albert1984,gelman2008}. However, Bayesian regularization does not eliminate prior sensitivity, particularly when outcome classes are rare or perfectly separated. Therefore, class counts, held-out performance, coefficient uncertainty, and prior sensitivity should be examined. Firth's bias-reduced logistic regression provides an established alternative \citep{firth1993,heinze2002,mansournia2018}.

Bayesian sampling and mean-field variational inference serve different purposes. MCMC is more appropriate when posterior uncertainty and convergence diagnostics are central, but it is slower, especially when repeated across validation folds. Variational Bayes supports rapid exploratory fitting but may understate uncertainty. The ordinary logistic equation offers a transparent maximum-likelihood comparison, yet agreement between the two fitted equations does not by itself establish mechanistic validity or out-of-sample performance.

Several methodological questions are important for interpreting REEBLRA correctly (refer to Section~\ref{subsec:reviewer-questions} in the Supplementary Information). These points clarify the novelty claim of our proposed method. REEBLRA is not novel because logistic regression itself is new, nor because machine learning can approximate epidemic models. The proposed contribution is the combination of (i) threshold-surface learning rather than full $R_0$ regression, (ii) Bayesian regularization and uncertainty for interpretable symbolic threshold equations, (iii) successive ordered thresholds for bracket recovery, and (iv) a mechanistic--surrogate concordance principle that makes the learned coefficient geometry testable against the gradient of a known reproduction-number surface. The implementation serves as a reproducible demonstration of this methodological framework.

Potential extensions and future versions of REEBLRA include comprehensive automated screening of parameter interactions and transformations, automated separation diagnostics and bias-reduced estimation for ordered-threshold models, sparse equation discovery, Gaussian-process classification, adaptive sampling near epidemic threshold boundaries, propagation of uncertainty from growth-rate estimation, and automated comparison of mechanistically derived and surrogate threshold surfaces. Another important direction is the development of a parsimonious, unified mathematical representation of the multiple threshold equations generated by the sequential ordered multiclass bracket REEBLRA model, potentially reducing them to a single interpretable equation or functional relationship while preserving their predictive accuracy and threshold-ordering properties. Extending REEBLRA to stochastic and agent-based simulators would further broaden its applicability, allowing the framework to accommodate simulation variability, population heterogeneity, and more complex epidemic dynamics.

\section{Conclusion}\label{sec:conclusion}

REEBLRA reframes reproduction-number approximation as threshold-surface learning. This is the central methodological contribution: instead of requiring a surrogate to reproduce the full nonlinear magnitude of $R_0(\boldsymbol\theta)$, the method learns interpretable Bayesian equations for the scientifically relevant events $R_0>c$ and, through successive cutoffs, ordered reproduction-number brackets. The mechanistic--surrogate concordance principle then makes the learned geometry testable: coefficient direction can be compared with the standardized local gradient of the mechanistic reproduction-number surface. Across seven built-in models, binary five-fold AUC ranged from 0.999 to approximately 1.000 and balanced accuracy from 0.974 to 0.999; ordered-bracket accuracy ranged from 0.951 to 0.998 and balanced accuracy from 0.922 to 0.992.

The advantages of the framework are interpretability, explicit formulation, threshold-focused predictive efficiency, posterior uncertainty, compatibility with complex simulators, and a direct route from machine-learned coefficients back to mechanistic hypotheses. Its limitations are equally important: performance is conditional on the simulated domain and chosen cutoff; feature-basis and prior choices matter; collinearity can preserve prediction while undermining individual coefficient interpretation; numerical target uncertainty is only partly propagated; and extrapolation or mechanistic equivalence is not guaranteed. The accompanying Shiny application demonstrates these operations. REEBLRA is therefore best viewed as a mechanism-aware surrogate and diagnostic layer that complements---rather than replaces---next-generation-matrix derivation, direct $R_0$ regression, and other mechanistic analyses.

\backmatter

\section*{Supplementary information}
Supplementary Files contain (i) SI additional details, (ii) the Python translation reproducibility package, and (iii) the Shiny source code. The Shiny app includes a link to the user manual.

\section*{Shiny application}\label{sec:code-listings}

The implementation is an R Shiny app. Its principal computational modules cover built-in ODE definitions, custom ODE and auxiliary-equation parsing, nonnegative-state safeguards, LHS, \texttt{lsoda} simulation, analytic $R_0$ or $R_e$ where available, growth-rate fitting and conversion, upload cleaning, CSV generation, predictor standardization, Bayesian and ordinary binary logistic regression, successive multiclass threshold modeling, monotonic probability post-processing, cross-validation, coefficient-matrix construction, diagnostics, plotting, and export.

\noindent Link to the online Shiny app: \url{https://jfrabajanteuplb.shinyapps.io/REEBLRA}

\section*{Acknowledgements}
The authors gratefully acknowledge Ariel L. Babierra and Allen L. Nazareno for their valuable insights and constructive suggestions that contributed to the development and refinement of this work. The authors also extend their appreciation to Borame Sue Lee Dickens of the Saw Swee Hock School of Public Health, National University of Singapore, for providing inspiration that helped motivate this research.

\section*{Declarations}

\begin{itemize}

\item \textbf{Funding:} This research received no specific grant from any funding agency in the public, commercial, or not-for-profit sectors.

\item \textbf{Conflict of interest/Competing interests:} The authors declare no competing interests.

\item \textbf{Ethics approval, consent to participate, and consent for publication:} Not applicable. This study used only synthetic and simulated data and did not involve human participants, human biological materials, or individual-level human data.

\item \textbf{Data availability:} All results reported in this article were generated from synthetic parameter samples, simulated data, and analytic reproduction-number threshold formulas. No individual-level human data were used. The information required to reproduce the numerical experiments is provided in the article and Supplementary Information.

\item \textbf{Code availability:} The computational procedures, algorithms, and code necessary to reproduce the analyses are provided in the Supplementary Information and accompanying supplementary files.

\item \textbf{Author contributions:} J.F.R. conceptualized and designed the study, developed and implemented the computational framework and code, formulated and analyzed the algorithm, conducted the main analyses, and supervised the overall development of the work. N.P.B., M.C.T.L.-A., D.S.M.L., and E.R.R. contributed to the conceptual development of the study, analysis and interpretation of the methodology and results, and testing and evaluation of the algorithm. All authors contributed to writing, reviewing, and revising the manuscript, approved the final version, and agreed to its submission and publication.

\item \textbf{Generative AI use:} GPT-5.6 Sol (High reasoning mode) was used as an assistive tool for grammar checking, initial language editing of the manuscript, and code optimization. The scientific concepts, methodological decisions, analyses, interpretations, and conclusions were developed and evaluated by the authors. All authors reviewed and vetted the AI-assisted revisions and computational code and take full responsibility for the accuracy, integrity, and final content of the work.

\end{itemize}

\begin{appendices}

\section{Supplementary Information: Additional Details}

\subsection{Symbols and notation}\label{app:symbols}

Table~\ref{tab:symbols} consolidates the symbols used in the main text. 

\noindent Note on notation: Some symbols used in this manuscript have multiple meanings because they follow established conventions in their respective scientific and mathematical disciplines. To maintain consistency with these widely accepted conventions, we have retained the standard notation rather than introducing alternative symbols. Consequently, the meaning of a symbol should be interpreted according to its context. For example, `R' may refer to a reproduction-number-related quantity, the recovered or removed compartment in an epidemiological model, the R programming language, or a statistical measure such as the coefficient of determination (\(R^2\)). 

\begin{longtable}{@{}p{0.19\textwidth}p{0.49\textwidth}p{0.24\textwidth}@{}}
\caption{Symbols and notation used in the REEBLRA framework and model demonstrations.}\label{tab:symbols}\\
\toprule
Symbol & Definition & Context \\
\midrule
\endfirsthead
\toprule
Symbol & Definition & Context \\
\midrule
\endhead
$R_0$ & Basic reproduction number at the disease-free state. & Mechanistic benchmark. \\
$R_e(0)$ or $R(0)$ & Effective/control reproduction number or generic reproduction-number target at the start of the simulation. & Used when interventions or non-DFE susceptible fractions are included. \\
$\widehat R$ / \texttt{R0estimate} & Numerical reproduction-number summary obtained from early-growth information or uploaded externally. & Numerical target option. \\
$\boldsymbol\theta=(\theta_1,\ldots,\theta_p)$ & Vector of original epidemiological/model parameters. & REEBLRA input space. \\
$c$; $c_k$ & Binary threshold; $k$th successive threshold. & Events $R>c$ and $R>c_k$. \\
$Y_c$ & Binary indicator $\mathbf 1\{R>c\}$. & Binary REEBLRA outcome. \\
$p_c(\boldsymbol\theta)$ & $\Pr(R>c\mid\boldsymbol\theta)$. & Bayesian binary probability. \\
$\eta_c$ & Logistic linear/additive score for threshold $c$. & $p_c=\operatorname{logit}^{-1}(\eta_c)$. \\
$a_c$ & Intercept of the threshold-$c$ REEBLRA equation. & Binary/multiclass score. \\
$\phi_j(\boldsymbol\theta)$ & $j$th declared feature term constructed from the original parameters (raw, log, power, or pairwise log interaction). & Feature map; not the mechanistic $R$ formula. \\
$b_{cj}$ & Fitted coefficient multiplying feature term $\phi_j$ in the threshold-$c$ score. & Sign/magnitude interpreted conditionally on the chosen basis. \\
$x_j$ & Generic fitted feature value before standardization. & Regression design matrix. \\
$z_j$ & Standardized feature, $z_j=(x_j-\bar x_j)/s_j$. & Bayesian/ordinary logistic fitting. \\
$\bar x_j$, $s_j$ & Training-fold mean and sample standard deviation of feature $x_j$. & Fold-specific standardization. \\
$q_k$ & Estimated cumulative exceedance probability $\Pr(R>c_k\mid\boldsymbol\theta)$. & Ordered-bracket REEBLRA. \\
$h$ & Spacing between successive reproduction-number thresholds. & Python benchmark uses $h=5$. \\
$D_j$ & Median local full-score derivative $\partial\eta/\partial\log\theta_j$ over sampled points with $0.8\le R\le1.2$. & Mechanistic--surrogate concordance diagnostic. \\
$S,E,I,R,V$ & Susceptible, exposed, infectious, removed/recovered, and vaccinated compartments. & Compartmental models. \\
$S_i,I_i,R_i$ & SIR states for group $i$. & Structured SIR models. \\
$S_h,I_h,R_h$ & Human SIR states. & Vector-borne model. \\
$S_v,I_v$ & Susceptible and infectious vector states. & Vector-borne model. \\
$\mathcal{C}(t)$ & Cumulative incidence state. & ODE simulation output. \\
$N$ & Total host population or normalized population mass. & Force-of-infection denominator. \\
$\beta$ & Effective transmission/contact rate. & SIR/SEIR and structured models. \\
$\gamma$ & Infectious removal/recovery rate. & SIR/SEIR models. \\
$\mu$ & Per-capita demographic turnover/natural mortality rate. & Open models. \\
$\sigma$ & Exposed-to-infectious progression rate. & SEIR models. \\
$v$ & Vaccination rate/pressure applied to susceptible individuals. & Vaccination model. \\
$c_{ij}$ & Contact coefficient linking infectious group $j$ to susceptible group $i$. & Two-group SIR. \\
$\mathbf C=(C_{ij})$ & Fixed $5\times5$ mixing matrix. & Fifteen-compartment five-group SIR. \\
$\gamma_i$ & Recovery rate in group $i$. & Structured SIR models. \\
$\mathbf K$ & Next-generation matrix. & $R_0=\rho(\mathbf K)$ in structured models. \\
$\rho(\mathbf K)$ & Spectral radius of the next-generation matrix. & Structured-model reproduction number. \\
$\beta_{hv}$ & Human-to-vector transmission rate. & Vector-borne model. \\
$\beta_{vh}$ & Vector-to-human transmission rate. & Vector-borne model. \\
$\gamma_h$ & Human infectious removal rate. & Vector-borne model. \\
$\mu_v$ & Vector death/removal rate. & Vector-borne model. \\
$m$ & Vector-to-human ratio. & Vector-borne model. \\
$r$ & Early exponential growth rate estimated from $\log I(t)=\alpha+rt$. & Numerical \texttt{R0estimate}. \\
$k$ & Number of cross-validation folds. & Python benchmark uses $k=5$. \\
$s_j=\tanh(b_j)$ & Bounded coefficient-based sensitivity summary. & App sensitivity tab; not a mechanistic elasticity. \\
\bottomrule
\end{longtable}

\subsection{Bayesian logistic regression}\label{subsec:why-bayesian}

The Bayesian formulation is used for three reasons. First, it provides posterior uncertainty for coefficients and derived sensitivity summaries, rather than only point estimates. Second, proper weakly informative priors regularize standardized coefficients and reduce instability when the simulated design is sparse, highly predictive, or nearly separated. Classical logistic maximum-likelihood estimates may not exist as finite values under complete or quasi-complete separation, whereas weakly informative Bayesian priors can regularize estimation and yield stable, finite posterior estimates in such settings. Third, Bayesian regularization is useful when transformed feature bases introduce several correlated terms and the likelihood alone is weakly informative.

In our implementation, standardized feature coefficients and the intercept receive proper normal priors centered at zero with nominal scale 2, with the app using the corresponding \texttt{rstanarm} autoscaling option. These priors are intentionally weak enough to allow strong effects while preventing the infinite-coefficient behavior of unpenalized maximum likelihood. Separation is therefore not ignored: the ordinary logistic comparator may show very large or unstable coefficients, whereas the Bayesian fit remains finite. Nevertheless, a finite Bayesian posterior does not make separation harmless. If one class is rare or perfectly predictable, the posterior can become prior-sensitive; class counts, held-out performance, coefficient uncertainty, and prior sensitivity should then be reported. Bias-reduced/Firth logistic regression is a well-established alternative for separation.

Multicollinearity requires a separate distinction. Priors can stabilize prediction and partially shrink correlated coefficients, but they do not restore structural identifiability when predictors carry nearly redundant information. REEBLRA therefore reports VIFs and the predictor correlation matrix, and coefficient-level mechanistic interpretation is discouraged when these diagnostics are severe. Remedies include redesigning the parameter ranges, combining mechanistically redundant parameters, reparameterizing the model, using stronger regularizing or hierarchical priors, or interpreting groups of correlated features rather than individual coefficients. In short, Bayesian regularization addresses estimation instability; it does not magically remove nonidentifiability.

\subsection{Shiny app}

\subsubsection{Built-in ODE simulation and Latin hypercube sampling}

For built-in ODE simulation, each parameter is assigned an editable minimum and maximum. REEBLRA uses random Latin hypercube sampling (LHS) to distribute parameter combinations across the resulting hyper-rectangle. The interface defaults to 400 simulations for quick runs and displays a screening recommendation of approximately $400p$ simulations for $p$ parameters. This recommendation is heuristic rather than a formal sample-size theorem; required simulation size depends on class balance, threshold complexity, validation demands, and computation time. The default horizon is 365 model-time units, the default output step is 1, the default initial infectious proportion is 0.01, and the default case-increase tolerance is $10^{-5}$.

The case-increase response is
\[
Y_{\mathrm{inc}}
=
\mathbf{1}\!\left\{
\max_t I_{\mathrm{major}}(t)
>
I_{\mathrm{major}}(0)+\varepsilon
\right\},
\]
where $\varepsilon$ is the selected tolerance. This operational definition prevents very small numerical changes from automatically being treated as meaningful growth.

\subsubsection{Custom ODE simulation, auxiliary equations, and nonnegative states}

The custom parser accepts up to 30 state equations and 50 parameters. State-equation left sides may be entered as \texttt{dx/dt}, \texttt{dI/dt}, \texttt{dx}, \texttt{dI}, or the state name itself. Right sides use R-style syntax. Initial values and parameter ranges are supplied as named lines. Unless declared explicitly, \texttt{N} is set to the sum of the current state values, and \texttt{t} and \texttt{time} are available during evaluation.

Up to 30 optional non-ODE equations may define derived quantities, totals, time-varying forcing terms, or intervention indicators. These auxiliary values are evaluated during ODE simulation and are added to the returned trajectory. Thus, when one host species has several infectious classes, a user may define, for example,
\[
I_{\mathrm{total}}=I_1+I_2+I_3,
\]
and select $I_{\mathrm{total}}$ as the early-growth target. The target may therefore be a simulated state or a valid auxiliary output.

The user can select which state variables are forced to be nonnegative. At a selected state's lower boundary, a negative derivative is replaced by zero, and negative numerical output values are clamped to zero. This is a numerical safeguard; it does not replace an analytic proof that the mathematical model is positively invariant.

\subsubsection{Uploaded data and CSV generation}

An uploaded table must contain an $R_0$-like numeric column selected as \texttt{R0estimate} and at least one numeric parameter column. The implementation requires at least 20 complete observations, and enforces a suggested maximum of 50 numeric parameter predictors for stable fitting and interpretation. The uploaded value may have been obtained from a numerical NGM spectral radius, early-growth conversion, or another method justified by the user.

The minimum of 20 observations is used only as a technical safeguard to prevent fitting logistic models to trivially small datasets. It is not a formal sample-size recommendation, and substantially larger samples may be needed depending on class balance, the number of fitted feature terms, separation, and the validation design. The suggested maximum of 50 predictors is likewise a practical limit rather than a theoretical one, intended to reduce computational burden, excessive feature expansion, collinearity, and difficulty in interpreting individual coefficients. Analyses with more predictors may require prior screening, dimension reduction, or stronger regularization.

The CSV generator is intended for users who have trajectories but not yet a parameter--\texttt{R0estimate} table. Its time-series file contains time in the first column and one epidemic signal per remaining column. A second file contains one parameter row for each trajectory. Growth windows may be selected automatically, by the first $n$ positive points, by a time interval, by a fraction-of-peak rule, or from all positive points. The generated table and a separate diagnostics table can be downloaded or passed directly to the main REEBLRA analyses.

\subsubsection{Numerical growth-rate \texttt{R0estimate}}\label{sec:growth-r0}

REEBLRA estimates early exponential growth by fitting
\[
\log Y(t)=a+rt
\]
to a positive infectious-state trajectory or selected output. In built-in models, $Y(t)=I_{\mathrm{major}}(t)$ and the candidate window is restricted by
\[
S_{\mathrm{major}}(t)\ge \tau S_{\mathrm{major}}(0),
\]
where $\tau=0.95$ by default. At least five positive infectious observations are required. If the requested window is too short, the app progressively relaxes $\tau$ through lower candidate values; if necessary, it falls back to the earliest positive observations. The output includes $\hat r$, a 95\% confidence interval, $R^2$, number of points, threshold used, fitting-window endpoint, a validity flag, and warnings. A confidence interval crossing zero or $R^2<0.80$ is explicitly flagged.

Custom ODE mode fits the first user-selected number of positive target observations, with eight as the default and five as the minimum. The CSV generator provides automatic and user-defined window rules. In its automatic mode, candidate early windows up to the first peak are compared using a score favoring positive growth, high log-linear $R^2$, and earlier windows.

The available conversion rules are:
\begin{align}
\text{SIR-like:}\qquad
\widehat R_{0,\mathrm{growth}}
&=1+\frac{r}{\gamma_{\mathrm{removal}}},\\
\text{SEIR-like:}\qquad
\widehat R_{0,\mathrm{growth}}
&=\frac{(r+\sigma_{\mathrm{latency}})(r+\gamma_{\mathrm{removal}})}
{\sigma_{\mathrm{latency}}\gamma_{\mathrm{removal}}},\\
\text{generation-time approximation:}\qquad
\widehat R_{0,\mathrm{growth}}
&=\exp(rT_{\mathrm{gen}}),\\
\text{unit-time Euler multiplier:}\qquad
\widehat R_{0,\mathrm{growth}}
&=1+r.
\end{align}
The last rule preserves the threshold equivalence $r>0\Longleftrightarrow \widehat R_0>1$ for one model-time unit, but it is not a mechanistic generation-interval conversion and is time-scale dependent. For built-in simulations, a non-positive $r$ produces \texttt{R0estimate}=0 by convention. In custom and CSV-generator modes, the user may instead set non-positive growth to 0, 1, missing, or the unmodified formula value. All growth-based values are numerical summaries and should not be presented as exact NGM reproduction numbers without further justification.

\subsubsection{Collinearity and parameter sensitivity}\label{sec:collinearity-sensitivity}

The collinearity module reports variance inflation factors (VIFs), a parameter correlation matrix, and a correlation heatmap. A VIF near 1 indicates little collinearity; values around 2.5--5 suggest moderate collinearity; values above 5 indicate high collinearity; and values above 10 indicate serious collinearity. These are interpretation guides rather than universal rejection thresholds. Collinearity may leave prediction useful while making individual coefficient signs and magnitudes unstable, a well-known distinction between predictive adequacy and coefficient interpretability. The current normal priors provide shrinkage but do not eliminate this identifiability problem. If strong collinearity persists, the preferred response is mechanistic reparameterization, aggregation of redundant quantities, redesign of the sampled domain, or a stronger structured prior rather than interpreting unstable individual coefficients.

For each posterior draw of a standardized Bayesian feature coefficient $b_j$, REEBLRA computes
\[
s_j=\tanh(b_j),\qquad -1<s_j<1.
\]
The sensitivity table and plot report the posterior mean, median, and 95\% credible interval of $s_j$. Under a transformed basis these are feature-term sensitivities rather than one-to-one parameter effects. Positive values indicate increased fitted probability of threshold exceedance, negative values indicate decreased fitted probability, and values near zero indicate weak fitted influence on the standardized logit scale. This bounded transformation is not a mechanistic derivative, elasticity, partial rank correlation coefficient, or causal effect. It is a coefficient-based summary of the fitted binary Bayesian surrogate.

\subsubsection{Cross-language reproducibility benchmark}\label{sec:python-benchmark}

To check the calculations independently of the R Shiny implementation, we translated the built-in parameter ranges, analytic $R_0/R_e$ formulas, Latin-hypercube sampling, predictor standardization, logistic threshold equations, ordered successive-threshold construction, and normalized sensitivity transformation into Python.  

The interactive Shiny multiclass module fits independent, {unpenalized} frequentist threshold equations using \texttt{stats::glm}. In contrast, the supplied Python reproducibility script uses \texttt{fit\_laplace} for {both binary and ordered-threshold fits}: a logistic maximum-a-posteriori (MAP) estimate under independent $N(0,2^2)$ priors on the intercept and standardized coefficients, optimized with L-BFGS-B, with a Gaussian Laplace covariance approximation. The Python ordered predictions use plug-in MAP scores rather than posterior-predictive integration. These procedures are not identical to the Shiny multiclass GLM; neither fits a jointly constrained Bayesian ordinal-regression model. The Shiny binary Bayesian module instead uses \texttt{rstanarm}. Thus, the Python results are an independent numerical cross-check rather than a claim of draw-for-draw identity with Stan-based inference.

The final reproducibility benchmark contains seven built-in models, including a 15-compartment five-group SIR system. Binary demonstrations use 5,000 Latin-hypercube parameter sets and five-fold held-out prediction. A linear basis is sufficient for the two simple SIR benchmarks, while the model-aware option uses a generic compact log-cubic expansion for the curved positive-parameter threshold surfaces. Ordered-bracket demonstrations use 8,000 independent parameter sets, the generic compact log-cubic basis, bracket step $h=5$, and five-fold validation. 

The feature expansion is generic and does not contain the analytic $R_0/R_e$ formula. Feature construction and standardization are performed inside each training fold, preventing validation leakage.

The probability cutoff is set to 0.5. All binary accuracy, balanced accuracy, AUC, and F1 values in the benchmark exceeded 0.90. All ordered-bracket overall accuracy, balanced accuracy, Cohen's $\kappa$, macro F1, and weighted F1 values also exceeded 0.90. 

The ODE model demonstrations in Section~\ref{sec:models} report the fitted binary equation, representative ordered-threshold equations, and metrics using the Python code. 

\subsubsection{Application outputs and reporting}\label{sec:outputs}

The main interface contains the following tabs: Overview, ODE equations, Parameter ranges, Simulation results, Collinearity analysis, Binary Bayesian $R_0$ threshold, Multiclass $R_0$ bracket, Parameter sensitivity analysis, CSV generator, and Downloads.

The binary main-output tab reports the Bayesian standardized equation and the corresponding unstandardized fitted-feature equation (which is an original-parameter-scale equation only under a linear basis), apparent and optional $k$-fold metrics (accuracy, balanced accuracy, AUC, sensitivity/recall, specificity, precision, F1, NPV, Cohen's $\kappa$, and Brier score), MCMC diagnostics when applicable, original-scale trial prediction, posterior coefficients, ROC curve, an interactive predicted-probability display, and ordinary logistic equations for comparison. When the outcome contains only one class, the fitting routine provides a specific warning rather than attempting an undefined logistic model.

The multiclass tab reports the successive-threshold method, all estimable standardized and unstandardized fitted-feature equations, explicit notices for thresholds with only one observed class, the row-normalized coefficient heatmap, raw coefficient table, original-scale trial prediction, overall accuracy/balanced accuracy/Cohen's $\kappa$, macro and weighted summaries, per-class recall, specificity, precision, F1, NPV, one-vs-rest balanced accuracy and AUC, class distribution, apparent and cross-validated one-vs-rest AUC, confusion matrix, and predicted-bracket display.

The CSV generator reports input previews, the generated \texttt{R0estimate} table, growth diagnostics, and warnings. Its outputs can be downloaded or transferred into the binary and multiclass workflows without a separate upload.

The Downloads tab provides a ZIP archive of available outputs and individual files. CSV exports include simulation or uploaded data, parameter ranges or summaries, parameter definitions, simulation summaries, analytic-versus-numeric $R_0$ comparison, VIFs, parameter correlations, binary metrics and predictions, held-out predictions, ROC data, posterior coefficients, MCMC diagnostics, sensitivity summaries, multiclass predictions, multiclass held-out predictions, multiclass metrics, coefficients, confusion matrix, class distribution, and apparent and cross-validated one-vs-rest AUC. PNG exports include the principal binary, multiclass, sensitivity, VIF, and correlation plots. Publication-ready LaTeX files are generated for the fitted binary and multiclass equations.

\section{Reproducibility specification for the main-text results}\label{app:reproducibility}

This appendix records the numerical inputs and settings needed to reproduce the main-text benchmark table, metric figure, fitted REEBLRA equations, and parameter-concordance summaries. Two computational paths should be distinguished. The R/Shiny application is the reference implementation of the interactive workflow. The main cross-language benchmark reported in Section~\ref{sec:python-benchmark} was generated independently in Python from the same analytic reproduction-number definitions and parameter domains. The latter does not numerically integrate the ODEs to create its binary labels; it samples parameters, evaluates the analytic $R_0/R_e$, and learns the corresponding threshold surface. The ODE defaults in Table~\ref{tab:app-ode-defaults} are supplied so that the built-in Shiny simulations and growth-rate outputs can also be reproduced.

\subsection*{Files and global benchmark settings}

\begin{table}[H]
\centering
\small
\caption{Global settings used for the main cross-language benchmark.}
\label{tab:global-repro-settings}
\begin{tabularx}{\textwidth}{@{}lX@{}}
\toprule
Setting & Value \\
\midrule
Reproducibility script & \texttt{reeblra\_python\_reproducibility.py} (included in the revised source package) \\
Binary sample size & 5,000 independent Latin-hypercube parameter sets per model \\
Ordered-bracket sample size & 8,000 independent Latin-hypercube parameter sets per model \\
Cross-validation & Stratified five-fold; feature construction and standardization repeated within each training fold \\
Binary target & $Y=\mathbf 1\{R>1\}$ \\
Binary probability cutoff & 0.5 \\
Binary prior in Python check & Independent $N(0,2^2)$ Gaussian prior on the fitted logit intercept and coefficients; Laplace/Gaussian posterior approximation \\
Binary feature basis & Linear for Closed SIR and Open SIR; compact log-cubic for the other five models \\
Ordered-bracket feature basis & Compact log-cubic for all seven models \\
Ordered thresholds & $c_k\in\{1,6,10\}$ (step $h=5$), yielding $<1$, $1$--$6$, $6$--$10$, and $>10$ brackets \\
Ordered-threshold fit & Separate Gaussian-prior logistic MAP/Laplace fits (independent $N(0,2^2)$ priors; L-BFGS-B optimizer), matching the Python binary fitting algorithm; cumulative probability monotonicization and differencing follow. The Shiny multiclass module instead uses unpenalized \texttt{stats::glm}. \\
Feature standardization & Training-fold mean and sample SD ($\mathrm{ddof}=1$); the held-out fold uses the training-fold constants \\
Main fitted equations & Full-data fit after validation; held-out metrics are calculated from pooled fold predictions \\
\bottomrule
\end{tabularx}
\end{table}

The compact log-cubic feature map for positive parameters is
\[
\phi(\boldsymbol\theta)=\left\{\log\theta_j,(\log\theta_j)^2,(\log\theta_j)^3,\log\theta_j\log\theta_\ell:\ j<\ell\right\}.
\]

No analytic $R_0/R_e$ expression, ratio, product, spectral radius, or model-specific mechanistic interaction is supplied as a predictor. The expansion is generic and is declared before validation.

\subsection*{Parameter domains and random seeds}

\begin{longtable}{@{}p{0.27\textwidth}p{0.20\textwidth}rr@{}}
\caption{Latin-hypercube parameter domains used in the Python benchmark.}\label{tab:benchmark-ranges}\\
\toprule
Model & Parameter & Minimum & Maximum \\
\midrule
\endfirsthead
\toprule
Model & Parameter & Minimum & Maximum \\
\midrule
\endhead
Closed SIR & $\beta$ & 0.05 & 1.00 \\
 & $\gamma$ & 0.05 & 1.00 \\
Open SIR & $\beta$ & 0.05 & 1.20 \\
 & $\gamma$ & 0.03 & 0.60 \\
 & $\mu$ & 0.005 & 0.25 \\
Open SEIR & $\beta$ & 0.05 & 1.50 \\
 & $\sigma$ & 0.05 & 1.00 \\
 & $\gamma$ & 0.03 & 0.60 \\
 & $\mu$ & 0.005 & 0.20 \\
SEIR with vaccination & $\beta$ & 0.10 & 4.00 \\
 & $\sigma$ & 0.05 & 1.00 \\
 & $\gamma$ & 0.03 & 0.60 \\
 & $\mu$ & 0.01 & 0.30 \\
 & $v$ & 0.001 & 0.15 \\
Age-structured SIR & $\beta$ & 0.03 & 0.80 \\
 & $\gamma_1$ & 0.05 & 0.50 \\
 & $\gamma_2$ & 0.05 & 0.50 \\
 & $c_{11}$ & 0.10 & 1.50 \\
 & $c_{12}$ & 0.05 & 1.00 \\
 & $c_{21}$ & 0.05 & 1.00 \\
 & $c_{22}$ & 0.10 & 1.50 \\
15-compartment five-group SIR & $\beta$ & 0.005 & 0.30 \\
 & $\gamma_1$ & 0.05 & 0.50 \\
 & $\gamma_2$ & 0.05 & 0.50 \\
 & $\gamma_3$ & 0.05 & 0.50 \\
 & $\gamma_4$ & 0.05 & 0.50 \\
 & $\gamma_5$ & 0.05 & 0.50 \\
Vector-borne & $\beta_{hv}$ & 0.05 & 1.50 \\
 & $\beta_{vh}$ & 0.05 & 1.50 \\
 & $\gamma_h$ & 0.05 & 0.60 \\
 & $\mu_v$ & 0.05 & 0.60 \\
 & $m$ & 0.50 & 6.00 \\
\bottomrule
\end{longtable}

\begin{table}[H]
\centering
\small
\caption{Model order and deterministic random seeds used by the Python benchmark.}
\label{tab:benchmark-seeds}
\begin{tabular}{@{}lrr@{}}
\toprule
Model & Binary LHS/CV seed & Ordered-bracket LHS/CV seed \\
\midrule
Closed SIR & 123 & 1700 \\
Open SIR & 124 & 1701 \\
Open SEIR & 125 & 1702 \\
SEIR with vaccination & 126 & 1703 \\
Age-structured SIR & 127 & 1704 \\
15-compartment five-group SIR & 128 & 1705 \\
Vector-borne & 129 & 1706 \\
\bottomrule
\end{tabular}
\end{table}

\subsection*{Built-in ODE simulation defaults in the Shiny implementation}

\begin{table}[H]
\centering
\small
\caption{Global built-in ODE defaults in the current Shiny implementation.}
\label{tab:app-ode-defaults}
\begin{tabularx}{\textwidth}{@{}lX@{}}
\toprule
Input & Default \\
\midrule
LHS simulation count & 400 \\
ODE solver & \texttt{deSolve::ode} with \texttt{method="lsoda"} \\
Time horizon & 365 time units \\
Output time step & 1 time unit \\
Initial infectious proportion & $I_0=0.01$ \\
Early-growth susceptible threshold & 0.95 of the initial major-host susceptible level; relaxed stepwise only when fewer than five positive infectious points are available \\
Case-increase tolerance & $10^{-5}$ \\
Growth fit & Linear regression of $\log I_{\rm major}(t)$ on time over the selected early-growth window \\
SIR-like conversion & $\widehat R=1+r/\gamma_{\rm removal}$ \\
SEIR-like conversion & $\widehat R=(r+\sigma_{\rm latency})(r+\gamma_{\rm removal})/(\sigma_{\rm latency}\gamma_{\rm removal})$ \\
Non-growing built-in convention & Growth-based $\widehat R<1$ is set to 0 to flag a non-growing/subthreshold run \\
\bottomrule
\end{tabularx}
\end{table}

\begin{table}[H]
\centering
\scriptsize
\caption{Built-in early-growth targets and growth-rate conversion quantities used by the Shiny implementation.}
\label{tab:growth-conversion-map}
\begin{tabularx}{\textwidth}{@{}lXXX@{}}
\toprule
Model & Major infectious/susceptible output & Removal rate used & Latency rate / conversion \\
\midrule
Closed SIR & $I/S$ & $\gamma$ & SIR-like \\
Open SIR & $I/S$ & $\gamma+\mu$ & SIR-like \\
Open SEIR & $I/S$ & $\gamma+\mu$ & $\sigma+\mu$; SEIR-like \\
SEIR with vaccination & $I/S$ & $\gamma+\mu$ & $\sigma+\mu$; SEIR-like \\
Age-structured SIR & $I_{\rm total}/S_{\rm total}$ & $(\gamma_1+\gamma_2)/2$ & SIR-like \\
15-compartment five-group SIR & $I_{\rm total}/S_{\rm total}$ & $(\gamma_1+\cdots+\gamma_5)/5$ & SIR-like \\
Vector-borne & $I_h/S_h$ & $\gamma_h$ & SIR-like host-growth approximation \\
\bottomrule
\end{tabularx}
\end{table}

The built-in initial states at $I_0=0.01$ are: Closed/Open SIR, $(S,I,R)=(0.99,0.01,0)$; Open SEIR, $(S,E,I,R)=(0.99,0,0.01,0)$; vaccination SEIR, $(S,V,E,I,R)=(0.99,0,0,0.01,0)$; two-group SIR, $(S_i,I_i,R_i)=(0.495,0.005,0)$ for $i=1,2$; five-group SIR, $(S_i,I_i,R_i)=(0.198,0.002,0)$ for $i=1,\ldots,5$; and the vector-borne model, $(S_h,I_h,R_h,S_v,I_v)=(0.99,0.01,0,1,0)$. Cumulative-incidence states are initialized to the corresponding seeded host infection amount in the application.

\subsection{Accuracy and AUC metrics for the ODE demonstrations} 

Table~\ref{tab:cross-model-v4} shows the accuracy, balanced accuracy, and AUC of the binary and multiclass models for the seven ODE demonstrations.

\begin{table*}[ht]
\centering
\caption{Five-fold held-out REEBLRA benchmark across the seven built-in models. Binary models use the declared model-aware basis (linear for the two simple SIR systems and compact log-cubic for the curved higher-dimensional systems); ordered brackets use the generic compact log-cubic basis with step 5 and Bayesian Laplace threshold fits. Values are proportions.}
\label{tab:cross-model-v4}
\scriptsize
\setlength{\tabcolsep}{3pt}
\begin{tabular}{lrrrrrrr}
\toprule
Model & Binary acc. & Binary bal. acc. & Binary AUC & Multi acc. & Multi bal. acc. & Multi AUC & Multi $\kappa$ \\
\midrule
Closed SIR & 0.999 & 0.999 & 1.000 & 0.998 & 0.992 & 1.000 & 0.996 \\
Open SIR & 0.997 & 0.997 & 1.000 & 0.992 & 0.934 & 1.000 & 0.983 \\
Open SEIR & 0.995 & 0.994 & 1.000 & 0.989 & 0.935 & 1.000 & 0.979 \\
SEIR with vaccination & 0.991 & 0.990 & 1.000 & 0.983 & 0.922 & 0.999 & 0.965 \\
Age-structured SIR & 0.985 & 0.978 & 0.999 & 0.970 & 0.929 & 0.998 & 0.940 \\
15-compartment five-group SIR & 0.993 & 0.974 & 0.999 & 0.951 & 0.939 & 0.996 & 0.921 \\
Vector-borne & 0.997 & 0.983 & 1.000 & 0.993 & 0.985 & 1.000 & 0.986 \\
\bottomrule
\end{tabular}
\end{table*}

\pagebreak
\section{{Alternative feature bases} and finer ordered brackets}\label{app:alternate-equations}

The supplementary sensitivity analyses below repeat all seven analytic-label ODE benchmarks while changing one design choice at a time. First, we use untransformed parameter values in a standardized linear logistic score for both binary ($R>1$) and ordered brackets ($h=5$, as in the main text). Second, we retain exactly the compact log-cubic feature transformations of the main-text ordered-bracket benchmark and reduce the bracket increment to $h=0.5$. These scenarios use the same parameter ranges, sample sizes, design seeds, prior specification, five-fold held-out validation, and probabilistic monotonicization algorithm as the main reported benchmark. They are not independent epidemiological studies or validation on observed epidemic data.

Let $\phi(\boldsymbol\theta)$ be a declared feature vector and $z_j=(\phi_j-\bar\phi_j)/s_j$. At cutoff $c$, all fitted equations take the explicit form $\eta_c=a_c+\sum_j b_{cj}z_j$ and $q_c=\{1+\exp(-\eta_c)\}^{-1}$. The numerical values of $a_c$, $b_{cj}$, the feature ordering, and the normalization constants are printed for every ODE model and cutoff below; full-precision values and fold-held-out diagnostic tables are also exported as CSV. Ordered probabilities are $\pi_0=1-q_{c_1}$, $\pi_j=q_{c_j}-q_{c_{j+1}}$, and $\pi_K=q_{c_K}$ after non-increasing projection, with the most probable bracket predicted. The equations are threshold-exceedance classifiers, not analytic identities for $R_0$ or $R_e$.

\subsection{{Linear parameters} with standardization: binary and sequential h = 5}\label{app:linear_h5}

Standardized linear-parameter basis; main-text bracket width $h=5$. The underlying parameter tables and analytic labels are generated from independent Latin-hypercube designs for binary ($n=5000$, seed $123+i$) and ordered ($n=8000$, seed $1700+i$) fits, where $i=0,\ldots,6$ indexes the seven benchmark models. A proper independent coefficient prior $N(0,2^2)$ is used; MAP/Laplace logistic fits and training-fold-only standardization match the accompanying reproducibility Python. Reported scores are pooled five-fold held-out metrics. AUC in the ordered summaries is macro/weighted one-vs-rest and is undefined when the prescribed full class set is not represented; per-bracket AUC is reported only where positive and negative observations exist.

\subsubsection{Five-fold binary evaluation for raw standardized parameters}

For each of the seven binary fits, the threshold event is $R>1$, the decision threshold is 0.5, and the binary equations use only the original parameters, individually standardized on the training sample. Sensitivity, specificity, precision, negative predictive value, F1, and Cohen's $\kappa$ are defined against the $R>1$ positive class. The Brier score evaluates squared probability error.

\begingroup\scriptsize\setlength{\LTcapwidth}{\linewidth}\setlength{\tabcolsep}{3pt}

\begin{longtable}{@{}lrrrrrrrr@{}}

\caption{Binary five-fold held-out performance using standardized linear parameters. AUC denotes binary ROC AUC; Brier denotes mean squared probability error.}\label{tab:supp-linear-binary-metrics}\\

\toprule Model & Accuracy & Balanced acc. & AUC & Sens. & Spec. & F1 & $\kappa$ & Brier \\ \midrule\endfirsthead

\toprule Model & Accuracy & Balanced acc. & AUC & Sens. & Spec. & F1 & $\kappa$ & Brier \\ \midrule\endhead

 Closed SIR & 0.999 & 0.999 & 1.000 & 0.998 & 1.000 & 0.999 & 0.997 & 0.007 \\

 Open SIR & 0.997 & 0.997 & 1.000 & 0.997 & 0.997 & 0.998 & 0.993 & 0.007 \\

 Open SEIR & 0.952 & 0.946 & 0.989 & 0.963 & 0.929 & 0.963 & 0.893 & 0.035 \\

 SEIR with vaccination & 0.925 & 0.906 & 0.968 & 0.949 & 0.863 & 0.948 & 0.814 & 0.056 \\

 Age-structured SIR & 0.951 & 0.927 & 0.986 & 0.970 & 0.883 & 0.969 & 0.855 & 0.036 \\

 15-compartment five-group SIR & 0.974 & 0.923 & 0.994 & 0.986 & 0.861 & 0.986 & 0.848 & 0.018 \\

 Vector-borne & 0.974 & 0.817 & 0.984 & 0.992 & 0.643 & 0.987 & 0.701 & 0.019 \\

\bottomrule\end{longtable}\endgroup

\begingroup\scriptsize\setlength{\LTcapwidth}{\linewidth}\setlength{\tabcolsep}{3pt}

\begin{longtable}{@{}lrrrr@{}}

\caption{Additional binary classification diagnostics for the standardized linear basis.}\label{tab:supp-linear-binary-extra}\\

\toprule Model & Positive $n$ & Negative $n$ & Precision & NPV \\ \midrule\endfirsthead

\toprule Model & Positive $n$ & Negative $n$ & Precision & NPV \\ \midrule\endhead

 Closed SIR & 2480 & 2520 & 1.000 & 0.998 \\

 Open SIR & 3285 & 1715 & 0.998 & 0.995 \\

 Open SEIR & 3276 & 1724 & 0.963 & 0.930 \\

 SEIR with vaccination & 3590 & 1410 & 0.946 & 0.870 \\

 Age-structured SIR & 3921 & 1079 & 0.968 & 0.890 \\

 15-compartment five-group SIR & 4526 & 474 & 0.985 & 0.864 \\

 Vector-borne & 4751 & 249 & 0.981 & 0.804 \\

\bottomrule\end{longtable}\endgroup

\subsubsection{Five-fold ordered-bracket evaluation}

Exact bracket classification is assessed under the four-class partition $R\le 1$, $1<R\le 6$, $6<R\le 10$, and $R>10$. A change in the number and granularity of classes changes the difficulty of exact-bracket prediction; accuracy and macro F1 should therefore not be interpreted as direct apples-to-apples comparisons across these two scenarios.

\begingroup\scriptsize\setlength{\LTcapwidth}{\linewidth}\setlength{\tabcolsep}{3pt}

\begin{longtable}{@{}lrrrrrrr@{}}

\caption{Five-fold ordered-bracket classification metrics for linear $h=5$. AUC values are OvR; a dash indicates that at least one prescribed class is absent, making the full-class multiclass AUC undefined.}\label{tab:supp-linear_h5-ordered}\\

\toprule Model & Acc. & Bal. acc. & Macro F1 & Weighted F1 & Macro AUC & Weighted AUC & $\kappa$ \\ \midrule\endfirsthead

\toprule Model & Acc. & Bal. acc. & Macro F1 & Weighted F1 & Macro AUC & Weighted AUC & $\kappa$ \\ \midrule\endhead

 Closed SIR & 0.985 & 0.818 & 0.852 & 0.984 & 0.998 & 1.000 & 0.971 \\

 Open SIR & 0.991 & 0.837 & 0.875 & 0.991 & 0.999 & 1.000 & 0.982 \\

 Open SEIR & 0.940 & 0.844 & 0.875 & 0.940 & 0.993 & 0.988 & 0.881 \\

 SEIR with vaccination & 0.894 & 0.758 & 0.788 & 0.893 & 0.978 & 0.962 & 0.779 \\

 Age-structured SIR & 0.867 & 0.646 & 0.681 & 0.860 & 0.950 & 0.948 & 0.720 \\

 15-compartment five-group SIR & 0.767 & 0.714 & 0.724 & 0.764 & 0.928 & 0.919 & 0.617 \\

 Vector-borne & 0.882 & 0.777 & 0.799 & 0.881 & 0.972 & 0.967 & 0.754 \\

\bottomrule\end{longtable}\endgroup

The CSV files include support-aware macro/weighted AUC calculations restricted to evaluable brackets and ROC AUC of each $R>c$ threshold. These alternatives are clearly distinguished from full-class macro/weighted AUC and do not silently fill absent classes. Each per-bracket diagnostic table below includes observed support, sensitivity (recall), specificity, precision, F1, negative predictive value, and OvR AUC.

\subsubsection{Closed SIR: validation and equations}

For the binary $R_0>1$ model, $z_j=(\theta_j-\bar\theta_j)/s_j$; the binary training-sample scaling constants (independent of ordered-sample scaling) are: $z_{1}=(\beta-0.525001)/0.274270$, $z_{2}=(\gamma-0.525001)/0.274269$.

\noindent\textbf{Binary cutoff $R_0>1$.}
\[\begin{aligned}
\eta_{R_0>1} &= -0.04280 \\
&\quad + 16.47672z_{1} - 16.49150z_{2}
\end{aligned}\]

\begingroup\scriptsize\setlength{\LTcapwidth}{\linewidth}\setlength{\tabcolsep}{3pt}

\begin{longtable}{@{}lrrrrrrr@{}}

\caption{Class-specific held-out metrics for Closed SIR with linear $h=5$.}\label{tab:supp-class-linear_h5-closed-sir}\\

\toprule Bracket & Support & Recall & Spec. & Precision & F1 & NPV & OvR AUC \\ \midrule\endfirsthead

\toprule Bracket & Support & Recall & Spec. & Precision & F1 & NPV & OvR AUC \\ \midrule\endhead

 \texttt{$<$1} & 4012 & 0.998 & 0.999 & 0.999 & 0.999 & 0.998 & 1.000 \\

 \texttt{1-6} & 3625 & 0.999 & 0.987 & 0.985 & 0.992 & 0.999 & 1.000 \\

 \texttt{6-10} & 250 & 0.780 & 0.993 & 0.771 & 0.775 & 0.993 & 0.994 \\

 \texttt{$>$10} & 113 & 0.496 & 0.999 & 0.918 & 0.644 & 0.993 & 0.998 \\

\bottomrule\end{longtable}\endgroup

\paragraph{Ordered-threshold feature map and equations.}

All $2$ fitted predictors are standardized using $z_j=(\phi_j(\boldsymbol\theta)-\bar\phi_j)/s_j$, where the values below use the complete training sample for the printed (full-data) equations. Held-out predictions separately recompute means, standard deviations, and fits within each training fold. Here $\phi_j$ is the feature shown in the table. The score for every estimable cutoff is $q_c=\operatorname{logit}^{-1}(\eta_c)$; positive scores increase the fitted exceedance probability, not the numerical reproduction number.

\begingroup\scriptsize\setlength{\LTcapwidth}{\linewidth}

\begin{longtable}{@{}r p{0.44\linewidth} r r@{}}

\caption{Standardization constants and ordered features for Closed SIR (raw linear basis).}\label{tab:supp-features-closed-sir-linear}\\

\toprule Index & Raw feature $\phi_j$ & Training mean & Training SD \\ \midrule\endfirsthead

\toprule Index & Raw feature $\phi_j$ & Training mean & Training SD \\ \midrule\endhead

 1 & $\beta$ & 0.525000 & 0.274259 \\

 2 & $\gamma$ & 0.525000 & 0.274259 \\

\bottomrule\end{longtable}\endgroup

Each displayed score below is a numeric equation in the standardized variables $z_1,\ldots,z_2$. The intercept and slopes are posterior modes from the independent $N(0,2^2)$ coefficient prior under the Gaussian Laplace approximation. A cutoff with only one observed class cannot produce a fitted finite logistic score, so its constant exceedance probability is stated explicitly.

\noindent\textbf{Cutoff $R_0>1$.}
\[\begin{aligned}
\eta_{R_0>1} &= -0.02550 \\
&\quad + 19.52491z_{1} - 19.40418z_{2}
\end{aligned}\]

\noindent\textbf{Cutoff $R_0>6$.}
\[\begin{aligned}
\eta_{R_0>6} &= -22.12775 \\
&\quad + 2.78498z_{1} - 13.58500z_{2}
\end{aligned}\]

\noindent\textbf{Cutoff $R_0>10$.}
\[\begin{aligned}
\eta_{R_0>10} &= -17.64488 \\
&\quad + 2.07489z_{1} - 9.25181z_{2}
\end{aligned}\]

\subsubsection{Open SIR: validation and equations}

For the binary $R_0>1$ model, $z_j=(\theta_j-\bar\theta_j)/s_j$; the binary training-sample scaling constants (independent of ordered-sample scaling) are: $z_{1}=(\beta-0.625001)/0.332010$, $z_{2}=(\gamma-0.315000)/0.164561$, $z_{3}=(\mu-0.127500)/0.070732$.

\noindent\textbf{Binary cutoff $R_0>1$.}
\[\begin{aligned}
\eta_{R_0>1} &= 9.99254 \\
&\quad + 18.16614z_{1} - 8.99984z_{2} - 3.66832z_{3}
\end{aligned}\]

\begingroup\scriptsize\setlength{\LTcapwidth}{\linewidth}\setlength{\tabcolsep}{3pt}

\begin{longtable}{@{}lrrrrrrr@{}}

\caption{Class-specific held-out metrics for Open SIR with linear $h=5$.}\label{tab:supp-class-linear_h5-open-sir}\\

\toprule Bracket & Support & Recall & Spec. & Precision & F1 & NPV & OvR AUC \\ \midrule\endfirsthead

\toprule Bracket & Support & Recall & Spec. & Precision & F1 & NPV & OvR AUC \\ \midrule\endhead

 \texttt{$<$1} & 2709 & 0.998 & 0.998 & 0.997 & 0.997 & 0.999 & 1.000 \\

 \texttt{1-6} & 5070 & 0.998 & 0.989 & 0.994 & 0.996 & 0.997 & 1.000 \\

 \texttt{6-10} & 162 & 0.827 & 0.996 & 0.822 & 0.825 & 0.996 & 0.998 \\

 \texttt{$>$10} & 59 & 0.525 & 1.000 & 0.969 & 0.681 & 0.996 & 1.000 \\

\bottomrule\end{longtable}\endgroup

\paragraph{Ordered-threshold feature map and equations.}

All $3$ fitted predictors are standardized using $z_j=(\phi_j(\boldsymbol\theta)-\bar\phi_j)/s_j$, where the values below use the complete training sample for the printed (full-data) equations. Held-out predictions separately recompute means, standard deviations, and fits within each training fold. Here $\phi_j$ is the feature shown in the table. The score for every estimable cutoff is $q_c=\operatorname{logit}^{-1}(\eta_c)$; positive scores increase the fitted exceedance probability, not the numerical reproduction number.

\begingroup\scriptsize\setlength{\LTcapwidth}{\linewidth}

\begin{longtable}{@{}r p{0.44\linewidth} r r@{}}

\caption{Standardization constants and ordered features for Open SIR (raw linear basis).}\label{tab:supp-features-open-sir-linear}\\

\toprule Index & Raw feature $\phi_j$ & Training mean & Training SD \\ \midrule\endfirsthead

\toprule Index & Raw feature $\phi_j$ & Training mean & Training SD \\ \midrule\endhead

 1 & $\beta$ & 0.625000 & 0.331997 \\

 2 & $\gamma$ & 0.315000 & 0.164555 \\

 3 & $\mu$ & 0.127500 & 0.070730 \\

\bottomrule\end{longtable}\endgroup

Each displayed score below is a numeric equation in the standardized variables $z_1,\ldots,z_3$. The intercept and slopes are posterior modes from the independent $N(0,2^2)$ coefficient prior under the Gaussian Laplace approximation. A cutoff with only one observed class cannot produce a fitted finite logistic score, so its constant exceedance probability is stated explicitly.

\noindent\textbf{Cutoff $R_0>1$.}
\[\begin{aligned}
\eta_{R_0>1} &= 11.55806 \\
&\quad + 21.05637z_{1} - 10.37200z_{2} - 4.60783z_{3}
\end{aligned}\]

\noindent\textbf{Cutoff $R_0>6$.}
\[\begin{aligned}
\eta_{R_0>6} &= -18.57968 \\
&\quad + 3.18594z_{1} - 8.47904z_{2} - 4.13327z_{3}
\end{aligned}\]

\noindent\textbf{Cutoff $R_0>10$.}
\[\begin{aligned}
\eta_{R_0>10} &= -14.91386 \\
&\quad + 1.70025z_{1} - 5.00689z_{2} - 3.61580z_{3}
\end{aligned}\]

\subsubsection{Open SEIR: validation and equations}

For the binary $R_0>1$ model, $z_j=(\theta_j-\bar\theta_j)/s_j$; the binary training-sample scaling constants (independent of ordered-sample scaling) are: $z_{1}=(\beta-0.774999)/0.418622$, $z_{2}=(\sigma-0.525001)/0.274270$, $z_{3}=(\gamma-0.315001)/0.164562$, $z_{4}=(\mu-0.102500)/0.056297$.

\noindent\textbf{Binary cutoff $R_0>1$.}
\[\begin{aligned}
\eta_{R_0>1} &= 3.72466 \\
&\quad + 6.86214z_{1} + 1.65577z_{2} - 3.55826z_{3} \\
&\quad - 2.20596z_{4}
\end{aligned}\]

\begingroup\scriptsize\setlength{\LTcapwidth}{\linewidth}\setlength{\tabcolsep}{3pt}

\begin{longtable}{@{}lrrrrrrr@{}}

\caption{Class-specific held-out metrics for Open SEIR with linear $h=5$.}\label{tab:supp-class-linear_h5-open-seir}\\

\toprule Bracket & Support & Recall & Spec. & Precision & F1 & NPV & OvR AUC \\ \midrule\endfirsthead

\toprule Bracket & Support & Recall & Spec. & Precision & F1 & NPV & OvR AUC \\ \midrule\endhead

 \texttt{$<$1} & 2764 & 0.924 & 0.965 & 0.934 & 0.929 & 0.960 & 0.989 \\

 \texttt{1-6} & 4899 & 0.961 & 0.918 & 0.949 & 0.955 & 0.937 & 0.986 \\

 \texttt{6-10} & 244 & 0.803 & 0.995 & 0.831 & 0.817 & 0.994 & 0.997 \\

 \texttt{$>$10} & 93 & 0.688 & 1.000 & 0.955 & 0.800 & 0.996 & 1.000 \\

\bottomrule\end{longtable}\endgroup

\paragraph{Ordered-threshold feature map and equations.}

All $4$ fitted predictors are standardized using $z_j=(\phi_j(\boldsymbol\theta)-\bar\phi_j)/s_j$, where the values below use the complete training sample for the printed (full-data) equations. Held-out predictions separately recompute means, standard deviations, and fits within each training fold. Here $\phi_j$ is the feature shown in the table. The score for every estimable cutoff is $q_c=\operatorname{logit}^{-1}(\eta_c)$; positive scores increase the fitted exceedance probability, not the numerical reproduction number.

\begingroup\scriptsize\setlength{\LTcapwidth}{\linewidth}

\begin{longtable}{@{}r p{0.44\linewidth} r r@{}}

\caption{Standardization constants and ordered features for Open SEIR (raw linear basis).}\label{tab:supp-features-open-seir-linear}\\

\toprule Index & Raw feature $\phi_j$ & Training mean & Training SD \\ \midrule\endfirsthead

\toprule Index & Raw feature $\phi_j$ & Training mean & Training SD \\ \midrule\endhead

 1 & $\beta$ & 0.775000 & 0.418605 \\

 2 & $\sigma$ & 0.525000 & 0.274259 \\

 3 & $\gamma$ & 0.315000 & 0.164555 \\

 4 & $\mu$ & 0.102500 & 0.056295 \\

\bottomrule\end{longtable}\endgroup

Each displayed score below is a numeric equation in the standardized variables $z_1,\ldots,z_4$. The intercept and slopes are posterior modes from the independent $N(0,2^2)$ coefficient prior under the Gaussian Laplace approximation. A cutoff with only one observed class cannot produce a fitted finite logistic score, so its constant exceedance probability is stated explicitly.

\noindent\textbf{Cutoff $R_0>1$.}
\[\begin{aligned}
\eta_{R_0>1} &= 3.93164 \\
&\quad + 7.19453z_{1} + 1.79351z_{2} - 3.75294z_{3} \\
&\quad - 2.55239z_{4}
\end{aligned}\]

\noindent\textbf{Cutoff $R_0>6$.}
\[\begin{aligned}
\eta_{R_0>6} &= -19.00917 \\
&\quad + 3.75660z_{1} + 0.78636z_{2} - 9.43432z_{3} \\
&\quad - 4.54975z_{4}
\end{aligned}\]

\noindent\textbf{Cutoff $R_0>10$.}
\[\begin{aligned}
\eta_{R_0>10} &= -16.00918 \\
&\quad + 2.32024z_{1} + 0.26876z_{2} - 6.16483z_{3} \\
&\quad - 3.44168z_{4}
\end{aligned}\]

\subsubsection{SEIR with vaccination: validation and equations}

For the binary $R_e(0)>1$ model, $z_j=(\theta_j-\bar\theta_j)/s_j$; the binary training-sample scaling constants (independent of ordered-sample scaling) are: $z_{1}=(\beta-2.049998)/1.125945$, $z_{2}=(\sigma-0.524999)/0.274267$, $z_{3}=(\gamma-0.315000)/0.164562$, $z_{4}=(\mu-0.155000)/0.083724$, $z_{5}=(v-0.075500)/0.043017$.

\noindent\textbf{Binary cutoff $R_e(0)>1$.}
\[\begin{aligned}
\eta_{R_e(0)>1} &= 3.22076 \\
&\quad + 4.11195z_{1} + 1.02771z_{2} - 1.67089z_{3} \\
&\quad - 0.22074z_{4} - 0.92731z_{5}
\end{aligned}\]

\begingroup\scriptsize\setlength{\LTcapwidth}{\linewidth}\setlength{\tabcolsep}{3pt}

\begin{longtable}{@{}lrrrrrrr@{}}

\caption{Class-specific held-out metrics for SEIR with vaccination with linear $h=5$.}\label{tab:supp-class-linear_h5-seir-with-vaccination}\\

\toprule Bracket & Support & Recall & Spec. & Precision & F1 & NPV & OvR AUC \\ \midrule\endfirsthead

\toprule Bracket & Support & Recall & Spec. & Precision & F1 & NPV & OvR AUC \\ \midrule\endhead

 \texttt{$<$1} & 2221 & 0.854 & 0.952 & 0.872 & 0.863 & 0.944 & 0.969 \\

 \texttt{1-6} & 5265 & 0.934 & 0.838 & 0.917 & 0.926 & 0.868 & 0.956 \\

 \texttt{6-10} & 403 & 0.675 & 0.985 & 0.701 & 0.688 & 0.983 & 0.988 \\

 \texttt{$>$10} & 111 & 0.568 & 0.998 & 0.840 & 0.677 & 0.994 & 0.998 \\

\bottomrule\end{longtable}\endgroup

\paragraph{Ordered-threshold feature map and equations.}

All $5$ fitted predictors are standardized using $z_j=(\phi_j(\boldsymbol\theta)-\bar\phi_j)/s_j$, where the values below use the complete training sample for the printed (full-data) equations. Held-out predictions separately recompute means, standard deviations, and fits within each training fold. Here $\phi_j$ is the feature shown in the table. The score for every estimable cutoff is $q_c=\operatorname{logit}^{-1}(\eta_c)$; positive scores increase the fitted exceedance probability, not the numerical reproduction number.

\begingroup\scriptsize\setlength{\LTcapwidth}{\linewidth}

\begin{longtable}{@{}r p{0.44\linewidth} r r@{}}

\caption{Standardization constants and ordered features for SEIR with vaccination (raw linear basis).}\label{tab:supp-features-seir-with-vaccination-linear}\\

\toprule Index & Raw feature $\phi_j$ & Training mean & Training SD \\ \midrule\endfirsthead

\toprule Index & Raw feature $\phi_j$ & Training mean & Training SD \\ \midrule\endhead

 1 & $\beta$ & 2.050003 & 1.125903 \\

 2 & $\sigma$ & 0.525001 & 0.274258 \\

 3 & $\gamma$ & 0.315000 & 0.164555 \\

 4 & $\mu$ & 0.155000 & 0.083721 \\

 5 & $v$ & 0.075500 & 0.043015 \\

\bottomrule\end{longtable}\endgroup

Each displayed score below is a numeric equation in the standardized variables $z_1,\ldots,z_5$. The intercept and slopes are posterior modes from the independent $N(0,2^2)$ coefficient prior under the Gaussian Laplace approximation. A cutoff with only one observed class cannot produce a fitted finite logistic score, so its constant exceedance probability is stated explicitly.

\noindent\textbf{Cutoff $R_e(0)>1$.}
\[\begin{aligned}
\eta_{R_e(0)>1} &= 3.38455 \\
&\quad + 4.24173z_{1} + 1.13002z_{2} - 1.66392z_{3} \\
&\quad - 0.12278z_{4} - 1.11260z_{5}
\end{aligned}\]

\noindent\textbf{Cutoff $R_e(0)>6$.}
\[\begin{aligned}
\eta_{R_e(0)>6} &= -12.36869 \\
&\quad + 3.98513z_{1} + 1.28538z_{2} - 5.90936z_{3} \\
&\quad - 1.98077z_{4} - 2.57111z_{5}
\end{aligned}\]

\noindent\textbf{Cutoff $R_e(0)>10$.}
\[\begin{aligned}
\eta_{R_e(0)>10} &= -14.85073 \\
&\quad + 2.52536z_{1} + 0.83210z_{2} - 5.65536z_{3} \\
&\quad - 2.22535z_{4} - 2.29947z_{5}
\end{aligned}\]

\subsubsection{Age-structured SIR: validation and equations}

For the binary $R_0>1$ model, $z_j=(\theta_j-\bar\theta_j)/s_j$; the binary training-sample scaling constants (independent of ordered-sample scaling) are: $z_{1}=(\beta-0.415000)/0.222302$, $z_{2}=(\gamma_1-0.275000)/0.129917$, $z_{3}=(\gamma_2-0.275000)/0.129917$, $z_{4}=(c_{11}-0.800000)/0.404188$, $z_{5}=(c_{12}-0.525000)/0.274270$, $z_{6}=(c_{21}-0.525000)/0.274268$, $z_{7}=(c_{22}-0.799999)/0.404185$.

\noindent\textbf{Binary cutoff $R_0>1$.}
\[\begin{aligned}
\eta_{R_0>1} &= 6.46265 \\
&\quad + 6.42182z_{1} - 1.84060z_{2} - 1.60343z_{3} \\
&\quad + 1.28088z_{4} + 0.81248z_{5} + 0.79145z_{6} \\
&\quad + 1.28878z_{7}
\end{aligned}\]

\begingroup\scriptsize\setlength{\LTcapwidth}{\linewidth}\setlength{\tabcolsep}{3pt}

\begin{longtable}{@{}lrrrrrrr@{}}

\caption{Class-specific held-out metrics for Age-structured SIR with linear $h=5$.}\label{tab:supp-class-linear_h5-age-structured-sir}\\

\toprule Bracket & Support & Recall & Spec. & Precision & F1 & NPV & OvR AUC \\ \midrule\endfirsthead

\toprule Bracket & Support & Recall & Spec. & Precision & F1 & NPV & OvR AUC \\ \midrule\endhead

 \texttt{$<$1} & 1729 & 0.873 & 0.970 & 0.888 & 0.880 & 0.965 & 0.986 \\

 \texttt{1-6} & 5366 & 0.940 & 0.776 & 0.895 & 0.917 & 0.865 & 0.939 \\

 \texttt{6-10} & 654 & 0.462 & 0.967 & 0.555 & 0.504 & 0.953 & 0.922 \\

 \texttt{$>$10} & 251 & 0.311 & 0.995 & 0.650 & 0.420 & 0.978 & 0.952 \\

\bottomrule\end{longtable}\endgroup

\paragraph{Ordered-threshold feature map and equations.}

All $7$ fitted predictors are standardized using $z_j=(\phi_j(\boldsymbol\theta)-\bar\phi_j)/s_j$, where the values below use the complete training sample for the printed (full-data) equations. Held-out predictions separately recompute means, standard deviations, and fits within each training fold. Here $\phi_j$ is the feature shown in the table. The score for every estimable cutoff is $q_c=\operatorname{logit}^{-1}(\eta_c)$; positive scores increase the fitted exceedance probability, not the numerical reproduction number.

\begingroup\scriptsize\setlength{\LTcapwidth}{\linewidth}

\begin{longtable}{@{}r p{0.44\linewidth} r r@{}}

\caption{Standardization constants and ordered features for Age-structured SIR (raw linear basis).}\label{tab:supp-features-age-structured-sir-linear}\\

\toprule Index & Raw feature $\phi_j$ & Training mean & Training SD \\ \midrule\endfirsthead

\toprule Index & Raw feature $\phi_j$ & Training mean & Training SD \\ \midrule\endhead

 1 & $\beta$ & 0.414999 & 0.222294 \\

 2 & $\gamma_1$ & 0.275000 & 0.129912 \\

 3 & $\gamma_2$ & 0.275000 & 0.129912 \\

 4 & $c_{11}$ & 0.800000 & 0.404172 \\

 5 & $c_{12}$ & 0.525000 & 0.274259 \\

 6 & $c_{21}$ & 0.525000 & 0.274259 \\

 7 & $c_{22}$ & 0.800001 & 0.404171 \\

\bottomrule\end{longtable}\endgroup

Each displayed score below is a numeric equation in the standardized variables $z_1,\ldots,z_7$. The intercept and slopes are posterior modes from the independent $N(0,2^2)$ coefficient prior under the Gaussian Laplace approximation. A cutoff with only one observed class cannot produce a fitted finite logistic score, so its constant exceedance probability is stated explicitly.

\noindent\textbf{Cutoff $R_0>1$.}
\[\begin{aligned}
\eta_{R_0>1} &= 6.77355 \\
&\quad + 6.68576z_{1} - 1.79851z_{2} - 1.81802z_{3} \\
&\quad + 1.31850z_{4} + 0.61787z_{5} + 0.60576z_{6} \\
&\quad + 1.28146z_{7}
\end{aligned}\]

\noindent\textbf{Cutoff $R_0>6$.}
\[\begin{aligned}
\eta_{R_0>6} &= -5.02772 \\
&\quad + 2.32909z_{1} - 1.67799z_{2} - 1.85421z_{3} \\
&\quad + 0.85883z_{4} + 0.27804z_{5} + 0.32068z_{6} \\
&\quad + 0.96436z_{7}
\end{aligned}\]

\noindent\textbf{Cutoff $R_0>10$.}
\[\begin{aligned}
\eta_{R_0>10} &= -6.60593 \\
&\quad + 2.01211z_{1} - 1.55392z_{2} - 1.41256z_{3} \\
&\quad + 0.71699z_{4} + 0.22277z_{5} + 0.25685z_{6} \\
&\quad + 0.69982z_{7}
\end{aligned}\]

\subsubsection{15-compartment five-group SIR: validation and equations}

For the binary $R_0>1$ model, $z_j=(\theta_j-\bar\theta_j)/s_j$; the binary training-sample scaling constants (independent of ordered-sample scaling) are: $z_{1}=(\beta-0.152500)/0.085168$, $z_{2}=(\gamma_1-0.275000)/0.129917$, $z_{3}=(\gamma_2-0.275000)/0.129917$, $z_{4}=(\gamma_3-0.275000)/0.129917$, $z_{5}=(\gamma_4-0.275000)/0.129917$, $z_{6}=(\gamma_5-0.275000)/0.129917$.

\noindent\textbf{Binary cutoff $R_0>1$.}
\[\begin{aligned}
\eta_{R_0>1} &= 16.10250 \\
&\quad + 11.45645z_{1} - 0.59421z_{2} - 0.71818z_{3} \\
&\quad - 0.56415z_{4} - 0.61466z_{5} - 0.40383z_{6}
\end{aligned}\]

\begingroup\scriptsize\setlength{\LTcapwidth}{\linewidth}\setlength{\tabcolsep}{3pt}

\begin{longtable}{@{}lrrrrrrr@{}}

\caption{Class-specific held-out metrics for 15-compartment five-group SIR with linear $h=5$.}\label{tab:supp-class-linear_h5-15-compartment-five-group-sir}\\

\toprule Bracket & Support & Recall & Spec. & Precision & F1 & NPV & OvR AUC \\ \midrule\endfirsthead

\toprule Bracket & Support & Recall & Spec. & Precision & F1 & NPV & OvR AUC \\ \midrule\endhead

 \texttt{$<$1} & 776 & 0.874 & 0.986 & 0.868 & 0.871 & 0.986 & 0.995 \\

 \texttt{1-6} & 4425 & 0.883 & 0.802 & 0.846 & 0.864 & 0.847 & 0.929 \\

 \texttt{6-10} & 1824 & 0.568 & 0.869 & 0.562 & 0.565 & 0.872 & 0.856 \\

 \texttt{$>$10} & 975 & 0.531 & 0.966 & 0.682 & 0.597 & 0.937 & 0.932 \\

\bottomrule\end{longtable}\endgroup

\paragraph{Ordered-threshold feature map and equations.}

All $6$ fitted predictors are standardized using $z_j=(\phi_j(\boldsymbol\theta)-\bar\phi_j)/s_j$, where the values below use the complete training sample for the printed (full-data) equations. Held-out predictions separately recompute means, standard deviations, and fits within each training fold. Here $\phi_j$ is the feature shown in the table. The score for every estimable cutoff is $q_c=\operatorname{logit}^{-1}(\eta_c)$; positive scores increase the fitted exceedance probability, not the numerical reproduction number.

\begingroup\scriptsize\setlength{\LTcapwidth}{\linewidth}

\begin{longtable}{@{}r p{0.44\linewidth} r r@{}}

\caption{Standardization constants and ordered features for 15-compartment five-group SIR (raw linear basis).}\label{tab:supp-features-15-compartment-five-group-sir-linear}\\

\toprule Index & Raw feature $\phi_j$ & Training mean & Training SD \\ \midrule\endfirsthead

\toprule Index & Raw feature $\phi_j$ & Training mean & Training SD \\ \midrule\endhead

 1 & $\beta$ & 0.152500 & 0.085164 \\

 2 & $\gamma_1$ & 0.275000 & 0.129912 \\

 3 & $\gamma_2$ & 0.275000 & 0.129912 \\

 4 & $\gamma_3$ & 0.275000 & 0.129912 \\

 5 & $\gamma_4$ & 0.275000 & 0.129912 \\

 6 & $\gamma_5$ & 0.275000 & 0.129912 \\

\bottomrule\end{longtable}\endgroup

Each displayed score below is a numeric equation in the standardized variables $z_1,\ldots,z_6$. The intercept and slopes are posterior modes from the independent $N(0,2^2)$ coefficient prior under the Gaussian Laplace approximation. A cutoff with only one observed class cannot produce a fitted finite logistic score, so its constant exceedance probability is stated explicitly.

\noindent\textbf{Cutoff $R_0>1$.}
\[\begin{aligned}
\eta_{R_0>1} &= 17.58651 \\
&\quad + 12.60636z_{1} - 0.65001z_{2} - 0.88915z_{3} \\
&\quad - 0.74272z_{4} - 0.57191z_{5} - 0.45613z_{6}
\end{aligned}\]

\noindent\textbf{Cutoff $R_0>6$.}
\[\begin{aligned}
\eta_{R_0>6} &= -1.67511 \\
&\quad + 3.00743z_{1} - 0.88540z_{2} - 1.07712z_{3} \\
&\quad - 0.91067z_{4} - 0.74923z_{5} - 0.57453z_{6}
\end{aligned}\]

\noindent\textbf{Cutoff $R_0>10$.}
\[\begin{aligned}
\eta_{R_0>10} &= -4.19228 \\
&\quad + 2.56890z_{1} - 0.56975z_{2} - 1.02280z_{3} \\
&\quad - 0.75957z_{4} - 0.59262z_{5} - 0.27580z_{6}
\end{aligned}\]

\subsubsection{Vector-borne: validation and equations}

For the binary $R_0>1$ model, $z_j=(\theta_j-\bar\theta_j)/s_j$; the binary training-sample scaling constants (independent of ordered-sample scaling) are: $z_{1}=(\beta_{hv}-0.775000)/0.418622$, $z_{2}=(\beta_{vh}-0.775001)/0.418621$, $z_{3}=(\gamma_h-0.325001)/0.158787$, $z_{4}=(\mu_v-0.325000)/0.158786$, $z_{5}=(m-3.249998)/1.587873$.

\noindent\textbf{Binary cutoff $R_0>1$.}
\[\begin{aligned}
\eta_{R_0>1} &= 9.11503 \\
&\quad + 3.08141z_{1} + 2.99646z_{2} - 1.59636z_{3} \\
&\quad - 1.52440z_{4} + 2.13926z_{5}
\end{aligned}\]

\begingroup\scriptsize\setlength{\LTcapwidth}{\linewidth}\setlength{\tabcolsep}{3pt}

\begin{longtable}{@{}lrrrrrrr@{}}

\caption{Class-specific held-out metrics for Vector-borne with linear $h=5$.}\label{tab:supp-class-linear_h5-vector-borne}\\

\toprule Bracket & Support & Recall & Spec. & Precision & F1 & NPV & OvR AUC \\ \midrule\endfirsthead

\toprule Bracket & Support & Recall & Spec. & Precision & F1 & NPV & OvR AUC \\ \midrule\endhead

 \texttt{$<$1} & 380 & 0.624 & 0.991 & 0.777 & 0.692 & 0.981 & 0.982 \\

 \texttt{1-6} & 5476 & 0.948 & 0.834 & 0.925 & 0.936 & 0.880 & 0.966 \\

 \texttt{6-10} & 1397 & 0.742 & 0.945 & 0.739 & 0.741 & 0.945 & 0.954 \\

 \texttt{$>$10} & 747 & 0.794 & 0.987 & 0.867 & 0.829 & 0.979 & 0.988 \\

\bottomrule\end{longtable}\endgroup

\paragraph{Ordered-threshold feature map and equations.}

All $5$ fitted predictors are standardized using $z_j=(\phi_j(\boldsymbol\theta)-\bar\phi_j)/s_j$, where the values below use the complete training sample for the printed (full-data) equations. Held-out predictions separately recompute means, standard deviations, and fits within each training fold. Here $\phi_j$ is the feature shown in the table. The score for every estimable cutoff is $q_c=\operatorname{logit}^{-1}(\eta_c)$; positive scores increase the fitted exceedance probability, not the numerical reproduction number.

\begingroup\scriptsize\setlength{\LTcapwidth}{\linewidth}

\begin{longtable}{@{}r p{0.44\linewidth} r r@{}}

\caption{Standardization constants and ordered features for Vector-borne (raw linear basis).}\label{tab:supp-features-vector-borne-linear}\\

\toprule Index & Raw feature $\phi_j$ & Training mean & Training SD \\ \midrule\endfirsthead

\toprule Index & Raw feature $\phi_j$ & Training mean & Training SD \\ \midrule\endhead

 1 & $\beta_{hv}$ & 0.775000 & 0.418606 \\

 2 & $\beta_{vh}$ & 0.775000 & 0.418605 \\

 3 & $\gamma_h$ & 0.325000 & 0.158781 \\

 4 & $\mu_v$ & 0.325000 & 0.158782 \\

 5 & $m$ & 3.249999 & 1.587814 \\

\bottomrule\end{longtable}\endgroup

Each displayed score below is a numeric equation in the standardized variables $z_1,\ldots,z_5$. The intercept and slopes are posterior modes from the independent $N(0,2^2)$ coefficient prior under the Gaussian Laplace approximation. A cutoff with only one observed class cannot produce a fitted finite logistic score, so its constant exceedance probability is stated explicitly.

\noindent\textbf{Cutoff $R_0>1$.}
\[\begin{aligned}
\eta_{R_0>1} &= 8.85460 \\
&\quad + 2.92937z_{1} + 2.91160z_{2} - 1.27723z_{3} \\
&\quad - 1.41902z_{4} + 2.07702z_{5}
\end{aligned}\]

\noindent\textbf{Cutoff $R_0>6$.}
\[\begin{aligned}
\eta_{R_0>6} &= -4.43547 \\
&\quad + 2.92508z_{1} + 3.06523z_{2} - 3.18408z_{3} \\
&\quad - 3.24908z_{4} + 2.74502z_{5}
\end{aligned}\]

\noindent\textbf{Cutoff $R_0>10$.}
\[\begin{aligned}
\eta_{R_0>10} &= -9.44178 \\
&\quad + 2.74462z_{1} + 2.55680z_{2} - 3.58983z_{3} \\
&\quad - 3.65854z_{4} + 2.48482z_{5}
\end{aligned}\]

\subsection{{Main-text transformations} with finer sequential brackets h = 0.5}\label{app:log_cubic_h0p5}

Main-text compact log-cubic transformation; finer bracket width $h=0.5$. The underlying parameter tables and analytic labels are generated from independent Latin-hypercube designs for binary ($n=5000$, seed $123+i$) and ordered ($n=8000$, seed $1700+i$) fits, where $i=0,\ldots,6$ indexes the seven benchmark models. A proper independent coefficient prior $N(0,2^2)$ is used; MAP/Laplace logistic fits and training-fold-only standardization match the accompanying reproducibility Python. Reported scores are pooled five-fold held-out metrics. AUC in the ordered summaries is macro/weighted one-vs-rest and is undefined when the prescribed full class set is not represented; per-bracket AUC is reported only where positive and negative observations exist.

\subsubsection{Five-fold ordered-bracket evaluation}

Exact bracket classification is assessed under the 20-class partition with boundaries $1,1.5,2,\ldots,10$. A change in the number and granularity of classes changes the difficulty of exact-bracket prediction; accuracy and macro F1 should therefore not be interpreted as direct apples-to-apples comparisons across these two scenarios.

\begingroup\scriptsize\setlength{\LTcapwidth}{\linewidth}\setlength{\tabcolsep}{3pt}

\begin{longtable}{@{}lrrrrrrr@{}}

\caption{Five-fold ordered-bracket classification metrics for log cubic $h=0.5$. AUC values are OvR; a dash indicates that at least one prescribed class is absent, making the full-class multiclass AUC undefined.}\label{tab:supp-log_cubic_h0p5-ordered}\\

\toprule Model & Acc. & Bal. acc. & Macro F1 & Weighted F1 & Macro AUC & Weighted AUC & $\kappa$ \\ \midrule\endfirsthead

\toprule Model & Acc. & Bal. acc. & Macro F1 & Weighted F1 & Macro AUC & Weighted AUC & $\kappa$ \\ \midrule\endhead

 Closed SIR & 0.962 & 0.714 & 0.708 & 0.959 & 0.999 & 1.000 & 0.946 \\

 Open SIR & 0.931 & 0.552 & 0.540 & 0.928 & 0.989 & 0.999 & 0.915 \\

 Open SEIR & 0.917 & 0.570 & 0.575 & 0.915 & 0.993 & 0.998 & 0.898 \\

 SEIR with vaccination & 0.860 & 0.561 & 0.555 & 0.856 & 0.986 & 0.995 & 0.837 \\

 Age-structured SIR & 0.775 & 0.537 & 0.538 & 0.772 & 0.969 & 0.985 & 0.748 \\

 15-compartment five-group SIR & 0.698 & 0.588 & 0.586 & 0.691 & 0.974 & 0.980 & 0.678 \\

 Vector-borne & 0.930 & 0.873 & 0.878 & 0.929 & 0.997 & 0.999 & 0.926 \\

\bottomrule\end{longtable}\endgroup

The CSV files include support-aware macro/weighted AUC calculations restricted to evaluable brackets and ROC AUC of each $R>c$ threshold. These alternatives are clearly distinguished from full-class macro/weighted AUC and do not silently fill absent classes. Each per-bracket diagnostic table below includes observed support, sensitivity (recall), specificity, precision, F1, negative predictive value, and OvR AUC.

\subsubsection{Closed SIR: validation and equations}

\begingroup\scriptsize\setlength{\LTcapwidth}{\linewidth}\setlength{\tabcolsep}{3pt}

\begin{longtable}{@{}lrrrrrrr@{}}

\caption{Class-specific held-out metrics for Closed SIR with log cubic $h=0.5$.}\label{tab:supp-class-log_cubic_h0p5-closed-sir}\\

\toprule Bracket & Support & Recall & Spec. & Precision & F1 & NPV & OvR AUC \\ \midrule\endfirsthead

\toprule Bracket & Support & Recall & Spec. & Precision & F1 & NPV & OvR AUC \\ \midrule\endhead

 \texttt{$<$1} & 4012 & 0.998 & 0.999 & 1.000 & 0.999 & 0.998 & 1.000 \\

 \texttt{1-1.5} & 1452 & 0.991 & 0.996 & 0.980 & 0.986 & 0.998 & 1.000 \\

 \texttt{1.5-2} & 728 & 0.967 & 0.997 & 0.966 & 0.966 & 0.997 & 1.000 \\

 \texttt{2-2.5} & 438 & 0.947 & 0.998 & 0.958 & 0.953 & 0.997 & 1.000 \\

 \texttt{2.5-3} & 315 & 0.949 & 0.996 & 0.901 & 0.924 & 0.998 & 0.999 \\

 \texttt{3-3.5} & 194 & 0.876 & 0.998 & 0.914 & 0.895 & 0.997 & 1.000 \\

 \texttt{3.5-4} & 135 & 0.867 & 0.999 & 0.921 & 0.893 & 0.998 & 1.000 \\

 \texttt{4-4.5} & 127 & 0.913 & 0.997 & 0.835 & 0.872 & 0.999 & 0.999 \\

 \texttt{4.5-5} & 96 & 0.781 & 0.998 & 0.824 & 0.802 & 0.997 & 0.999 \\

 \texttt{5-5.5} & 79 & 0.785 & 0.997 & 0.713 & 0.747 & 0.998 & 0.999 \\

 \texttt{5.5-6} & 61 & 0.525 & 0.999 & 0.821 & 0.640 & 0.996 & 0.998 \\

 \texttt{6-6.5} & 58 & 0.914 & 0.997 & 0.679 & 0.779 & 0.999 & 0.999 \\

 \texttt{6.5-7} & 37 & 0.378 & 1.000 & 0.875 & 0.528 & 0.997 & 0.999 \\

 \texttt{7-7.5} & 40 & 0.950 & 0.998 & 0.679 & 0.792 & 1.000 & 0.999 \\

 \texttt{7.5-8} & 32 & 0.500 & 0.999 & 0.727 & 0.593 & 0.998 & 0.999 \\

 \texttt{8-8.5} & 33 & 0.788 & 0.997 & 0.542 & 0.642 & 0.999 & 0.998 \\

 \texttt{8.5-9} & 20 & 0.150 & 1.000 & 0.750 & 0.250 & 0.998 & 0.998 \\

 \texttt{9-9.5} & 14 & 0.000 & 1.000 & -- & 0.000 & 0.998 & 0.998 \\

 \texttt{9.5-10} & 16 & 0.000 & 1.000 & -- & 0.000 & 0.998 & 0.998 \\

 \texttt{$>$10} & 113 & 1.000 & 0.997 & 0.813 & 0.897 & 1.000 & 1.000 \\

\bottomrule\end{longtable}\endgroup

\paragraph{Ordered-threshold feature map and equations.}

All $7$ fitted predictors are standardized using $z_j=(\phi_j(\boldsymbol\theta)-\bar\phi_j)/s_j$, where the values below use the complete training sample for the printed (full-data) equations. Held-out predictions separately recompute means, standard deviations, and fits within each training fold. Here $\phi_j$ is the feature shown in the table. The score for every estimable cutoff is $q_c=\operatorname{logit}^{-1}(\eta_c)$; positive scores increase the fitted exceedance probability, not the numerical reproduction number.

\begingroup\scriptsize\setlength{\LTcapwidth}{\linewidth}

\begin{longtable}{@{}r p{0.44\linewidth} r r@{}}

\caption{Standardization constants and ordered features for Closed SIR (compact log-cubic basis).}\label{tab:supp-features-closed-sir-fine}\\

\toprule Index & Raw feature $\phi_j$ & Training mean & Training SD \\ \midrule\endfirsthead

\toprule Index & Raw feature $\phi_j$ & Training mean & Training SD \\ \midrule\endhead

 1 & $\log(\beta)$ & -0.842331 & 0.709134 \\

 2 & $\log(\gamma)$ & -0.842330 & 0.709130 \\

 3 & $[\log(\beta)]^2$ & 1.212328 & 1.783167 \\

 4 & $[\log(\gamma)]^2$ & 1.212323 & 1.783137 \\

 5 & $\log(\beta)\,\log(\gamma)$ & 0.709592 & 0.978385 \\

 6 & $[\log(\beta)]^3$ & -2.221997 & 4.505500 \\

 7 & $[\log(\gamma)]^3$ & -2.221968 & 4.505350 \\

\bottomrule\end{longtable}\endgroup

Each displayed score below is a numeric equation in the standardized variables $z_1,\ldots,z_7$. The intercept and slopes are posterior modes from the independent $N(0,2^2)$ coefficient prior under the Gaussian Laplace approximation. A cutoff with only one observed class cannot produce a fitted finite logistic score, so its constant exceedance probability is stated explicitly.

\noindent\textbf{Cutoff $R_0>1$.}
\[\begin{aligned}
\eta_{R_0>1} &= -0.34751 \\
&\quad + 19.62432z_{1} - 19.24511z_{2} - 6.32871z_{3} \\
&\quad + 6.80608z_{4} + 0.28151z_{5} + 3.51080z_{6} \\
&\quad - 2.36030z_{7}
\end{aligned}\]

\noindent\textbf{Cutoff $R_0>1.5$.}
\[\begin{aligned}
\eta_{R_0>1.5} &= -8.31120 \\
&\quad + 13.78345z_{1} - 14.72849z_{2} - 3.12752z_{3} \\
&\quad + 9.28616z_{4} - 7.62345z_{5} + 0.99548z_{6} \\
&\quad - 4.84850z_{7}
\end{aligned}\]

\noindent\textbf{Cutoff $R_0>2$.}
\[\begin{aligned}
\eta_{R_0>2} &= -11.33320 \\
&\quad + 9.10945z_{1} - 10.30802z_{2} - 1.09070z_{3} \\
&\quad + 9.11445z_{4} - 10.90040z_{5} - 0.69440z_{6} \\
&\quad - 5.72695z_{7}
\end{aligned}\]

\noindent\textbf{Cutoff $R_0>2.5$.}
\[\begin{aligned}
\eta_{R_0>2.5} &= -12.23313 \\
&\quad + 5.55116z_{1} - 7.17690z_{2} + 0.14613z_{3} \\
&\quad + 8.53612z_{4} - 11.77120z_{5} - 1.15529z_{6} \\
&\quad - 5.40317z_{7}
\end{aligned}\]

\noindent\textbf{Cutoff $R_0>3$.}
\[\begin{aligned}
\eta_{R_0>3} &= -12.40649 \\
&\quad + 3.45909z_{1} - 4.00045z_{2} + 0.75965z_{3} \\
&\quad + 7.68162z_{4} - 12.02208z_{5} - 1.08228z_{6} \\
&\quad - 6.04115z_{7}
\end{aligned}\]

\noindent\textbf{Cutoff $R_0>3.5$.}
\[\begin{aligned}
\eta_{R_0>3.5} &= -11.77181 \\
&\quad + 1.72503z_{1} - 2.41667z_{2} + 0.87061z_{3} \\
&\quad + 6.48068z_{4} - 11.80882z_{5} - 1.93373z_{6} \\
&\quad - 6.02209z_{7}
\end{aligned}\]

\noindent\textbf{Cutoff $R_0>4$.}
\[\begin{aligned}
\eta_{R_0>4} &= -11.74804 \\
&\quad + 1.35584z_{1} - 0.91787z_{2} + 0.91556z_{3} \\
&\quad + 5.58430z_{4} - 11.03977z_{5} - 1.62785z_{6} \\
&\quad - 6.15810z_{7}
\end{aligned}\]

\noindent\textbf{Cutoff $R_0>4.5$.}
\[\begin{aligned}
\eta_{R_0>4.5} &= -11.89186 \\
&\quad - 0.18872z_{1} - 0.00367z_{2} + 0.42517z_{3} \\
&\quad + 5.22846z_{4} - 10.99795z_{5} - 0.84249z_{6} \\
&\quad - 5.81497z_{7}
\end{aligned}\]

\noindent\textbf{Cutoff $R_0>5$.}
\[\begin{aligned}
\eta_{R_0>5} &= -11.76264 \\
&\quad - 0.06486z_{1} + 1.01319z_{2} + 0.15513z_{3} \\
&\quad + 4.22526z_{4} - 10.25798z_{5} - 0.94971z_{6} \\
&\quad - 6.18443z_{7}
\end{aligned}\]

\noindent\textbf{Cutoff $R_0>5.5$.}
\[\begin{aligned}
\eta_{R_0>5.5} &= -11.61764 \\
&\quad + 0.06604z_{1} + 1.59781z_{2} + 0.42448z_{3} \\
&\quad + 3.61333z_{4} - 9.56914z_{5} - 0.60807z_{6} \\
&\quad - 6.01186z_{7}
\end{aligned}\]

\noindent\textbf{Cutoff $R_0>6$.}
\[\begin{aligned}
\eta_{R_0>6} &= -11.64499 \\
&\quad - 0.13383z_{1} + 1.73018z_{2} + 0.21888z_{3} \\
&\quad + 3.16349z_{4} - 9.23197z_{5} - 0.58794z_{6} \\
&\quad - 5.69038z_{7}
\end{aligned}\]

\noindent\textbf{Cutoff $R_0>6.5$.}
\[\begin{aligned}
\eta_{R_0>6.5} &= -11.56966 \\
&\quad - 0.58602z_{1} + 1.81811z_{2} - 0.07978z_{3} \\
&\quad + 2.77082z_{4} - 8.66368z_{5} - 0.47363z_{6} \\
&\quad - 5.27595z_{7}
\end{aligned}\]

\noindent\textbf{Cutoff $R_0>7$.}
\[\begin{aligned}
\eta_{R_0>7} &= -11.60823 \\
&\quad + 0.59613z_{1} + 2.21612z_{2} + 0.67104z_{3} \\
&\quad + 1.92741z_{4} - 8.07011z_{5} - 0.56137z_{6} \\
&\quad - 5.74410z_{7}
\end{aligned}\]

\noindent\textbf{Cutoff $R_0>7.5$.}
\[\begin{aligned}
\eta_{R_0>7.5} &= -11.40467 \\
&\quad - 0.65982z_{1} + 1.96799z_{2} - 0.26643z_{3} \\
&\quad + 1.98423z_{4} - 7.69791z_{5} - 0.47417z_{6} \\
&\quad - 4.91722z_{7}
\end{aligned}\]

\noindent\textbf{Cutoff $R_0>8$.}
\[\begin{aligned}
\eta_{R_0>8} &= -11.34055 \\
&\quad + 0.65787z_{1} + 2.22889z_{2} + 0.68951z_{3} \\
&\quad + 1.07190z_{4} - 6.89694z_{5} - 0.52209z_{6} \\
&\quad - 5.37517z_{7}
\end{aligned}\]

\noindent\textbf{Cutoff $R_0>8.5$.}
\[\begin{aligned}
\eta_{R_0>8.5} &= -11.25016 \\
&\quad + 0.88623z_{1} + 2.14540z_{2} + 0.88792z_{3} \\
&\quad + 0.67302z_{4} - 6.42783z_{5} - 0.49363z_{6} \\
&\quad - 5.14375z_{7}
\end{aligned}\]

\noindent\textbf{Cutoff $R_0>9$.}
\[\begin{aligned}
\eta_{R_0>9} &= -11.35939 \\
&\quad + 0.42767z_{1} + 2.09376z_{2} + 0.54320z_{3} \\
&\quad + 0.74283z_{4} - 6.60861z_{5} - 0.46606z_{6} \\
&\quad - 4.82771z_{7}
\end{aligned}\]

\noindent\textbf{Cutoff $R_0>9.5$.}
\[\begin{aligned}
\eta_{R_0>9.5} &= -11.31421 \\
&\quad + 0.61092z_{1} + 2.03066z_{2} + 0.72277z_{3} \\
&\quad + 0.53526z_{4} - 6.40347z_{5} - 0.42941z_{6} \\
&\quad - 4.72167z_{7}
\end{aligned}\]

\noindent\textbf{Cutoff $R_0>10$.}
\[\begin{aligned}
\eta_{R_0>10} &= -11.15649 \\
&\quad + 0.40232z_{1} + 2.00522z_{2} + 0.54330z_{3} \\
&\quad + 0.39681z_{4} - 6.01513z_{5} - 0.41552z_{6} \\
&\quad - 4.55482z_{7}
\end{aligned}\]

\subsubsection{Open SIR: validation and equations}

\begingroup\scriptsize\setlength{\LTcapwidth}{\linewidth}\setlength{\tabcolsep}{3pt}

\begin{longtable}{@{}lrrrrrrr@{}}

\caption{Class-specific held-out metrics for Open SIR with log cubic $h=0.5$.}\label{tab:supp-class-log_cubic_h0p5-open-sir}\\

\toprule Bracket & Support & Recall & Spec. & Precision & F1 & NPV & OvR AUC \\ \midrule\endfirsthead

\toprule Bracket & Support & Recall & Spec. & Precision & F1 & NPV & OvR AUC \\ \midrule\endhead

 \texttt{$<$1} & 2709 & 0.992 & 0.997 & 0.994 & 0.993 & 0.996 & 1.000 \\

 \texttt{1-1.5} & 1541 & 0.970 & 0.993 & 0.972 & 0.971 & 0.993 & 0.999 \\

 \texttt{1.5-2} & 1330 & 0.963 & 0.990 & 0.952 & 0.958 & 0.993 & 0.999 \\

 \texttt{2-2.5} & 754 & 0.927 & 0.993 & 0.928 & 0.928 & 0.992 & 0.999 \\

 \texttt{2.5-3} & 506 & 0.923 & 0.994 & 0.912 & 0.917 & 0.995 & 0.999 \\

 \texttt{3-3.5} & 329 & 0.897 & 0.994 & 0.873 & 0.885 & 0.996 & 0.999 \\

 \texttt{3.5-4} & 219 & 0.849 & 0.996 & 0.842 & 0.845 & 0.996 & 0.998 \\

 \texttt{4-4.5} & 142 & 0.796 & 0.995 & 0.753 & 0.774 & 0.996 & 0.997 \\

 \texttt{4.5-5} & 114 & 0.640 & 0.995 & 0.652 & 0.646 & 0.995 & 0.996 \\

 \texttt{5-5.5} & 82 & 0.549 & 0.995 & 0.523 & 0.536 & 0.995 & 0.996 \\

 \texttt{5.5-6} & 53 & 0.396 & 0.995 & 0.356 & 0.375 & 0.996 & 0.985 \\

 \texttt{6-6.5} & 28 & 0.000 & 0.999 & 0.000 & 0.000 & 0.996 & 0.994 \\

 \texttt{6.5-7} & 26 & 0.154 & 0.999 & 0.333 & 0.211 & 0.997 & 0.973 \\

 \texttt{7-7.5} & 27 & 0.444 & 0.997 & 0.333 & 0.381 & 0.998 & 0.995 \\

 \texttt{7.5-8} & 25 & 0.440 & 0.997 & 0.355 & 0.393 & 0.998 & 0.996 \\

 \texttt{8-8.5} & 19 & 0.105 & 0.999 & 0.167 & 0.129 & 0.998 & 0.967 \\

 \texttt{8.5-9} & 15 & 0.000 & 0.999 & 0.000 & 0.000 & 0.998 & 0.958 \\

 \texttt{9-9.5} & 9 & 0.000 & 1.000 & -- & 0.000 & 0.999 & 0.934 \\

 \texttt{9.5-10} & 13 & 0.000 & 1.000 & 0.000 & 0.000 & 0.998 & 0.994 \\

 \texttt{$>$10} & 59 & 1.000 & 0.998 & 0.756 & 0.861 & 1.000 & 1.000 \\

\bottomrule\end{longtable}\endgroup

\paragraph{Ordered-threshold feature map and equations.}

All $12$ fitted predictors are standardized using $z_j=(\phi_j(\boldsymbol\theta)-\bar\phi_j)/s_j$, where the values below use the complete training sample for the printed (full-data) equations. Held-out predictions separately recompute means, standard deviations, and fits within each training fold. Here $\phi_j$ is the feature shown in the table. The score for every estimable cutoff is $q_c=\operatorname{logit}^{-1}(\eta_c)$; positive scores increase the fitted exceedance probability, not the numerical reproduction number.

\begingroup\scriptsize\setlength{\LTcapwidth}{\linewidth}

\begin{longtable}{@{}r p{0.44\linewidth} r r@{}}

\caption{Standardization constants and ordered features for Open SIR (compact log-cubic basis).}\label{tab:supp-features-open-sir-fine}\\

\toprule Index & Raw feature $\phi_j$ & Training mean & Training SD \\ \midrule\endfirsthead

\toprule Index & Raw feature $\phi_j$ & Training mean & Training SD \\ \midrule\endhead

 1 & $\log(\beta)$ & -0.679502 & 0.736099 \\

 2 & $\log(\gamma)$ & -1.353160 & 0.709139 \\

 3 & $\log(\mu)$ & -2.306458 & 0.825462 \\

 4 & $[\log(\beta)]^2$ & 1.003497 & 1.683385 \\

 5 & $[\log(\gamma)]^2$ & 2.333856 & 2.481677 \\

 6 & $[\log(\mu)]^2$ & 6.001050 & 4.735152 \\

 7 & $\log(\beta)\,\log(\gamma)$ & 0.920092 & 1.230820 \\

 8 & $\log(\beta)\,\log(\mu)$ & 1.558951 & 1.860880 \\

 9 & $\log(\gamma)\,\log(\mu)$ & 3.118139 & 2.066217 \\

 10 & $[\log(\beta)]^3$ & -1.835254 & 4.180755 \\

 11 & $[\log(\gamma)]^3$ & -4.872630 & 7.705844 \\

 12 & $[\log(\mu)]^3$ & -17.686315 & 22.716890 \\

\bottomrule\end{longtable}\endgroup

Each displayed score below is a numeric equation in the standardized variables $z_1,\ldots,z_{12}$. The intercept and slopes are posterior modes from the independent $N(0,2^2)$ coefficient prior under the Gaussian Laplace approximation. A cutoff with only one observed class cannot produce a fitted finite logistic score, so its constant exceedance probability is stated explicitly.

\noindent\textbf{Cutoff $R_0>1$.}
\[\begin{aligned}
\eta_{R_0>1} &= 7.46826 \\
&\quad + 14.67004z_{1} - 4.80754z_{2} - 3.50666z_{3} \\
&\quad - 3.78021z_{4} - 1.46874z_{5} - 0.97335z_{6} \\
&\quad + 0.31052z_{7} - 6.55771z_{8} + 16.28174z_{9} \\
&\quad - 2.39530z_{10} + 3.38176z_{11} + 2.39640z_{12}
\end{aligned}\]

\noindent\textbf{Cutoff $R_0>1.5$.}
\[\begin{aligned}
\eta_{R_0>1.5} &= -2.52750 \\
&\quad + 15.28168z_{1} - 6.28693z_{2} - 3.60071z_{3} \\
&\quad + 2.09571z_{4} - 1.10551z_{5} - 2.25299z_{6} \\
&\quad - 3.97028z_{7} - 9.65250z_{8} + 19.33191z_{9} \\
&\quad - 3.43284z_{10} + 3.68635z_{11} + 2.44277z_{12}
\end{aligned}\]

\noindent\textbf{Cutoff $R_0>2$.}
\[\begin{aligned}
\eta_{R_0>2} &= -9.78242 \\
&\quad + 11.24018z_{1} - 5.01661z_{2} - 2.02774z_{3} \\
&\quad + 4.76890z_{4} - 1.65946z_{5} - 2.06027z_{6} \\
&\quad - 5.65853z_{7} - 9.54989z_{8} + 18.16669z_{9} \\
&\quad - 3.56081z_{10} + 2.49063z_{11} + 2.34364z_{12}
\end{aligned}\]

\noindent\textbf{Cutoff $R_0>2.5$.}
\[\begin{aligned}
\eta_{R_0>2.5} &= -12.74869 \\
&\quad + 8.04639z_{1} - 2.63923z_{2} - 1.96278z_{3} \\
&\quad + 5.89073z_{4} - 0.70879z_{5} - 2.97547z_{6} \\
&\quad - 6.14672z_{7} - 8.65213z_{8} + 16.21592z_{9} \\
&\quad - 3.94629z_{10} + 1.74921z_{11} + 2.04892z_{12}
\end{aligned}\]

\noindent\textbf{Cutoff $R_0>3$.}
\[\begin{aligned}
\eta_{R_0>3} &= -14.34196 \\
&\quad + 4.52972z_{1} - 1.94366z_{2} - 1.59386z_{3} \\
&\quad + 4.49552z_{4} - 0.36577z_{5} - 3.08799z_{6} \\
&\quad - 6.22340z_{7} - 7.64633z_{8} + 14.71490z_{9} \\
&\quad - 4.26327z_{10} + 2.07159z_{11} + 2.03672z_{12}
\end{aligned}\]

\noindent\textbf{Cutoff $R_0>3.5$.}
\[\begin{aligned}
\eta_{R_0>3.5} &= -14.73691 \\
&\quad + 3.21318z_{1} - 1.19829z_{2} - 1.34703z_{3} \\
&\quad + 5.05645z_{4} + 0.28315z_{5} - 2.01857z_{6} \\
&\quad - 5.64727z_{7} - 6.84405z_{8} + 12.22451z_{9} \\
&\quad - 4.16379z_{10} + 1.61093z_{11} + 2.09733z_{12}
\end{aligned}\]

\noindent\textbf{Cutoff $R_0>4$.}
\[\begin{aligned}
\eta_{R_0>4} &= -14.94549 \\
&\quad + 1.22892z_{1} - 0.52606z_{2} - 1.22069z_{3} \\
&\quad + 4.15366z_{4} + 0.59584z_{5} - 2.49942z_{6} \\
&\quad - 5.17052z_{7} - 6.75139z_{8} + 11.19494z_{9} \\
&\quad - 3.43573z_{10} + 1.87841z_{11} + 2.06700z_{12}
\end{aligned}\]

\noindent\textbf{Cutoff $R_0>4.5$.}
\[\begin{aligned}
\eta_{R_0>4.5} &= -14.42461 \\
&\quad + 1.60640z_{1} + 0.73057z_{2} - 0.28356z_{3} \\
&\quad + 4.19910z_{4} + 0.08508z_{5} - 1.95247z_{6} \\
&\quad - 5.53635z_{7} - 4.16867z_{8} + 10.77750z_{9} \\
&\quad - 2.67829z_{10} + 0.89056z_{11} + 1.90623z_{12}
\end{aligned}\]

\noindent\textbf{Cutoff $R_0>5$.}
\[\begin{aligned}
\eta_{R_0>5} &= -13.98378 \\
&\quad + 0.04705z_{1} + 0.80433z_{2} - 0.90402z_{3} \\
&\quad + 3.11016z_{4} + 0.31009z_{5} - 1.61285z_{6} \\
&\quad - 5.47285z_{7} - 3.57261z_{8} + 9.41207z_{9} \\
&\quad - 2.30124z_{10} + 0.91419z_{11} + 2.45694z_{12}
\end{aligned}\]

\noindent\textbf{Cutoff $R_0>5.5$.}
\[\begin{aligned}
\eta_{R_0>5.5} &= -13.40710 \\
&\quad + 0.11582z_{1} + 1.65024z_{2} - 0.77982z_{3} \\
&\quad + 3.03383z_{4} - 0.05217z_{5} - 1.45853z_{6} \\
&\quad - 4.49549z_{7} - 3.82273z_{8} + 8.59869z_{9} \\
&\quad - 1.99542z_{10} + 0.25234z_{11} + 2.38597z_{12}
\end{aligned}\]

\noindent\textbf{Cutoff $R_0>6$.}
\[\begin{aligned}
\eta_{R_0>6} &= -12.68590 \\
&\quad + 0.07772z_{1} + 2.28098z_{2} - 0.26962z_{3} \\
&\quad + 2.70355z_{4} - 0.28378z_{5} - 1.33522z_{6} \\
&\quad - 4.27413z_{7} - 2.89790z_{8} + 7.99950z_{9} \\
&\quad - 1.58269z_{10} - 0.35017z_{11} + 2.02158z_{12}
\end{aligned}\]

\noindent\textbf{Cutoff $R_0>6.5$.}
\[\begin{aligned}
\eta_{R_0>6.5} &= -12.71323 \\
&\quad - 0.67868z_{1} + 2.06687z_{2} - 0.34264z_{3} \\
&\quad + 2.44411z_{4} + 0.12839z_{5} - 1.05588z_{6} \\
&\quad - 4.80824z_{7} - 2.60966z_{8} + 7.14850z_{9} \\
&\quad - 1.35705z_{10} - 0.23081z_{11} + 1.82279z_{12}
\end{aligned}\]

\noindent\textbf{Cutoff $R_0>7$.}
\[\begin{aligned}
\eta_{R_0>7} &= -12.25940 \\
&\quad - 0.41377z_{1} + 2.44814z_{2} - 0.10028z_{3} \\
&\quad + 2.45227z_{4} - 0.02936z_{5} - 0.86004z_{6} \\
&\quad - 4.48809z_{7} - 2.19277z_{8} + 6.78644z_{9} \\
&\quad - 1.11175z_{10} - 0.55255z_{11} + 1.72479z_{12}
\end{aligned}\]

\noindent\textbf{Cutoff $R_0>7.5$.}
\[\begin{aligned}
\eta_{R_0>7.5} &= -11.77929 \\
&\quad - 1.87693z_{1} + 2.05054z_{2} + 0.15997z_{3} \\
&\quad + 1.33839z_{4} + 0.17288z_{5} - 0.90844z_{6} \\
&\quad - 4.28131z_{7} - 2.34163z_{8} + 6.41912z_{9} \\
&\quad - 1.03298z_{10} - 0.02127z_{11} + 1.38512z_{12}
\end{aligned}\]

\noindent\textbf{Cutoff $R_0>8$.}
\[\begin{aligned}
\eta_{R_0>8} &= -11.26817 \\
&\quad - 1.31291z_{1} + 2.57897z_{2} + 0.10973z_{3} \\
&\quad + 1.64383z_{4} - 0.16602z_{5} - 0.87938z_{6} \\
&\quad - 3.92476z_{7} - 2.09698z_{8} + 6.08289z_{9} \\
&\quad - 0.74022z_{10} - 0.56576z_{11} + 1.52204z_{12}
\end{aligned}\]

\noindent\textbf{Cutoff $R_0>8.5$.}
\[\begin{aligned}
\eta_{R_0>8.5} &= -10.92068 \\
&\quad - 0.96101z_{1} + 2.77359z_{2} + 0.35495z_{3} \\
&\quad + 1.67946z_{4} - 0.29521z_{5} - 0.72469z_{6} \\
&\quad - 3.95367z_{7} - 1.52932z_{8} + 5.70355z_{9} \\
&\quad - 0.56698z_{10} - 0.80184z_{11} + 1.29713z_{12}
\end{aligned}\]

\noindent\textbf{Cutoff $R_0>9$.}
\[\begin{aligned}
\eta_{R_0>9} &= -10.75283 \\
&\quad - 1.05759z_{1} + 2.78763z_{2} + 0.53779z_{3} \\
&\quad + 1.51999z_{4} - 0.30603z_{5} - 0.55562z_{6} \\
&\quad - 3.91958z_{7} - 1.36218z_{8} + 5.34927z_{9} \\
&\quad - 0.49691z_{10} - 0.84903z_{11} + 1.05479z_{12}
\end{aligned}\]

\noindent\textbf{Cutoff $R_0>9.5$.}
\[\begin{aligned}
\eta_{R_0>9.5} &= -10.67688 \\
&\quad - 1.42355z_{1} + 3.02094z_{2} + 0.03277z_{3} \\
&\quad + 1.33326z_{4} - 0.43930z_{5} - 0.28050z_{6} \\
&\quad - 3.42222z_{7} - 2.08096z_{8} + 4.98203z_{9} \\
&\quad - 0.56787z_{10} - 1.19003z_{11} + 1.40116z_{12}
\end{aligned}\]

\noindent\textbf{Cutoff $R_0>10$.}
\[\begin{aligned}
\eta_{R_0>10} &= -10.42224 \\
&\quad - 1.06967z_{1} + 3.00898z_{2} + 0.14322z_{3} \\
&\quad + 1.54254z_{4} - 0.54871z_{5} - 0.30249z_{6} \\
&\quad - 3.10851z_{7} - 1.88171z_{8} + 4.92761z_{9} \\
&\quad - 0.38738z_{10} - 1.09448z_{11} + 1.44062z_{12}
\end{aligned}\]

\subsubsection{Open SEIR: validation and equations}

\begingroup\scriptsize\setlength{\LTcapwidth}{\linewidth}\setlength{\tabcolsep}{3pt}

\begin{longtable}{@{}lrrrrrrr@{}}

\caption{Class-specific held-out metrics for Open SEIR with log cubic $h=0.5$.}\label{tab:supp-class-log_cubic_h0p5-open-seir}\\

\toprule Bracket & Support & Recall & Spec. & Precision & F1 & NPV & OvR AUC \\ \midrule\endfirsthead

\toprule Bracket & Support & Recall & Spec. & Precision & F1 & NPV & OvR AUC \\ \midrule\endhead

 \texttt{$<$1} & 2764 & 0.991 & 0.996 & 0.993 & 0.992 & 0.995 & 1.000 \\

 \texttt{1-1.5} & 1412 & 0.969 & 0.992 & 0.963 & 0.966 & 0.993 & 0.999 \\

 \texttt{1.5-2} & 1132 & 0.956 & 0.993 & 0.957 & 0.956 & 0.993 & 0.999 \\

 \texttt{2-2.5} & 771 & 0.944 & 0.992 & 0.924 & 0.934 & 0.994 & 0.998 \\

 \texttt{2.5-3} & 503 & 0.889 & 0.993 & 0.898 & 0.893 & 0.993 & 0.998 \\

 \texttt{3-3.5} & 341 & 0.871 & 0.993 & 0.856 & 0.863 & 0.994 & 0.995 \\

 \texttt{3.5-4} & 246 & 0.841 & 0.994 & 0.812 & 0.826 & 0.995 & 0.997 \\

 \texttt{4-4.5} & 176 & 0.750 & 0.995 & 0.754 & 0.752 & 0.994 & 0.996 \\

 \texttt{4.5-5} & 148 & 0.723 & 0.994 & 0.699 & 0.711 & 0.995 & 0.995 \\

 \texttt{5-5.5} & 98 & 0.531 & 0.995 & 0.578 & 0.553 & 0.994 & 0.981 \\

 \texttt{5.5-6} & 72 & 0.486 & 0.995 & 0.486 & 0.486 & 0.995 & 0.995 \\

 \texttt{6-6.5} & 59 & 0.390 & 0.995 & 0.371 & 0.380 & 0.995 & 0.972 \\

 \texttt{6.5-7} & 46 & 0.217 & 0.996 & 0.227 & 0.222 & 0.995 & 0.981 \\

 \texttt{7-7.5} & 39 & 0.231 & 0.996 & 0.237 & 0.234 & 0.996 & 0.991 \\

 \texttt{7.5-8} & 27 & 0.148 & 0.998 & 0.200 & 0.170 & 0.997 & 0.974 \\

 \texttt{8-8.5} & 19 & 0.000 & 0.999 & 0.000 & 0.000 & 0.998 & 0.994 \\

 \texttt{8.5-9} & 19 & 0.105 & 0.999 & 0.200 & 0.138 & 0.998 & 0.996 \\

 \texttt{9-9.5} & 14 & 0.286 & 1.000 & 0.800 & 0.421 & 0.999 & 0.996 \\

 \texttt{9.5-10} & 21 & 0.095 & 0.999 & 0.250 & 0.138 & 0.998 & 0.996 \\

 \texttt{$>$10} & 93 & 0.978 & 0.996 & 0.765 & 0.858 & 1.000 & 1.000 \\

\bottomrule\end{longtable}\endgroup

\paragraph{Ordered-threshold feature map and equations.}

All $18$ fitted predictors are standardized using $z_j=(\phi_j(\boldsymbol\theta)-\bar\phi_j)/s_j$, where the values below use the complete training sample for the printed (full-data) equations. Held-out predictions separately recompute means, standard deviations, and fits within each training fold. Here $\phi_j$ is the feature shown in the table. The score for every estimable cutoff is $q_c=\operatorname{logit}^{-1}(\eta_c)$; positive scores increase the fitted exceedance probability, not the numerical reproduction number.

\begingroup\scriptsize\setlength{\LTcapwidth}{\linewidth}

\begin{longtable}{@{}r p{0.44\linewidth} r r@{}}

\caption{Standardization constants and ordered features for Open SEIR (compact log-cubic basis).}\label{tab:supp-features-open-seir-fine}\\

\toprule Index & Raw feature $\phi_j$ & Training mean & Training SD \\ \midrule\endfirsthead

\toprule Index & Raw feature $\phi_j$ & Training mean & Training SD \\ \midrule\endhead

 1 & $\log(\beta)$ & -0.477253 & 0.766434 \\

 2 & $\log(\sigma)$ & -0.842332 & 0.709136 \\

 3 & $\log(\gamma)$ & -1.353154 & 0.709127 \\

 4 & $\log(\mu)$ & -2.514853 & 0.801389 \\

 5 & $[\log(\beta)]^2$ & 0.815118 & 1.543997 \\

 6 & $[\log(\sigma)]^2$ & 1.212334 & 1.783176 \\

 7 & $[\log(\gamma)]^2$ & 2.333824 & 2.481591 \\

 8 & $[\log(\mu)]^2$ & 6.966631 & 4.857262 \\

 9 & $\log(\beta)\,\log(\sigma)$ & 0.407571 & 0.934930 \\

 10 & $\log(\beta)\,\log(\gamma)$ & 0.651759 & 1.221282 \\

 11 & $\log(\beta)\,\log(\mu)$ & 1.194157 & 2.058436 \\

 12 & $\log(\sigma)\,\log(\gamma)$ & 1.143001 & 1.240328 \\

 13 & $\log(\sigma)\,\log(\mu)$ & 2.116874 & 1.992866 \\

 14 & $\log(\gamma)\,\log(\mu)$ & 3.411148 & 2.192436 \\

 15 & $[\log(\beta)]^3$ & -1.449334 & 3.799273 \\

 16 & $[\log(\sigma)]^3$ & -2.222014 & 4.505534 \\

 17 & $[\log(\gamma)]^3$ & -4.872489 & 7.705370 \\

 18 & $[\log(\mu)]^3$ & -21.362065 & 24.127790 \\

\bottomrule\end{longtable}\endgroup

Each displayed score below is a numeric equation in the standardized variables $z_1,\ldots,z_18$. The intercept and slopes are posterior modes from the independent $N(0,2^2)$ coefficient prior under the Gaussian Laplace approximation. A cutoff with only one observed class cannot produce a fitted finite logistic score, so its constant exceedance probability is stated explicitly.

\noindent\textbf{Cutoff $R_0>1$.}
\[\begin{aligned}
\eta_{R_0>1} &= 7.61537 \\
&\quad + 13.15205z_{1} + 4.99792z_{2} - 4.81975z_{3} \\
&\quad - 6.14029z_{4} - 3.31089z_{5} - 3.53457z_{6} \\
&\quad - 1.53577z_{7} - 0.40785z_{8} + 0.32218z_{9} \\
&\quad + 1.30968z_{10} - 8.53346z_{11} - 0.92676z_{12} \\
&\quad + 4.91187z_{13} + 15.47156z_{14} - 1.71192z_{15} \\
&\quad + 0.11273z_{16} + 3.10106z_{17} + 4.50490z_{18}
\end{aligned}\]

\noindent\textbf{Cutoff $R_0>1.5$.}
\[\begin{aligned}
\eta_{R_0>1.5} &= -1.48117 \\
&\quad + 13.83850z_{1} + 5.12702z_{2} - 5.13761z_{3} \\
&\quad - 5.27807z_{4} + 1.73180z_{5} - 1.48394z_{6} \\
&\quad - 1.06686z_{7} - 0.73552z_{8} + 2.77381z_{9} \\
&\quad - 2.32016z_{10} - 11.16631z_{11} - 2.49469z_{12} \\
&\quad + 4.36476z_{13} + 18.50167z_{14} - 1.70164z_{15} \\
&\quad + 0.96761z_{16} + 2.91866z_{17} + 4.77531z_{18}
\end{aligned}\]

\noindent\textbf{Cutoff $R_0>2$.}
\[\begin{aligned}
\eta_{R_0>2} &= -8.46123 \\
&\quad + 10.82044z_{1} + 3.83327z_{2} - 4.44026z_{3} \\
&\quad - 3.85080z_{4} + 4.58944z_{5} - 2.26666z_{6} \\
&\quad - 1.51424z_{7} - 1.84726z_{8} + 2.21702z_{9} \\
&\quad - 3.60000z_{10} - 11.55027z_{11} - 2.05197z_{12} \\
&\quad + 3.52457z_{13} + 17.95519z_{14} - 1.53413z_{15} \\
&\quad + 0.27559z_{16} + 2.67192z_{17} + 3.08315z_{18}
\end{aligned}\]

\noindent\textbf{Cutoff $R_0>2.5$.}
\[\begin{aligned}
\eta_{R_0>2.5} &= -12.16089 \\
&\quad + 7.11956z_{1} + 2.40991z_{2} - 3.76540z_{3} \\
&\quad - 3.08862z_{4} + 4.98158z_{5} - 2.40321z_{6} \\
&\quad - 0.89862z_{7} - 1.94677z_{8} + 2.21725z_{9} \\
&\quad - 3.80678z_{10} - 11.37089z_{11} - 1.90173z_{12} \\
&\quad + 2.88223z_{13} + 16.14232z_{14} - 1.56863z_{15} \\
&\quad + 0.12823z_{16} + 2.65558z_{17} + 3.02464z_{18}
\end{aligned}\]

\noindent\textbf{Cutoff $R_0>3$.}
\[\begin{aligned}
\eta_{R_0>3} &= -13.23081 \\
&\quad + 4.44267z_{1} + 1.22538z_{2} - 3.24243z_{3} \\
&\quad - 2.50751z_{4} + 4.36870z_{5} - 1.97247z_{6} \\
&\quad - 0.70961z_{7} - 2.11416z_{8} + 2.30879z_{9} \\
&\quad - 4.47886z_{10} - 10.27910z_{11} - 2.53294z_{12} \\
&\quad + 2.32673z_{13} + 14.50061z_{14} - 3.81251z_{15} \\
&\quad - 0.33482z_{16} + 2.60314z_{17} + 2.82373z_{18}
\end{aligned}\]

\noindent\textbf{Cutoff $R_0>3.5$.}
\[\begin{aligned}
\eta_{R_0>3.5} &= -14.26481 \\
&\quad + 3.45941z_{1} + 1.09227z_{2} - 1.43812z_{3} \\
&\quad - 1.90854z_{4} + 5.09495z_{5} - 1.48148z_{6} \\
&\quad - 0.73096z_{7} - 2.23212z_{8} + 2.14139z_{9} \\
&\quad - 4.39933z_{10} - 8.85713z_{11} - 1.77001z_{12} \\
&\quad + 2.18923z_{13} + 13.94741z_{14} - 2.75384z_{15} \\
&\quad + 0.69227z_{16} + 1.76795z_{17} + 2.81623z_{18}
\end{aligned}\]

\noindent\textbf{Cutoff $R_0>4$.}
\[\begin{aligned}
\eta_{R_0>4} &= -14.90892 \\
&\quad + 2.52036z_{1} + 0.91588z_{2} - 0.82768z_{3} \\
&\quad - 1.86895z_{4} + 4.21874z_{5} - 0.60394z_{6} \\
&\quad - 0.07956z_{7} - 2.19795z_{8} + 2.47539z_{9} \\
&\quad - 4.60225z_{10} - 7.93921z_{11} - 1.78932z_{12} \\
&\quad + 1.80030z_{13} + 12.51278z_{14} - 3.10742z_{15} \\
&\quad + 0.73656z_{16} + 1.63208z_{17} + 2.64583z_{18}
\end{aligned}\]

\noindent\textbf{Cutoff $R_0>4.5$.}
\[\begin{aligned}
\eta_{R_0>4.5} &= -14.97941 \\
&\quad + 1.10993z_{1} + 1.11319z_{2} - 0.48294z_{3} \\
&\quad - 1.64215z_{4} + 4.30148z_{5} - 0.51106z_{6} \\
&\quad + 0.21925z_{7} - 2.35551z_{8} + 1.56731z_{9} \\
&\quad - 4.01830z_{10} - 7.74196z_{11} - 1.61890z_{12} \\
&\quad + 1.92810z_{13} + 11.33100z_{14} - 2.63897z_{15} \\
&\quad + 0.63610z_{16} + 1.67303z_{17} + 2.35924z_{18}
\end{aligned}\]

\noindent\textbf{Cutoff $R_0>5$.}
\[\begin{aligned}
\eta_{R_0>5} &= -14.92442 \\
&\quad + 0.99574z_{1} + 0.37622z_{2} + 0.27043z_{3} \\
&\quad - 1.77643z_{4} + 3.88977z_{5} - 0.14813z_{6} \\
&\quad + 0.29164z_{7} - 1.56812z_{8} + 1.53629z_{9} \\
&\quad - 4.31594z_{10} - 5.80799z_{11} - 1.42636z_{12} \\
&\quad + 1.24860z_{13} + 10.54033z_{14} - 2.25434z_{15} \\
&\quad + 0.86253z_{16} + 1.22357z_{17} + 3.14667z_{18}
\end{aligned}\]

\noindent\textbf{Cutoff $R_0>5.5$.}
\[\begin{aligned}
\eta_{R_0>5.5} &= -14.64160 \\
&\quad + 0.61322z_{1} - 0.31178z_{2} + 0.94512z_{3} \\
&\quad - 1.34666z_{4} + 3.15830z_{5} - 0.25228z_{6} \\
&\quad + 0.30154z_{7} - 1.49016z_{8} + 1.66340z_{9} \\
&\quad - 4.04039z_{10} - 5.29195z_{11} - 1.04728z_{12} \\
&\quad + 0.33847z_{13} + 9.90356z_{14} - 2.03526z_{15} \\
&\quad + 0.72028z_{16} + 0.96043z_{17} + 2.71722z_{18}
\end{aligned}\]

\noindent\textbf{Cutoff $R_0>6$.}
\[\begin{aligned}
\eta_{R_0>6} &= -14.14785 \\
&\quad + 0.42771z_{1} + 0.22700z_{2} + 1.36384z_{3} \\
&\quad - 0.70714z_{4} + 3.33586z_{5} - 0.30747z_{6} \\
&\quad + 0.07500z_{7} - 1.47089z_{8} + 1.29237z_{9} \\
&\quad - 3.93111z_{10} - 4.55304z_{11} - 1.33819z_{12} \\
&\quad + 1.01134z_{13} + 9.34741z_{14} - 1.64972z_{15} \\
&\quad + 0.24918z_{16} + 0.53418z_{17} + 2.09691z_{18}
\end{aligned}\]

\noindent\textbf{Cutoff $R_0>6.5$.}
\[\begin{aligned}
\eta_{R_0>6.5} &= -13.78885 \\
&\quad - 0.15942z_{1} + 0.01942z_{2} + 1.86441z_{3} \\
&\quad - 0.48799z_{4} + 2.20282z_{5} - 0.47425z_{6} \\
&\quad + 0.07331z_{7} - 1.48487z_{8} + 1.77983z_{9} \\
&\quad - 3.96485z_{10} - 4.54298z_{11} - 1.15372z_{12} \\
&\quad + 0.92337z_{13} + 8.57049z_{14} - 1.55815z_{15} \\
&\quad + 0.19134z_{16} + 0.15678z_{17} + 1.78066z_{18}
\end{aligned}\]

\noindent\textbf{Cutoff $R_0>7$.}
\[\begin{aligned}
\eta_{R_0>7} &= -13.34311 \\
&\quad - 0.57258z_{1} + 0.07888z_{2} + 1.83517z_{3} \\
&\quad - 0.69168z_{4} + 2.29410z_{5} + 0.02717z_{6} \\
&\quad + 0.12507z_{7} - 1.04996z_{8} + 1.06597z_{9} \\
&\quad - 3.85897z_{10} - 3.91541z_{11} - 1.16081z_{12} \\
&\quad + 0.90859z_{13} + 7.83067z_{14} - 1.51414z_{15} \\
&\quad + 0.41813z_{16} + 0.09742z_{17} + 2.16759z_{18}
\end{aligned}\]

\noindent\textbf{Cutoff $R_0>7.5$.}
\[\begin{aligned}
\eta_{R_0>7.5} &= -13.00011 \\
&\quad - 0.52544z_{1} - 0.17429z_{2} + 2.04516z_{3} \\
&\quad - 0.39999z_{4} + 2.28071z_{5} - 0.51586z_{6} \\
&\quad - 0.01578z_{7} - 0.97979z_{8} + 1.02648z_{9} \\
&\quad - 4.13945z_{10} - 3.09698z_{11} - 0.97057z_{12} \\
&\quad + 1.01137z_{13} + 7.54307z_{14} - 1.21487z_{15} \\
&\quad + 0.58735z_{16} + 0.12749z_{17} + 2.02863z_{18}
\end{aligned}\]

\noindent\textbf{Cutoff $R_0>8$.}
\[\begin{aligned}
\eta_{R_0>8} &= -12.45383 \\
&\quad - 1.32431z_{1} - 0.52920z_{2} + 2.35454z_{3} \\
&\quad - 0.27695z_{4} + 1.69794z_{5} - 0.93394z_{6} \\
&\quad - 0.07072z_{7} - 1.06200z_{8} + 0.82849z_{9} \\
&\quad - 3.97040z_{10} - 3.15638z_{11} - 1.08943z_{12} \\
&\quad + 0.77209z_{13} + 7.29516z_{14} - 1.10227z_{15} \\
&\quad + 0.04182z_{16} - 0.08497z_{17} + 1.87452z_{18}
\end{aligned}\]

\noindent\textbf{Cutoff $R_0>8.5$.}
\[\begin{aligned}
\eta_{R_0>8.5} &= -12.15751 \\
&\quad - 1.49561z_{1} - 0.40847z_{2} + 2.68150z_{3} \\
&\quad - 0.19105z_{4} + 1.42247z_{5} - 0.64972z_{6} \\
&\quad - 0.15216z_{7} - 0.87716z_{8} + 0.89753z_{9} \\
&\quad - 3.94699z_{10} - 2.79257z_{11} - 1.38331z_{12} \\
&\quad + 1.07115z_{13} + 7.07574z_{14} - 0.87976z_{15} \\
&\quad + 0.03836z_{16} - 0.44522z_{17} + 1.91429z_{18}
\end{aligned}\]

\noindent\textbf{Cutoff $R_0>9$.}
\[\begin{aligned}
\eta_{R_0>9} &= -11.87523 \\
&\quad - 2.27237z_{1} - 0.69523z_{2} + 2.79985z_{3} \\
&\quad - 0.68522z_{4} + 1.04550z_{5} - 0.34040z_{6} \\
&\quad - 0.05267z_{7} - 0.61450z_{8} + 0.30564z_{9} \\
&\quad - 3.20408z_{10} - 3.35207z_{11} - 1.30507z_{12} \\
&\quad + 0.85043z_{13} + 6.49329z_{14} - 1.07379z_{15} \\
&\quad + 0.18961z_{16} - 0.68091z_{17} + 2.25835z_{18}
\end{aligned}\]

\noindent\textbf{Cutoff $R_0>9.5$.}
\[\begin{aligned}
\eta_{R_0>9.5} &= -11.68178 \\
&\quad - 1.46147z_{1} - 0.31885z_{2} + 2.72866z_{3} \\
&\quad - 0.43095z_{4} + 1.64504z_{5} + 0.10428z_{6} \\
&\quad - 0.15411z_{7} - 0.52750z_{8} + 0.50823z_{9} \\
&\quad - 3.23699z_{10} - 2.99536z_{11} - 1.03130z_{12} \\
&\quad + 0.73267z_{13} + 6.09032z_{14} - 0.77215z_{15} \\
&\quad + 0.31167z_{16} - 0.69654z_{17} + 1.96351z_{18}
\end{aligned}\]

\noindent\textbf{Cutoff $R_0>10$.}
\[\begin{aligned}
\eta_{R_0>10} &= -11.17800 \\
&\quad - 1.06124z_{1} - 0.73008z_{2} + 2.87753z_{3} \\
&\quad - 0.20564z_{4} + 2.06172z_{5} + 0.23030z_{6} \\
&\quad - 0.45098z_{7} - 0.55088z_{8} - 0.01853z_{9} \\
&\quad - 3.46461z_{10} - 2.08935z_{11} - 0.59849z_{12} \\
&\quad - 0.13861z_{13} + 5.97095z_{14} - 0.43742z_{15} \\
&\quad + 0.39063z_{16} - 0.65807z_{17} + 1.73046z_{18}
\end{aligned}\]

\subsubsection{SEIR with vaccination: validation and equations}

\begingroup\scriptsize\setlength{\LTcapwidth}{\linewidth}\setlength{\tabcolsep}{3pt}

\begin{longtable}{@{}lrrrrrrr@{}}

\caption{Class-specific held-out metrics for SEIR with vaccination with log cubic $h=0.5$.}\label{tab:supp-class-log_cubic_h0p5-seir-with-vaccination}\\

\toprule Bracket & Support & Recall & Spec. & Precision & F1 & NPV & OvR AUC \\ \midrule\endfirsthead

\toprule Bracket & Support & Recall & Spec. & Precision & F1 & NPV & OvR AUC \\ \midrule\endhead

 \texttt{$<$1} & 2221 & 0.985 & 0.995 & 0.987 & 0.986 & 0.994 & 1.000 \\

 \texttt{1-1.5} & 1108 & 0.940 & 0.990 & 0.939 & 0.940 & 0.990 & 0.997 \\

 \texttt{1.5-2} & 918 & 0.910 & 0.989 & 0.917 & 0.913 & 0.988 & 0.995 \\

 \texttt{2-2.5} & 833 & 0.897 & 0.987 & 0.887 & 0.892 & 0.988 & 0.993 \\

 \texttt{2.5-3} & 674 & 0.878 & 0.986 & 0.853 & 0.865 & 0.989 & 0.993 \\

 \texttt{3-3.5} & 527 & 0.822 & 0.989 & 0.836 & 0.829 & 0.987 & 0.992 \\

 \texttt{3.5-4} & 408 & 0.806 & 0.986 & 0.755 & 0.780 & 0.990 & 0.990 \\

 \texttt{4-4.5} & 288 & 0.660 & 0.989 & 0.683 & 0.671 & 0.987 & 0.989 \\

 \texttt{4.5-5} & 218 & 0.638 & 0.990 & 0.638 & 0.638 & 0.990 & 0.990 \\

 \texttt{5-5.5} & 170 & 0.576 & 0.990 & 0.551 & 0.563 & 0.991 & 0.989 \\

 \texttt{5.5-6} & 121 & 0.463 & 0.993 & 0.496 & 0.479 & 0.992 & 0.988 \\

 \texttt{6-6.5} & 107 & 0.495 & 0.992 & 0.453 & 0.473 & 0.993 & 0.984 \\

 \texttt{6.5-7} & 79 & 0.316 & 0.994 & 0.329 & 0.323 & 0.993 & 0.990 \\

 \texttt{7-7.5} & 61 & 0.164 & 0.996 & 0.227 & 0.190 & 0.994 & 0.979 \\

 \texttt{7.5-8} & 49 & 0.408 & 0.995 & 0.323 & 0.360 & 0.996 & 0.991 \\

 \texttt{8-8.5} & 33 & 0.091 & 0.998 & 0.158 & 0.115 & 0.996 & 0.955 \\

 \texttt{8.5-9} & 37 & 0.189 & 0.998 & 0.292 & 0.230 & 0.996 & 0.963 \\

 \texttt{9-9.5} & 21 & 0.000 & 1.000 & -- & 0.000 & 0.997 & 0.989 \\

 \texttt{9.5-10} & 16 & 0.000 & 1.000 & 0.000 & 0.000 & 0.998 & 0.956 \\

 \texttt{$>$10} & 111 & 0.982 & 0.996 & 0.768 & 0.862 & 1.000 & 1.000 \\

\bottomrule\end{longtable}\endgroup

\paragraph{Ordered-threshold feature map and equations.}

All $25$ fitted predictors are standardized using $z_j=(\phi_j(\boldsymbol\theta)-\bar\phi_j)/s_j$, where the values below use the complete training sample for the printed (full-data) equations. Held-out predictions separately recompute means, standard deviations, and fits within each training fold. Here $\phi_j$ is the feature shown in the table. The score for every estimable cutoff is $q_c=\operatorname{logit}^{-1}(\eta_c)$; positive scores increase the fitted exceedance probability, not the numerical reproduction number.

\begingroup\scriptsize\setlength{\LTcapwidth}{\linewidth}

\begin{longtable}{@{}r p{0.44\linewidth} r r@{}}

\caption{Standardization constants and ordered features for SEIR with vaccination (compact log-cubic basis).}\label{tab:supp-features-seir-with-vaccination-fine}\\

\toprule Index & Raw feature $\phi_j$ & Training mean & Training SD \\ \midrule\endfirsthead

\toprule Index & Raw feature $\phi_j$ & Training mean & Training SD \\ \midrule\endhead

 1 & $\log(\beta)$ & 0.480882 & 0.801387 \\

 2 & $\log(\sigma)$ & -0.842327 & 0.709129 \\

 3 & $\log(\gamma)$ & -1.353152 & 0.709124 \\

 4 & $\log(\mu)$ & -2.086690 & 0.766429 \\

 5 & $\log(v)$ & -2.863486 & 0.911286 \\

 6 & $[\log(\beta)]^2$ & 0.873388 & 0.776560 \\

 7 & $[\log(\sigma)]^2$ & 1.212316 & 1.783144 \\

 8 & $[\log(\gamma)]^2$ & 2.333814 & 2.481577 \\

 9 & $[\log(\mu)]^2$ & 4.941614 & 3.911104 \\

 10 & $[\log(v)]^2$ & 9.029888 & 6.574668 \\

 11 & $\log(\beta)\,\log(\sigma)$ & -0.401420 & 0.948040 \\

 12 & $\log(\beta)\,\log(\gamma)$ & -0.644841 & 1.281607 \\

 13 & $\log(\beta)\,\log(\mu)$ & -1.013175 & 1.802206 \\

 14 & $\log(\beta)\,\log(v)$ & -1.366084 & 2.461464 \\

 15 & $\log(\sigma)\,\log(\gamma)$ & 1.140857 & 1.228180 \\

 16 & $\log(\sigma)\,\log(\mu)$ & 1.765326 & 1.714641 \\

 17 & $\log(\sigma)\,\log(v)$ & 2.417957 & 2.275169 \\

 18 & $\log(\gamma)\,\log(\mu)$ & 2.821804 & 1.902196 \\

 19 & $\log(\gamma)\,\log(v)$ & 3.879660 & 2.483768 \\

 20 & $\log(\mu)\,\log(v)$ & 5.971571 & 2.975017 \\

 21 & $[\log(\beta)]^3$ & 0.425374 & 1.668421 \\

 22 & $[\log(\sigma)]^3$ & -2.221959 & 4.505420 \\

 23 & $[\log(\gamma)]^3$ & -4.872454 & 7.705320 \\

 24 & $[\log(\mu)]^3$ & -13.262471 & 16.642864 \\

 25 & $[\log(v)]^3$ & -31.750774 & 39.592829 \\

\bottomrule\end{longtable}\endgroup

Each displayed score below is a numeric equation in the standardized variables $z_1,\ldots,z_25$. The intercept and slopes are posterior modes from the independent $N(0,2^2)$ coefficient prior under the Gaussian Laplace approximation. A cutoff with only one observed class cannot produce a fitted finite logistic score, so its constant exceedance probability is stated explicitly.

\noindent\textbf{Cutoff $R_e(0)>1$.}
\[\begin{aligned}
\eta_{R_e(0)>1} &= 10.71780 \\
&\quad + 8.27007z_{1} + 3.23438z_{2} - 2.99121z_{3} \\
&\quad + 2.20903z_{4} - 0.23532z_{5} + 0.62097z_{6} \\
&\quad - 1.12286z_{7} - 1.54830z_{8} - 7.31001z_{9} \\
&\quad - 2.36950z_{10} + 1.71786z_{11} + 0.19775z_{12} \\
&\quad - 2.43881z_{13} - 7.02943z_{14} - 0.64259z_{15} \\
&\quad + 4.97721z_{16} - 3.65898z_{17} + 8.06506z_{18} \\
&\quad + 3.31982z_{19} + 13.63568z_{20} + 0.15899z_{21} \\
&\quad + 0.59597z_{22} + 2.44175z_{23} + 7.05001z_{24} \\
&\quad + 1.92850z_{25}
\end{aligned}\]

\noindent\textbf{Cutoff $R_e(0)>1.5$.}
\[\begin{aligned}
\eta_{R_e(0)>1.5} &= 4.20789 \\
&\quad + 6.78777z_{1} + 2.58165z_{2} - 3.30046z_{3} \\
&\quad + 1.53027z_{4} + 1.14170z_{5} + 2.23271z_{6} \\
&\quad - 1.94425z_{7} - 1.54491z_{8} - 9.51747z_{9} \\
&\quad - 2.99909z_{10} + 1.98699z_{11} - 2.09143z_{12} \\
&\quad - 1.66837z_{13} - 10.00431z_{14} - 1.23648z_{15} \\
&\quad + 5.67856z_{16} - 3.72290z_{17} + 9.04283z_{18} \\
&\quad + 4.21710z_{19} + 15.04248z_{20} - 1.12488z_{21} \\
&\quad + 0.81993z_{22} + 2.43517z_{23} + 7.55041z_{24} \\
&\quad + 1.87908z_{25}
\end{aligned}\]

\noindent\textbf{Cutoff $R_e(0)>2$.}
\[\begin{aligned}
\eta_{R_e(0)>2} &= -1.79598 \\
&\quad + 5.50834z_{1} + 1.09106z_{2} - 3.31485z_{3} \\
&\quad + 0.85678z_{4} + 0.51538z_{5} + 2.92191z_{6} \\
&\quad - 1.74342z_{7} - 1.82689z_{8} - 8.71827z_{9} \\
&\quad - 3.78374z_{10} + 4.19137z_{11} - 3.78227z_{12} \\
&\quad - 1.03897z_{13} - 10.03790z_{14} - 2.68514z_{15} \\
&\quad + 6.80886z_{16} - 3.98391z_{17} + 9.46555z_{18} \\
&\quad + 4.00849z_{19} + 13.95406z_{20} + 0.23842z_{21} \\
&\quad + 0.80549z_{22} + 1.92091z_{23} + 9.13927z_{24} \\
&\quad + 1.74018z_{25}
\end{aligned}\]

\noindent\textbf{Cutoff $R_e(0)>2.5$.}
\[\begin{aligned}
\eta_{R_e(0)>2.5} &= -7.24134 \\
&\quad + 4.44683z_{1} + 0.67843z_{2} - 4.02332z_{3} \\
&\quad + 1.51440z_{4} + 1.15642z_{5} + 3.13111z_{6} \\
&\quad - 1.75673z_{7} - 2.15666z_{8} - 9.93828z_{9} \\
&\quad - 4.20062z_{10} + 3.54137z_{11} - 2.96366z_{12} \\
&\quad - 1.57444z_{13} - 9.90381z_{14} - 2.13524z_{15} \\
&\quad + 5.88470z_{16} - 3.88674z_{17} + 9.45521z_{18} \\
&\quad + 3.58148z_{19} + 15.38229z_{20} + 0.16994z_{21} \\
&\quad + 0.65218z_{22} + 2.31014z_{23} + 7.94185z_{24} \\
&\quad + 2.02597z_{25}
\end{aligned}\]

\noindent\textbf{Cutoff $R_e(0)>3$.}
\[\begin{aligned}
\eta_{R_e(0)>3} &= -10.50063 \\
&\quad + 3.11718z_{1} + 0.85032z_{2} - 2.92088z_{3} \\
&\quad + 0.64544z_{4} + 1.94444z_{5} + 2.96463z_{6} \\
&\quad - 2.55969z_{7} - 2.95393z_{8} - 9.76937z_{9} \\
&\quad - 4.47532z_{10} + 3.38237z_{11} - 3.85867z_{12} \\
&\quad - 0.66077z_{13} - 9.13473z_{14} - 2.09864z_{15} \\
&\quad + 5.67992z_{16} - 2.80464z_{17} + 8.97262z_{18} \\
&\quad + 3.94164z_{19} + 14.17886z_{20} + 0.40452z_{21} \\
&\quad - 0.07853z_{22} + 1.31800z_{23} + 7.35065z_{24} \\
&\quad + 1.30299z_{25}
\end{aligned}\]

\noindent\textbf{Cutoff $R_e(0)>3.5$.}
\[\begin{aligned}
\eta_{R_e(0)>3.5} &= -12.86653 \\
&\quad + 1.33865z_{1} + 1.49820z_{2} - 3.25051z_{3} \\
&\quad + 0.81975z_{4} + 1.43944z_{5} + 2.62064z_{6} \\
&\quad - 1.46347z_{7} - 2.64202z_{8} - 8.78261z_{9} \\
&\quad - 4.10525z_{10} + 2.88560z_{11} - 3.49788z_{12} \\
&\quad - 1.26225z_{13} - 7.96236z_{14} - 1.65527z_{15} \\
&\quad + 5.47822z_{16} - 2.63260z_{17} + 8.00012z_{18} \\
&\quad + 3.11408z_{19} + 13.21129z_{20} + 0.98490z_{21} \\
&\quad + 0.59949z_{22} + 1.38234z_{23} + 7.29942z_{24} \\
&\quad + 1.65608z_{25}
\end{aligned}\]

\noindent\textbf{Cutoff $R_e(0)>4$.}
\[\begin{aligned}
\eta_{R_e(0)>4} &= -13.56163 \\
&\quad - 0.57925z_{1} - 0.00235z_{2} - 2.45469z_{3} \\
&\quad - 0.56393z_{4} + 0.63566z_{5} + 3.35199z_{6} \\
&\quad - 1.47644z_{7} - 2.30930z_{8} - 7.27011z_{9} \\
&\quad - 4.27895z_{10} + 3.12929z_{11} - 4.07937z_{12} \\
&\quad - 0.28275z_{13} - 7.79989z_{14} - 1.98697z_{15} \\
&\quad + 4.89278z_{16} - 2.67505z_{17} + 7.70286z_{18} \\
&\quad + 3.15197z_{19} + 10.76656z_{20} + 0.51692z_{21} \\
&\quad + 0.33156z_{22} + 1.59903z_{23} + 7.51196z_{24} \\
&\quad + 1.26239z_{25}
\end{aligned}\]

\noindent\textbf{Cutoff $R_e(0)>4.5$.}
\[\begin{aligned}
\eta_{R_e(0)>4.5} &= -14.74831 \\
&\quad - 0.28313z_{1} + 0.16381z_{2} - 2.14913z_{3} \\
&\quad + 0.37243z_{4} + 0.66451z_{5} + 3.31841z_{6} \\
&\quad - 0.31063z_{7} - 1.91621z_{8} - 6.78056z_{9} \\
&\quad - 3.64987z_{10} + 3.29400z_{11} - 3.70301z_{12} \\
&\quad - 1.21171z_{13} - 6.61800z_{14} - 1.76105z_{15} \\
&\quad + 4.04452z_{16} - 1.83185z_{17} + 7.72163z_{18} \\
&\quad + 1.99952z_{19} + 10.49991z_{20} - 0.21709z_{21} \\
&\quad + 0.86229z_{22} + 1.33218z_{23} + 7.09810z_{24} \\
&\quad + 1.25706z_{25}
\end{aligned}\]

\noindent\textbf{Cutoff $R_e(0)>5$.}
\[\begin{aligned}
\eta_{R_e(0)>5} &= -15.18745 \\
&\quad - 1.17129z_{1} + 0.49532z_{2} - 0.48506z_{3} \\
&\quad + 0.20160z_{4} + 0.40682z_{5} + 2.92941z_{6} \\
&\quad - 0.78759z_{7} - 1.25560z_{8} - 6.60782z_{9} \\
&\quad - 2.83902z_{10} + 2.68959z_{11} - 4.25283z_{12} \\
&\quad - 0.93107z_{13} - 5.14266z_{14} - 1.35728z_{15} \\
&\quad + 3.74240z_{16} - 1.56665z_{17} + 7.39151z_{18} \\
&\quad + 2.25238z_{19} + 10.04362z_{20} + 0.25919z_{21} \\
&\quad + 0.77258z_{22} + 1.32197z_{23} + 6.78869z_{24} \\
&\quad + 2.37159z_{25}
\end{aligned}\]

\noindent\textbf{Cutoff $R_e(0)>5.5$.}
\[\begin{aligned}
\eta_{R_e(0)>5.5} &= -15.27603 \\
&\quad - 1.59277z_{1} - 0.38502z_{2} - 0.73856z_{3} \\
&\quad - 0.44612z_{4} - 0.10936z_{5} + 2.77604z_{6} \\
&\quad - 0.78672z_{7} - 0.95595z_{8} - 5.37781z_{9} \\
&\quad - 2.91749z_{10} + 2.11621z_{11} - 3.73074z_{12} \\
&\quad - 0.28326z_{13} - 4.92944z_{14} - 1.48228z_{15} \\
&\quad + 3.17472z_{16} - 1.53038z_{17} + 6.62710z_{18} \\
&\quad + 2.06679z_{19} + 8.80175z_{20} + 0.53972z_{21} \\
&\quad + 0.86266z_{22} + 1.44337z_{23} + 6.64923z_{24} \\
&\quad + 2.12197z_{25}
\end{aligned}\]

\noindent\textbf{Cutoff $R_e(0)>6$.}
\[\begin{aligned}
\eta_{R_e(0)>6} &= -15.01673 \\
&\quad - 1.10561z_{1} - 1.19687z_{2} - 0.06696z_{3} \\
&\quad + 0.88935z_{4} + 0.78816z_{5} + 2.81025z_{6} \\
&\quad - 0.97725z_{7} - 1.45006z_{8} - 5.50217z_{9} \\
&\quad - 3.39831z_{10} + 1.88803z_{11} - 2.95865z_{12} \\
&\quad - 1.27236z_{13} - 4.24755z_{14} - 1.49252z_{15} \\
&\quad + 2.39908z_{16} - 1.80042z_{17} + 6.72777z_{18} \\
&\quad + 2.51451z_{19} + 8.80991z_{20} - 0.24474z_{21} \\
&\quad + 0.04698z_{22} + 0.99373z_{23} + 5.76245z_{24} \\
&\quad + 1.23385z_{25}
\end{aligned}\]

\noindent\textbf{Cutoff $R_e(0)>6.5$.}
\[\begin{aligned}
\eta_{R_e(0)>6.5} &= -15.24221 \\
&\quad - 1.61229z_{1} - 0.61890z_{2} - 0.60813z_{3} \\
&\quad - 0.64946z_{4} + 0.48310z_{5} + 2.76332z_{6} \\
&\quad + 0.25623z_{7} - 0.53706z_{8} - 4.05744z_{9} \\
&\quad - 2.91762z_{10} + 1.89080z_{11} - 3.16979z_{12} \\
&\quad - 0.75209z_{13} - 3.91761z_{14} - 1.28779z_{15} \\
&\quad + 2.15510z_{16} - 1.62007z_{17} + 5.36778z_{18} \\
&\quad + 1.92806z_{19} + 7.58084z_{20} + 0.06348z_{21} \\
&\quad + 0.64519z_{22} + 1.45215z_{23} + 6.31728z_{24} \\
&\quad + 0.89936z_{25}
\end{aligned}\]

\noindent\textbf{Cutoff $R_e(0)>7$.}
\[\begin{aligned}
\eta_{R_e(0)>7} &= -14.71459 \\
&\quad - 2.50976z_{1} + 0.10407z_{2} + 0.32542z_{3} \\
&\quad + 0.00003z_{4} + 0.61310z_{5} + 2.60255z_{6} \\
&\quad - 0.39457z_{7} - 0.31629z_{8} - 4.35341z_{9} \\
&\quad - 2.68685z_{10} + 1.63331z_{11} - 3.15364z_{12} \\
&\quad - 1.45451z_{13} - 2.82406z_{14} - 0.67381z_{15} \\
&\quad + 2.17061z_{16} - 1.46092z_{17} + 5.24612z_{18} \\
&\quad + 2.00365z_{19} + 7.31661z_{20} + 0.34774z_{21} \\
&\quad + 0.15588z_{22} + 1.44895z_{23} + 5.45430z_{24} \\
&\quad + 0.87853z_{25}
\end{aligned}\]

\noindent\textbf{Cutoff $R_e(0)>7.5$.}
\[\begin{aligned}
\eta_{R_e(0)>7.5} &= -14.38753 \\
&\quad - 1.97961z_{1} - 0.03309z_{2} + 1.48535z_{3} \\
&\quad - 0.47112z_{4} + 0.30885z_{5} + 2.39944z_{6} \\
&\quad - 0.00477z_{7} - 0.65167z_{8} - 3.83994z_{9} \\
&\quad - 1.79791z_{10} + 1.46392z_{11} - 3.77257z_{12} \\
&\quad - 0.56029z_{13} - 2.20213z_{14} - 1.14601z_{15} \\
&\quad + 2.39963z_{16} - 1.19476z_{17} + 5.39005z_{18} \\
&\quad + 2.07680z_{19} + 6.36521z_{20} + 0.16532z_{21} \\
&\quad + 0.81629z_{22} + 0.86341z_{23} + 5.39641z_{24} \\
&\quad + 1.53947z_{25}
\end{aligned}\]

\noindent\textbf{Cutoff $R_e(0)>8$.}
\[\begin{aligned}
\eta_{R_e(0)>8} &= -13.78360 \\
&\quad - 2.39805z_{1} - 0.59447z_{2} + 0.97357z_{3} \\
&\quad - 0.22786z_{4} + 0.65933z_{5} + 2.14984z_{6} \\
&\quad - 0.43624z_{7} - 0.13977z_{8} - 3.53353z_{9} \\
&\quad - 2.00786z_{10} + 1.17323z_{11} - 3.23301z_{12} \\
&\quad - 0.40877z_{13} - 2.27112z_{14} - 0.99200z_{15} \\
&\quad + 1.70363z_{16} - 1.08970z_{17} + 4.53906z_{18} \\
&\quad + 1.69835z_{19} + 6.38070z_{20} + 0.43589z_{21} \\
&\quad + 0.32859z_{22} + 0.99473z_{23} + 4.63894z_{24} \\
&\quad + 1.06182z_{25}
\end{aligned}\]

\noindent\textbf{Cutoff $R_e(0)>8.5$.}
\[\begin{aligned}
\eta_{R_e(0)>8.5} &= -13.27609 \\
&\quad - 2.14104z_{1} - 0.74135z_{2} + 1.39751z_{3} \\
&\quad - 0.36470z_{4} + 0.20461z_{5} + 1.86120z_{6} \\
&\quad - 0.41875z_{7} - 0.39339z_{8} - 3.10255z_{9} \\
&\quad - 1.43055z_{10} + 0.91973z_{11} - 3.48424z_{12} \\
&\quad + 0.55084z_{13} - 2.18701z_{14} - 1.22640z_{15} \\
&\quad + 1.35165z_{16} - 0.89208z_{17} + 4.74472z_{18} \\
&\quad + 1.55807z_{19} + 5.95772z_{20} + 0.35570z_{21} \\
&\quad + 0.26157z_{22} + 0.81148z_{23} + 4.73081z_{24} \\
&\quad + 1.79706z_{25}
\end{aligned}\]

\noindent\textbf{Cutoff $R_e(0)>9$.}
\[\begin{aligned}
\eta_{R_e(0)>9} &= -12.81389 \\
&\quad - 1.81266z_{1} - 0.96050z_{2} + 1.84363z_{3} \\
&\quad - 0.42037z_{4} + 0.17952z_{5} + 1.54035z_{6} \\
&\quad - 0.33783z_{7} - 0.21738z_{8} - 2.42739z_{9} \\
&\quad - 1.40401z_{10} + 1.21740z_{11} - 2.95234z_{12} \\
&\quad - 0.50468z_{13} - 1.59391z_{14} - 1.29605z_{15} \\
&\quad + 1.70227z_{16} - 0.47418z_{17} + 3.75919z_{18} \\
&\quad + 1.77272z_{19} + 5.44834z_{20} + 0.14259z_{21} \\
&\quad + 0.96422z_{22} + 0.16852z_{23} + 4.71299z_{24} \\
&\quad + 1.85999z_{25}
\end{aligned}\]

\noindent\textbf{Cutoff $R_e(0)>9.5$.}
\[\begin{aligned}
\eta_{R_e(0)>9.5} &= -12.52887 \\
&\quad - 1.86347z_{1} - 1.09717z_{2} + 2.21625z_{3} \\
&\quad - 0.42185z_{4} - 0.06114z_{5} + 1.58304z_{6} \\
&\quad + 0.01193z_{7} - 0.47595z_{8} - 2.18939z_{9} \\
&\quad - 1.17764z_{10} + 0.70105z_{11} - 3.06771z_{12} \\
&\quad - 0.25615z_{13} - 1.45287z_{14} - 1.17270z_{15} \\
&\quad + 1.61601z_{16} - 1.07085z_{17} + 3.62130z_{18} \\
&\quad + 1.70202z_{19} + 4.91091z_{20} - 0.08654z_{21} \\
&\quad + 1.19240z_{22} - 0.19859z_{23} + 4.31668z_{24} \\
&\quad + 1.73179z_{25}
\end{aligned}\]

\noindent\textbf{Cutoff $R_e(0)>10$.}
\[\begin{aligned}
\eta_{R_e(0)>10} &= -12.21405 \\
&\quad - 1.90160z_{1} - 0.45407z_{2} + 2.65710z_{3} \\
&\quad - 0.41694z_{4} - 0.08765z_{5} + 1.59209z_{6} \\
&\quad + 0.30459z_{7} - 0.52230z_{8} - 2.04028z_{9} \\
&\quad - 1.48870z_{10} + 0.64787z_{11} - 2.46858z_{12} \\
&\quad - 0.59162z_{13} - 1.58207z_{14} - 1.39528z_{15} \\
&\quad + 1.49692z_{16} - 0.45323z_{17} + 3.78997z_{18} \\
&\quad + 2.24520z_{19} + 4.25487z_{20} - 0.02739z_{21} \\
&\quad + 0.91962z_{22} - 0.22102z_{23} + 4.21636z_{24} \\
&\quad + 1.84043z_{25}
\end{aligned}\]

\subsubsection{Age-structured SIR: validation and equations}

\begingroup\scriptsize\setlength{\LTcapwidth}{\linewidth}\setlength{\tabcolsep}{3pt}

\begin{longtable}{@{}lrrrrrrr@{}}

\caption{Class-specific held-out metrics for Age-structured SIR with log cubic $h=0.5$.}\label{tab:supp-class-log_cubic_h0p5-age-structured-sir}\\

\toprule Bracket & Support & Recall & Spec. & Precision & F1 & NPV & OvR AUC \\ \midrule\endfirsthead

\toprule Bracket & Support & Recall & Spec. & Precision & F1 & NPV & OvR AUC \\ \midrule\endhead

 \texttt{$<$1} & 1729 & 0.973 & 0.993 & 0.973 & 0.973 & 0.993 & 0.999 \\

 \texttt{1-1.5} & 977 & 0.868 & 0.985 & 0.887 & 0.877 & 0.982 & 0.992 \\

 \texttt{1.5-2} & 938 & 0.868 & 0.979 & 0.844 & 0.855 & 0.982 & 0.988 \\

 \texttt{2-2.5} & 816 & 0.838 & 0.981 & 0.834 & 0.836 & 0.982 & 0.986 \\

 \texttt{2.5-3} & 655 & 0.783 & 0.980 & 0.775 & 0.779 & 0.981 & 0.984 \\

 \texttt{3-3.5} & 536 & 0.733 & 0.979 & 0.715 & 0.724 & 0.981 & 0.980 \\

 \texttt{3.5-4} & 437 & 0.696 & 0.982 & 0.686 & 0.691 & 0.982 & 0.979 \\

 \texttt{4-4.5} & 338 & 0.624 & 0.985 & 0.645 & 0.635 & 0.983 & 0.976 \\

 \texttt{4.5-5} & 275 & 0.542 & 0.983 & 0.527 & 0.534 & 0.984 & 0.973 \\

 \texttt{5-5.5} & 225 & 0.507 & 0.984 & 0.473 & 0.489 & 0.986 & 0.963 \\

 \texttt{5.5-6} & 169 & 0.349 & 0.989 & 0.396 & 0.371 & 0.986 & 0.973 \\

 \texttt{6-6.5} & 147 & 0.469 & 0.987 & 0.411 & 0.438 & 0.990 & 0.957 \\

 \texttt{6.5-7} & 102 & 0.255 & 0.992 & 0.286 & 0.269 & 0.990 & 0.954 \\

 \texttt{7-7.5} & 89 & 0.236 & 0.994 & 0.296 & 0.263 & 0.991 & 0.945 \\

 \texttt{7.5-8} & 88 & 0.398 & 0.993 & 0.389 & 0.393 & 0.993 & 0.953 \\

 \texttt{8-8.5} & 83 & 0.277 & 0.993 & 0.295 & 0.286 & 0.992 & 0.951 \\

 \texttt{8.5-9} & 52 & 0.173 & 0.998 & 0.346 & 0.231 & 0.995 & 0.936 \\

 \texttt{9-9.5} & 50 & 0.180 & 0.997 & 0.265 & 0.214 & 0.995 & 0.959 \\

 \texttt{9.5-10} & 43 & 0.023 & 0.997 & 0.042 & 0.030 & 0.995 & 0.928 \\

 \texttt{$>$10} & 251 & 0.944 & 0.993 & 0.809 & 0.871 & 0.998 & 0.999 \\

\bottomrule\end{longtable}\endgroup

\paragraph{Ordered-threshold feature map and equations.}

All $42$ fitted predictors are standardized using $z_j=(\phi_j(\boldsymbol\theta)-\bar\phi_j)/s_j$, where the values below use the complete training sample for the printed (full-data) equations. Held-out predictions separately recompute means, standard deviations, and fits within each training fold. Here $\phi_j$ is the feature shown in the table. The score for every estimable cutoff is $q_c=\operatorname{logit}^{-1}(\eta_c)$; positive scores increase the fitted exceedance probability, not the numerical reproduction number.

\begingroup\scriptsize\setlength{\LTcapwidth}{\linewidth}

\begin{longtable}{@{}r p{0.44\linewidth} r r@{}}

\caption{Standardization constants and ordered features for Age-structured SIR (compact log-cubic basis).}\label{tab:supp-features-age-structured-sir-fine}\\

\toprule Index & Raw feature $\phi_j$ & Training mean & Training SD \\ \midrule\endfirsthead

\toprule Index & Raw feature $\phi_j$ & Training mean & Training SD \\ \midrule\endhead

 1 & $\log(\beta)$ & -1.095221 & 0.750786 \\

 2 & $\log(\gamma_1)$ & -1.437305 & 0.587783 \\

 3 & $\log(\gamma_2)$ & -1.437305 & 0.587782 \\

 4 & $\log(c_{11})$ & -0.401106 & 0.662438 \\

 5 & $\log(c_{12})$ & -0.842331 & 0.709130 \\

 6 & $\log(c_{21})$ & -0.842331 & 0.709132 \\

 7 & $\log(c_{22})$ & -0.401102 & 0.662432 \\

 8 & $[\log(\beta)]^2$ & 1.763119 & 2.342718 \\

 9 & $[\log(\gamma_1)]^2$ & 2.411291 & 1.990102 \\

 10 & $[\log(\gamma_2)]^2$ & 2.411292 & 1.990096 \\

 11 & $[\log(c_{11})]^2$ & 0.599655 & 1.041028 \\

 12 & $[\log(c_{12})]^2$ & 1.212324 & 1.783126 \\

 13 & $[\log(c_{21})]^2$ & 1.212327 & 1.783148 \\

 14 & $[\log(c_{22})]^2$ & 0.599644 & 1.040997 \\

 15 & $\log(\beta)\,\log(\gamma_1)$ & 1.574936 & 1.343856 \\

 16 & $\log(\beta)\,\log(\gamma_2)$ & 1.570945 & 1.320221 \\

 17 & $\log(\beta)\,\log(c_{11})$ & 0.443235 & 0.928937 \\

 18 & $\log(\beta)\,\log(c_{12})$ & 0.931904 & 1.148798 \\

 19 & $\log(\beta)\,\log(c_{21})$ & 0.910706 & 1.125175 \\

 20 & $\log(\beta)\,\log(c_{22})$ & 0.442628 & 0.935096 \\

 21 & $\log(\gamma_1)\,\log(\gamma_2)$ & 2.070274 & 1.255602 \\

 22 & $\log(\gamma_1)\,\log(c_{11})$ & 0.582804 & 1.067473 \\

 23 & $\log(\gamma_1)\,\log(c_{12})$ & 1.211305 & 1.205154 \\

 24 & $\log(\gamma_1)\,\log(c_{21})$ & 1.213620 & 1.215890 \\

 25 & $\log(\gamma_1)\,\log(c_{22})$ & 0.572965 & 1.055041 \\

 26 & $\log(\gamma_2)\,\log(c_{11})$ & 0.572415 & 1.049115 \\

 27 & $\log(\gamma_2)\,\log(c_{12})$ & 1.212491 & 1.214611 \\

 28 & $\log(\gamma_2)\,\log(c_{21})$ & 1.215244 & 1.229895 \\

 29 & $\log(\gamma_2)\,\log(c_{22})$ & 0.583598 & 1.054823 \\

 30 & $\log(c_{11})\,\log(c_{12})$ & 0.334612 & 0.780339 \\

 31 & $\log(c_{11})\,\log(c_{21})$ & 0.328326 & 0.756712 \\

 32 & $\log(c_{11})\,\log(c_{22})$ & 0.157559 & 0.573840 \\

 33 & $\log(c_{12})\,\log(c_{21})$ & 0.708052 & 0.976793 \\

 34 & $\log(c_{12})\,\log(c_{22})$ & 0.327051 & 0.744933 \\

 35 & $\log(c_{21})\,\log(c_{22})$ & 0.347835 & 0.795570 \\

 36 & $[\log(\beta)]^3$ & -3.621050 & 6.958366 \\

 37 & $[\log(\gamma_1)]^3$ & -4.616689 & 5.652195 \\

 38 & $[\log(\gamma_2)]^3$ & -4.616686 & 5.652165 \\

 39 & $[\log(c_{11})]^3$ & -0.855536 & 2.027677 \\

 40 & $[\log(c_{12})]^3$ & -2.221963 & 4.505271 \\

 41 & $[\log(c_{21})]^3$ & -2.221983 & 4.505398 \\

 42 & $[\log(c_{22})]^3$ & -0.855504 & 2.027583 \\

\bottomrule\end{longtable}\endgroup

Each displayed score below is a numeric equation in the standardized variables $z_1,\ldots,z_42$. The intercept and slopes are posterior modes from the independent $N(0,2^2)$ coefficient prior under the Gaussian Laplace approximation. A cutoff with only one observed class cannot produce a fitted finite logistic score, so its constant exceedance probability is stated explicitly.

\noindent\textbf{Cutoff $R_0>1$.}
\[\begin{aligned}
\eta_{R_0>1} &= 10.44514 \\
&\quad + 8.44757z_{1} - 4.01733z_{2} - 4.45424z_{3} \\
&\quad + 0.29657z_{4} + 2.56731z_{5} + 2.76346z_{6} \\
&\quad - 0.21046z_{7} - 5.04798z_{8} + 3.01612z_{9} \\
&\quad + 3.05344z_{10} + 2.32689z_{11} + 0.51504z_{12} \\
&\quad + 1.47482z_{13} + 1.74452z_{14} + 1.56951z_{15} \\
&\quad + 1.74126z_{16} - 3.82530z_{17} + 0.80018z_{18} \\
&\quad + 0.94933z_{19} - 5.11917z_{20} - 5.00308z_{21} \\
&\quad - 4.34912z_{22} - 0.60779z_{23} - 0.41675z_{24} \\
&\quad + 7.18824z_{25} + 7.12928z_{26} - 0.74837z_{27} \\
&\quad - 1.18436z_{28} - 3.60027z_{29} - 1.07797z_{30} \\
&\quad - 0.97877z_{31} - 3.46162z_{32} + 1.09192z_{33} \\
&\quad - 0.72685z_{34} - 0.82324z_{35} + 1.63642z_{36} \\
&\quad - 1.84515z_{37} - 1.38299z_{38} + 0.25004z_{39} \\
&\quad - 0.64191z_{40} + 0.24141z_{41} - 0.11020z_{42}
\end{aligned}\]

\noindent\textbf{Cutoff $R_0>1.5$.}
\[\begin{aligned}
\eta_{R_0>1.5} &= 4.94485 \\
&\quad + 9.83767z_{1} - 4.83535z_{2} - 4.77582z_{3} \\
&\quad + 1.22797z_{4} + 2.08950z_{5} + 2.37947z_{6} \\
&\quad + 1.74866z_{7} - 5.35363z_{8} + 4.89956z_{9} \\
&\quad + 3.53328z_{10} + 1.43072z_{11} + 0.22050z_{12} \\
&\quad + 0.04446z_{13} + 2.93402z_{14} - 0.78814z_{15} \\
&\quad + 0.15422z_{16} - 5.37290z_{17} + 1.10321z_{18} \\
&\quad + 1.21692z_{19} - 4.93209z_{20} - 4.45144z_{21} \\
&\quad - 4.20686z_{22} - 0.86273z_{23} - 0.67365z_{24} \\
&\quad + 7.21502z_{25} + 7.89113z_{26} - 1.26330z_{27} \\
&\quad - 0.94795z_{28} - 3.75942z_{29} - 1.20095z_{30} \\
&\quad - 1.14540z_{31} - 3.81128z_{32} + 0.97971z_{33} \\
&\quad - 0.87282z_{34} - 0.80764z_{35} + 1.06708z_{36} \\
&\quad - 1.87243z_{37} - 2.76389z_{38} - 0.63671z_{39} \\
&\quad - 0.75029z_{40} - 0.68493z_{41} + 0.95234z_{42}
\end{aligned}\]

\noindent\textbf{Cutoff $R_0>2$.}
\[\begin{aligned}
\eta_{R_0>2} &= -0.11788 \\
&\quad + 12.68453z_{1} - 6.01658z_{2} - 6.34297z_{3} \\
&\quad + 3.40039z_{4} + 3.24265z_{5} + 2.71372z_{6} \\
&\quad + 2.54397z_{7} - 3.22000z_{8} + 4.05219z_{9} \\
&\quad + 4.19490z_{10} + 2.35003z_{11} + 2.43701z_{12} \\
&\quad + 0.57219z_{13} + 2.65930z_{14} - 2.72183z_{15} \\
&\quad - 1.24018z_{16} - 4.09549z_{17} + 2.70868z_{18} \\
&\quad + 2.08461z_{19} - 3.72989z_{20} - 3.94513z_{21} \\
&\quad - 4.63997z_{22} - 1.65857z_{23} - 1.22782z_{24} \\
&\quad + 6.88608z_{25} + 7.68685z_{26} - 1.64726z_{27} \\
&\quad - 1.34749z_{28} - 5.06132z_{29} - 0.65355z_{30} \\
&\quad - 0.63919z_{31} - 3.52039z_{32} + 1.57760z_{33} \\
&\quad - 0.71000z_{34} - 0.86732z_{35} - 0.54457z_{36} \\
&\quad - 2.95392z_{37} - 1.51741z_{38} + 0.12523z_{39} \\
&\quad + 0.61411z_{40} - 0.50665z_{41} + 0.22418z_{42}
\end{aligned}\]

\noindent\textbf{Cutoff $R_0>2.5$.}
\[\begin{aligned}
\eta_{R_0>2.5} &= -4.65505 \\
&\quad + 12.50568z_{1} - 6.33800z_{2} - 7.19840z_{3} \\
&\quad + 3.98976z_{4} + 2.54935z_{5} + 2.84696z_{6} \\
&\quad + 4.10756z_{7} - 1.72990z_{8} + 4.87121z_{9} \\
&\quad + 4.22336z_{10} + 2.49674z_{11} + 1.78934z_{12} \\
&\quad + 1.51010z_{13} + 3.00860z_{14} - 3.10748z_{15} \\
&\quad - 4.00433z_{16} - 2.65440z_{17} + 2.91933z_{18} \\
&\quad + 2.62618z_{19} - 4.21882z_{20} - 4.44820z_{21} \\
&\quad - 5.13862z_{22} - 2.10318z_{23} - 1.30781z_{24} \\
&\quad + 7.95861z_{25} + 6.74295z_{26} - 2.08014z_{27} \\
&\quad - 1.68049z_{28} - 5.06588z_{29} - 0.40082z_{30} \\
&\quad - 0.76123z_{31} - 3.46274z_{32} + 1.74145z_{33} \\
&\quad - 0.46472z_{34} - 0.80110z_{35} - 1.28884z_{36} \\
&\quad - 2.23354z_{37} - 2.92758z_{38} + 0.04985z_{39} \\
&\quad + 0.15074z_{40} + 0.17368z_{41} + 0.86562z_{42}
\end{aligned}\]

\noindent\textbf{Cutoff $R_0>3$.}
\[\begin{aligned}
\eta_{R_0>3} &= -9.17416 \\
&\quad + 11.54394z_{1} - 5.88213z_{2} - 5.58771z_{3} \\
&\quad + 4.32133z_{4} + 2.54019z_{5} + 2.21579z_{6} \\
&\quad + 3.57435z_{7} - 1.13181z_{8} + 4.76040z_{9} \\
&\quad + 4.76072z_{10} + 3.34099z_{11} + 1.49651z_{12} \\
&\quad + 1.05244z_{13} + 1.98646z_{14} - 4.12695z_{15} \\
&\quad - 4.09199z_{16} - 4.30581z_{17} + 2.60359z_{18} \\
&\quad + 3.20700z_{19} - 5.03172z_{20} - 3.86802z_{21} \\
&\quad - 4.82794z_{22} - 1.16632z_{23} - 1.91202z_{24} \\
&\quad + 7.66036z_{25} + 7.39022z_{26} - 1.35237z_{27} \\
&\quad - 1.99477z_{28} - 4.63211z_{29} - 0.63206z_{30} \\
&\quad - 0.37355z_{31} - 4.28597z_{32} + 1.48240z_{33} \\
&\quad - 0.59171z_{34} - 0.67127z_{35} - 1.54112z_{36} \\
&\quad - 2.27508z_{37} - 2.77849z_{38} + 1.26328z_{39} \\
&\quad + 0.31569z_{40} - 0.59410z_{41} + 0.34347z_{42}
\end{aligned}\]

\noindent\textbf{Cutoff $R_0>3.5$.}
\[\begin{aligned}
\eta_{R_0>3.5} &= -11.84464 \\
&\quad + 10.15756z_{1} - 5.18383z_{2} - 4.88051z_{3} \\
&\quad + 3.51356z_{4} + 1.79817z_{5} + 0.98914z_{6} \\
&\quad + 3.92209z_{7} - 1.60927z_{8} + 4.47177z_{9} \\
&\quad + 5.41439z_{10} + 1.19383z_{11} + 1.67493z_{12} \\
&\quad + 0.73988z_{13} + 2.83366z_{14} - 4.10805z_{15} \\
&\quad - 4.79697z_{16} - 4.83846z_{17} + 3.40320z_{18} \\
&\quad + 3.22937z_{19} - 3.58693z_{20} - 2.55306z_{21} \\
&\quad - 4.73067z_{22} - 1.89549z_{23} - 2.35233z_{24} \\
&\quad + 5.99724z_{25} + 6.69220z_{26} - 2.25981z_{27} \\
&\quad - 2.45919z_{28} - 4.91594z_{29} - 0.44214z_{30} \\
&\quad - 0.37527z_{31} - 3.82264z_{32} + 1.52718z_{33} \\
&\quad - 0.26357z_{34} - 0.20753z_{35} - 1.37022z_{36} \\
&\quad - 2.41158z_{37} - 2.28521z_{38} - 0.39418z_{39} \\
&\quad + 0.01522z_{40} - 0.43988z_{41} + 0.55930z_{42}
\end{aligned}\]

\noindent\textbf{Cutoff $R_0>4$.}
\[\begin{aligned}
\eta_{R_0>4} &= -13.69569 \\
&\quad + 8.74854z_{1} - 5.05688z_{2} - 5.15910z_{3} \\
&\quad + 4.59894z_{4} + 0.96515z_{5} + 0.51503z_{6} \\
&\quad + 3.71080z_{7} - 1.27580z_{8} + 4.41408z_{9} \\
&\quad + 4.04955z_{10} + 1.49526z_{11} + 0.65758z_{12} \\
&\quad + 0.60289z_{13} + 2.08487z_{14} - 4.12358z_{15} \\
&\quad - 5.74645z_{16} - 3.36879z_{17} + 3.10485z_{18} \\
&\quad + 3.10037z_{19} - 3.00585z_{20} - 2.93092z_{21} \\
&\quad - 4.22175z_{22} - 1.77074z_{23} - 2.12382z_{24} \\
&\quad + 6.04259z_{25} + 6.38551z_{26} - 2.02737z_{27} \\
&\quad - 2.15282z_{28} - 4.99214z_{29} - 0.05517z_{30} \\
&\quad - 0.25618z_{31} - 3.47603z_{32} + 1.37005z_{33} \\
&\quad - 0.38415z_{34} - 0.55396z_{35} - 1.44491z_{36} \\
&\quad - 2.13745z_{37} - 3.07159z_{38} + 0.20432z_{39} \\
&\quad - 0.26994z_{40} - 0.43227z_{41} + 0.24223z_{42}
\end{aligned}\]

\noindent\textbf{Cutoff $R_0>4.5$.}
\[\begin{aligned}
\eta_{R_0>4.5} &= -14.76835 \\
&\quad + 7.37321z_{1} - 5.64858z_{2} - 5.14804z_{3} \\
&\quad + 3.89965z_{4} + 1.10930z_{5} + 0.05782z_{6} \\
&\quad + 3.23885z_{7} - 1.45850z_{8} + 3.84472z_{9} \\
&\quad + 5.23671z_{10} + 1.37970z_{11} + 1.35330z_{12} \\
&\quad + 1.22050z_{13} + 0.87972z_{14} - 3.94614z_{15} \\
&\quad - 5.42572z_{16} - 3.10436z_{17} + 2.41821z_{18} \\
&\quad + 2.73579z_{19} - 2.12226z_{20} - 3.73875z_{21} \\
&\quad - 3.78235z_{22} - 1.49372z_{23} - 2.07239z_{24} \\
&\quad + 5.52322z_{25} + 5.30175z_{26} - 1.48409z_{27} \\
&\quad - 2.14505z_{28} - 4.96853z_{29} - 0.58658z_{30} \\
&\quad - 0.49241z_{31} - 2.87755z_{32} + 1.05783z_{33} \\
&\quad - 0.74190z_{34} - 0.58233z_{35} - 1.66704z_{36} \\
&\quad - 1.81610z_{37} - 1.60644z_{38} - 0.07809z_{39} \\
&\quad + 0.07892z_{40} + 0.03115z_{41} - 1.09646z_{42}
\end{aligned}\]

\noindent\textbf{Cutoff $R_0>5$.}
\[\begin{aligned}
\eta_{R_0>5} &= -15.31216 \\
&\quad + 6.01744z_{1} - 3.99120z_{2} - 3.69159z_{3} \\
&\quad + 3.88062z_{4} + 0.92936z_{5} + 0.68754z_{6} \\
&\quad + 3.97069z_{7} - 0.08683z_{8} + 4.08285z_{9} \\
&\quad + 4.28510z_{10} + 0.77566z_{11} + 1.10190z_{12} \\
&\quad + 1.21556z_{13} + 0.73647z_{14} - 4.95861z_{15} \\
&\quad - 4.78835z_{16} - 3.43650z_{17} + 1.04639z_{18} \\
&\quad + 2.74919z_{19} - 2.52250z_{20} - 2.27185z_{21} \\
&\quad - 3.59057z_{22} - 1.49160z_{23} - 1.53972z_{24} \\
&\quad + 5.24571z_{25} + 5.62405z_{26} - 0.60748z_{27} \\
&\quad - 1.76842z_{28} - 3.37523z_{29} - 0.73438z_{30} \\
&\quad - 0.48466z_{31} - 3.12188z_{32} + 1.19170z_{33} \\
&\quad - 0.98040z_{34} - 0.51611z_{35} - 2.64274z_{36} \\
&\quad - 1.56299z_{37} - 1.42544z_{38} - 0.46700z_{39} \\
&\quad + 0.17527z_{40} + 0.05378z_{41} - 0.53981z_{42}
\end{aligned}\]

\noindent\textbf{Cutoff $R_0>5.5$.}
\[\begin{aligned}
\eta_{R_0>5.5} &= -15.66640 \\
&\quad + 4.67447z_{1} - 2.82089z_{2} - 3.56673z_{3} \\
&\quad + 4.34271z_{4} + 0.73026z_{5} + 0.72815z_{6} \\
&\quad + 4.60067z_{7} + 0.48233z_{8} + 3.62814z_{9} \\
&\quad + 2.65492z_{10} + 1.19979z_{11} + 0.77898z_{12} \\
&\quad + 0.89309z_{13} + 1.11309z_{14} - 5.30166z_{15} \\
&\quad - 5.66412z_{16} - 2.81486z_{17} + 1.56417z_{18} \\
&\quad + 2.48214z_{19} - 2.55163z_{20} - 1.92619z_{21} \\
&\quad - 3.59664z_{22} - 1.28068z_{23} - 1.23465z_{24} \\
&\quad + 5.58981z_{25} + 5.28646z_{26} - 0.85841z_{27} \\
&\quad - 1.33623z_{28} - 3.25641z_{29} - 0.19483z_{30} \\
&\quad - 0.54522z_{31} - 2.93283z_{32} + 1.05819z_{33} \\
&\quad - 1.00957z_{34} - 0.50513z_{35} - 3.15991z_{36} \\
&\quad - 2.43095z_{37} - 2.76026z_{38} - 0.07096z_{39} \\
&\quad - 0.00850z_{40} - 0.00644z_{41} + 0.14511z_{42}
\end{aligned}\]

\noindent\textbf{Cutoff $R_0>6$.}
\[\begin{aligned}
\eta_{R_0>6} &= -15.90365 \\
&\quad + 4.48624z_{1} - 3.47777z_{2} - 3.27489z_{3} \\
&\quad + 4.25175z_{4} - 0.26133z_{5} - 0.70625z_{6} \\
&\quad + 3.71169z_{7} + 0.36323z_{8} + 2.95068z_{9} \\
&\quad + 3.58753z_{10} + 2.21317z_{11} + 0.32932z_{12} \\
&\quad + 1.25440z_{13} + 1.70499z_{14} - 5.14840z_{15} \\
&\quad - 4.96453z_{16} - 2.26727z_{17} + 1.59463z_{18} \\
&\quad + 2.03108z_{19} - 2.61999z_{20} - 1.91161z_{21} \\
&\quad - 3.50706z_{22} - 1.37828z_{23} - 1.74459z_{24} \\
&\quad + 4.52748z_{25} + 4.83737z_{26} - 1.28381z_{27} \\
&\quad - 1.94537z_{28} - 3.45027z_{29} - 0.14400z_{30} \\
&\quad - 0.43284z_{31} - 2.59127z_{32} + 0.98222z_{33} \\
&\quad - 0.51727z_{34} - 0.40417z_{35} - 3.09605z_{36} \\
&\quad - 2.26662z_{37} - 1.87952z_{38} + 1.15009z_{39} \\
&\quad - 0.05171z_{40} + 0.41639z_{41} + 0.52837z_{42}
\end{aligned}\]

\noindent\textbf{Cutoff $R_0>6.5$.}
\[\begin{aligned}
\eta_{R_0>6.5} &= -15.80113 \\
&\quad + 3.90281z_{1} - 2.71853z_{2} - 2.70751z_{3} \\
&\quad + 3.93697z_{4} - 0.65844z_{5} - 0.21176z_{6} \\
&\quad + 4.03812z_{7} + 1.50870z_{8} + 3.08721z_{9} \\
&\quad + 3.11559z_{10} + 1.78970z_{11} + 0.42834z_{12} \\
&\quad + 1.12444z_{13} + 1.65660z_{14} - 4.61988z_{15} \\
&\quad - 5.64562z_{16} - 1.14652z_{17} + 1.45696z_{18} \\
&\quad + 1.24606z_{19} - 3.20616z_{20} - 1.52291z_{21} \\
&\quad - 3.45212z_{22} - 1.22319z_{23} - 1.37283z_{24} \\
&\quad + 4.85806z_{25} + 3.88374z_{26} - 1.30102z_{27} \\
&\quad - 1.22917z_{28} - 3.27182z_{29} - 0.28220z_{30} \\
&\quad - 0.26695z_{31} - 2.53977z_{32} + 0.84615z_{33} \\
&\quad - 0.35755z_{34} - 0.58751z_{35} - 3.67668z_{36} \\
&\quad - 1.94502z_{37} - 2.23564z_{38} + 0.52400z_{39} \\
&\quad + 0.13036z_{40} + 0.27080z_{41} + 0.50421z_{42}
\end{aligned}\]

\noindent\textbf{Cutoff $R_0>7$.}
\[\begin{aligned}
\eta_{R_0>7} &= -15.84230 \\
&\quad + 3.48889z_{1} - 1.78608z_{2} - 1.60256z_{3} \\
&\quad + 4.27613z_{4} + 0.09948z_{5} - 0.94263z_{6} \\
&\quad + 5.03073z_{7} + 0.92395z_{8} + 2.86250z_{9} \\
&\quad + 2.81682z_{10} + 2.38172z_{11} + 0.56816z_{12} \\
&\quad + 0.55796z_{13} + 0.75576z_{14} - 4.46518z_{15} \\
&\quad - 4.97193z_{16} - 1.72152z_{17} + 1.68936z_{18} \\
&\quad + 1.37679z_{19} - 2.92609z_{20} - 1.25841z_{21} \\
&\quad - 2.92508z_{22} - 0.96672z_{23} - 1.43261z_{24} \\
&\quad + 4.97393z_{25} + 3.88909z_{26} - 1.19147z_{27} \\
&\quad - 1.40709z_{28} - 2.64698z_{29} + 0.08538z_{30} \\
&\quad - 0.24608z_{31} - 2.54586z_{32} + 0.69338z_{33} \\
&\quad - 0.22593z_{34} - 0.14884z_{35} - 3.09316z_{36} \\
&\quad - 2.16317z_{37} - 2.65921z_{38} + 1.32601z_{39} \\
&\quad - 0.10996z_{40} + 0.08509z_{41} - 0.03676z_{42}
\end{aligned}\]

\noindent\textbf{Cutoff $R_0>7.5$.}
\[\begin{aligned}
\eta_{R_0>7.5} &= -15.53341 \\
&\quad + 2.57859z_{1} - 2.08729z_{2} - 2.11452z_{3} \\
&\quad + 2.93477z_{4} + 0.37793z_{5} - 0.14658z_{6} \\
&\quad + 4.06607z_{7} + 0.87963z_{8} + 2.52362z_{9} \\
&\quad + 2.46714z_{10} + 1.49763z_{11} + 1.23610z_{12} \\
&\quad + 0.73874z_{13} + 0.99527z_{14} - 5.32963z_{15} \\
&\quad - 4.49473z_{16} - 1.85290z_{17} + 1.32765z_{18} \\
&\quad + 1.75865z_{19} - 1.83505z_{20} - 1.36975z_{21} \\
&\quad - 3.27294z_{22} - 0.94207z_{23} - 1.04902z_{24} \\
&\quad + 3.89518z_{25} + 3.45561z_{26} - 0.97769z_{27} \\
&\quad - 1.36644z_{28} - 2.96509z_{29} - 0.22279z_{30} \\
&\quad - 0.27463z_{31} - 2.09590z_{32} + 0.75338z_{33} \\
&\quad - 0.14055z_{34} - 0.17584z_{35} - 3.00695z_{36} \\
&\quad - 2.27273z_{37} - 2.33105z_{38} + 0.59665z_{39} \\
&\quad + 0.42797z_{40} + 0.05679z_{41} - 0.25523z_{42}
\end{aligned}\]

\noindent\textbf{Cutoff $R_0>8$.}
\[\begin{aligned}
\eta_{R_0>8} &= -15.63194 \\
&\quad + 1.70062z_{1} - 2.39743z_{2} - 2.23231z_{3} \\
&\quad + 2.40464z_{4} + 1.41253z_{5} + 0.68103z_{6} \\
&\quad + 3.08414z_{7} + 1.06419z_{8} + 2.14635z_{9} \\
&\quad + 1.54496z_{10} + 1.12289z_{11} + 1.15306z_{12} \\
&\quad + 0.73655z_{13} + 1.14963z_{14} - 4.86549z_{15} \\
&\quad - 5.50696z_{16} - 1.99523z_{17} + 1.28563z_{18} \\
&\quad + 1.54471z_{19} - 2.10816z_{20} - 1.56371z_{21} \\
&\quad - 3.20393z_{22} - 0.34292z_{23} - 0.88971z_{24} \\
&\quad + 3.59818z_{25} + 3.33518z_{26} - 0.48418z_{27} \\
&\quad - 0.79018z_{28} - 3.16707z_{29} - 0.44698z_{30} \\
&\quad - 0.28732z_{31} - 1.95887z_{32} + 0.88556z_{33} \\
&\quad - 0.23112z_{34} - 0.09142z_{35} - 3.07411z_{36} \\
&\quad - 1.96113z_{37} - 2.76055z_{38} + 0.19348z_{39} \\
&\quad + 0.32222z_{40} - 0.01585z_{41} + 0.04183z_{42}
\end{aligned}\]

\noindent\textbf{Cutoff $R_0>8.5$.}
\[\begin{aligned}
\eta_{R_0>8.5} &= -15.25771 \\
&\quad + 3.04197z_{1} - 2.22216z_{2} - 1.17590z_{3} \\
&\quad + 2.77612z_{4} + 0.60271z_{5} + 0.31047z_{6} \\
&\quad + 2.98268z_{7} + 1.76143z_{8} + 1.93513z_{9} \\
&\quad + 1.75004z_{10} + 0.46809z_{11} + 1.35062z_{12} \\
&\quad + 0.14664z_{13} + 1.56421z_{14} - 4.93588z_{15} \\
&\quad - 4.24065z_{16} - 1.61671z_{17} + 0.92333z_{18} \\
&\quad + 1.64643z_{19} - 1.02959z_{20} - 1.22177z_{21} \\
&\quad - 3.03701z_{22} - 0.76366z_{23} - 0.91768z_{24} \\
&\quad + 3.06782z_{25} + 3.16495z_{26} - 0.54748z_{27} \\
&\quad - 0.75723z_{28} - 3.06052z_{29} - 0.29673z_{30} \\
&\quad - 0.04115z_{31} - 1.53148z_{32} + 0.80870z_{33} \\
&\quad - 0.12691z_{34} - 0.18349z_{35} - 3.32347z_{36} \\
&\quad - 1.96716z_{37} - 2.73967z_{38} - 0.30903z_{39} \\
&\quad + 0.57024z_{40} - 0.36192z_{41} + 0.09760z_{42}
\end{aligned}\]

\noindent\textbf{Cutoff $R_0>9$.}
\[\begin{aligned}
\eta_{R_0>9} &= -15.11630 \\
&\quad + 1.91934z_{1} - 1.13605z_{2} - 0.58472z_{3} \\
&\quad + 2.96726z_{4} - 0.46532z_{5} - 0.87247z_{6} \\
&\quad + 3.39520z_{7} + 2.13689z_{8} + 2.04741z_{9} \\
&\quad + 1.31991z_{10} + 1.24961z_{11} + 0.56588z_{12} \\
&\quad - 0.16925z_{13} + 1.40096z_{14} - 5.01784z_{15} \\
&\quad - 4.53633z_{16} - 1.19602z_{17} + 0.44049z_{18} \\
&\quad + 1.34947z_{19} - 0.51521z_{20} - 0.48836z_{21} \\
&\quad - 2.94063z_{22} - 0.83480z_{23} - 1.05247z_{24} \\
&\quad + 3.55298z_{25} + 2.97990z_{26} - 0.59065z_{27} \\
&\quad - 1.01552z_{28} - 3.07164z_{29} - 0.49774z_{30} \\
&\quad - 0.12792z_{31} - 1.70310z_{32} + 0.61140z_{33} \\
&\quad - 0.32430z_{34} - 0.05694z_{35} - 2.95262z_{36} \\
&\quad - 2.21883z_{37} - 3.20066z_{38} + 0.18282z_{39} \\
&\quad + 0.25023z_{40} - 0.29548z_{41} + 0.72588z_{42}
\end{aligned}\]

\noindent\textbf{Cutoff $R_0>9.5$.}
\[\begin{aligned}
\eta_{R_0>9.5} &= -14.90360 \\
&\quad + 1.93144z_{1} - 0.57039z_{2} - 1.08536z_{3} \\
&\quad + 3.56919z_{4} + 0.15741z_{5} - 0.75277z_{6} \\
&\quad + 1.95920z_{7} + 2.12535z_{8} + 1.93775z_{9} \\
&\quad + 1.30476z_{10} + 0.34643z_{11} + 1.55126z_{12} \\
&\quad + 0.14589z_{13} + 1.72665z_{14} - 5.09716z_{15} \\
&\quad - 4.00535z_{16} - 0.29006z_{17} + 0.51655z_{18} \\
&\quad + 1.07472z_{19} - 0.48601z_{20} - 0.59525z_{21} \\
&\quad - 2.33841z_{22} - 0.84337z_{23} - 0.80763z_{24} \\
&\quad + 2.89003z_{25} + 2.90386z_{26} - 0.78892z_{27} \\
&\quad - 1.13430z_{28} - 3.25522z_{29} - 0.32925z_{30} \\
&\quad + 0.17638z_{31} - 2.38840z_{32} + 0.55010z_{33} \\
&\quad - 0.09172z_{34} - 0.36019z_{35} - 2.91016z_{36} \\
&\quad - 2.57909z_{37} - 2.55206z_{38} - 0.38010z_{39} \\
&\quad + 0.35612z_{40} - 0.06306z_{41} + 1.81130z_{42}
\end{aligned}\]

\noindent\textbf{Cutoff $R_0>10$.}
\[\begin{aligned}
\eta_{R_0>10} &= -14.30425 \\
&\quad + 1.50753z_{1} - 0.56550z_{2} - 1.12201z_{3} \\
&\quad + 4.00676z_{4} - 0.36753z_{5} - 0.15066z_{6} \\
&\quad + 2.78118z_{7} + 2.07750z_{8} + 1.49171z_{9} \\
&\quad + 1.03568z_{10} + 1.42563z_{11} + 0.41303z_{12} \\
&\quad + 0.44255z_{13} + 0.95635z_{14} - 4.75061z_{15} \\
&\quad - 3.95850z_{16} + 0.32633z_{17} + 0.19301z_{18} \\
&\quad + 0.64598z_{19} + 0.42021z_{20} - 0.66549z_{21} \\
&\quad - 1.93118z_{22} - 0.65990z_{23} - 0.76568z_{24} \\
&\quad + 2.95914z_{25} + 2.92203z_{26} - 0.81768z_{27} \\
&\quad - 0.88065z_{28} - 3.04952z_{29} - 0.61288z_{30} \\
&\quad + 0.17801z_{31} - 1.81765z_{32} + 0.33178z_{33} \\
&\quad + 0.06814z_{34} - 0.16709z_{35} - 3.02667z_{36} \\
&\quad - 2.67522z_{37} - 2.54307z_{38} + 0.69079z_{39} \\
&\quad - 0.16420z_{40} - 0.33495z_{41} + 1.01365z_{42}
\end{aligned}\]

\subsubsection{15-compartment five-group SIR: validation and equations}

\begingroup\scriptsize\setlength{\LTcapwidth}{\linewidth}\setlength{\tabcolsep}{3pt}

\begin{longtable}{@{}lrrrrrrr@{}}

\caption{Class-specific held-out metrics for 15-compartment five-group SIR with log cubic $h=0.5$.}\label{tab:supp-class-log_cubic_h0p5-15-compartment-five-group-sir}\\

\toprule Bracket & Support & Recall & Spec. & Precision & F1 & NPV & OvR AUC \\ \midrule\endfirsthead

\toprule Bracket & Support & Recall & Spec. & Precision & F1 & NPV & OvR AUC \\ \midrule\endhead

 \texttt{$<$1} & 776 & 0.964 & 0.998 & 0.977 & 0.970 & 0.996 & 1.000 \\

 \texttt{1-1.5} & 495 & 0.917 & 0.993 & 0.899 & 0.908 & 0.995 & 0.997 \\

 \texttt{1.5-2} & 435 & 0.860 & 0.993 & 0.876 & 0.868 & 0.992 & 0.995 \\

 \texttt{2-2.5} & 452 & 0.796 & 0.989 & 0.809 & 0.803 & 0.988 & 0.990 \\

 \texttt{2.5-3} & 466 & 0.773 & 0.985 & 0.756 & 0.764 & 0.986 & 0.984 \\

 \texttt{3-3.5} & 513 & 0.772 & 0.983 & 0.757 & 0.764 & 0.984 & 0.983 \\

 \texttt{3.5-4} & 430 & 0.667 & 0.983 & 0.688 & 0.678 & 0.981 & 0.974 \\

 \texttt{4-4.5} & 444 & 0.664 & 0.979 & 0.650 & 0.657 & 0.980 & 0.973 \\

 \texttt{4.5-5} & 431 & 0.668 & 0.980 & 0.652 & 0.660 & 0.981 & 0.975 \\

 \texttt{5-5.5} & 394 & 0.584 & 0.979 & 0.585 & 0.584 & 0.978 & 0.968 \\

 \texttt{5.5-6} & 365 & 0.570 & 0.977 & 0.545 & 0.557 & 0.979 & 0.969 \\

 \texttt{6-6.5} & 298 & 0.507 & 0.981 & 0.512 & 0.509 & 0.981 & 0.968 \\

 \texttt{6.5-7} & 299 & 0.515 & 0.977 & 0.462 & 0.487 & 0.981 & 0.963 \\

 \texttt{7-7.5} & 260 & 0.288 & 0.979 & 0.315 & 0.301 & 0.976 & 0.959 \\

 \texttt{7.5-8} & 237 & 0.367 & 0.979 & 0.347 & 0.357 & 0.981 & 0.960 \\

 \texttt{8-8.5} & 215 & 0.312 & 0.982 & 0.322 & 0.317 & 0.981 & 0.957 \\

 \texttt{8.5-9} & 216 & 0.375 & 0.980 & 0.343 & 0.358 & 0.983 & 0.962 \\

 \texttt{9-9.5} & 170 & 0.153 & 0.992 & 0.302 & 0.203 & 0.982 & 0.953 \\

 \texttt{9.5-10} & 129 & 0.047 & 0.997 & 0.188 & 0.075 & 0.985 & 0.961 \\

 \texttt{$>$10} & 975 & 0.962 & 0.978 & 0.860 & 0.908 & 0.995 & 0.997 \\

\bottomrule\end{longtable}\endgroup

\paragraph{Ordered-threshold feature map and equations.}

All $33$ fitted predictors are standardized using $z_j=(\phi_j(\boldsymbol\theta)-\bar\phi_j)/s_j$, where the values below use the complete training sample for the printed (full-data) equations. Held-out predictions separately recompute means, standard deviations, and fits within each training fold. Here $\phi_j$ is the feature shown in the table. The score for every estimable cutoff is $q_c=\operatorname{logit}^{-1}(\eta_c)$; positive scores increase the fitted exceedance probability, not the numerical reproduction number.

\begingroup\scriptsize\setlength{\LTcapwidth}{\linewidth}

\begin{longtable}{@{}r p{0.44\linewidth} r r@{}}

\caption{Standardization constants and ordered features for 15-compartment five-group SIR (compact log-cubic basis).}\label{tab:supp-features-15-compartment-five-group-sir-fine}\\

\toprule Index & Raw feature $\phi_j$ & Training mean & Training SD \\ \midrule\endfirsthead

\toprule Index & Raw feature $\phi_j$ & Training mean & Training SD \\ \midrule\endhead

 1 & $\log(\beta)$ & -2.134574 & 0.843287 \\

 2 & $\log(\gamma_1)$ & -1.437303 & 0.587779 \\

 3 & $\log(\gamma_2)$ & -1.437305 & 0.587782 \\

 4 & $\log(\gamma_3)$ & -1.437305 & 0.587784 \\

 5 & $\log(\gamma_4)$ & -1.437306 & 0.587785 \\

 6 & $\log(\gamma_5)$ & -1.437304 & 0.587784 \\

 7 & $[\log(\beta)]^2$ & 5.267448 & 4.610981 \\

 8 & $[\log(\gamma_1)]^2$ & 2.411280 & 1.990078 \\

 9 & $[\log(\gamma_2)]^2$ & 2.411290 & 1.990094 \\

 10 & $[\log(\gamma_3)]^2$ & 2.411292 & 1.990107 \\

 11 & $[\log(\gamma_4)]^2$ & 2.411298 & 1.990113 \\

 12 & $[\log(\gamma_5)]^2$ & 2.411290 & 1.990108 \\

 13 & $\log(\beta)\,\log(\gamma_1)$ & 3.074725 & 1.834870 \\

 14 & $\log(\beta)\,\log(\gamma_2)$ & 3.059064 & 1.784483 \\

 15 & $\log(\beta)\,\log(\gamma_3)$ & 3.064852 & 1.803930 \\

 16 & $\log(\beta)\,\log(\gamma_4)$ & 3.072759 & 1.837645 \\

 17 & $\log(\beta)\,\log(\gamma_5)$ & 3.070575 & 1.829424 \\

 18 & $\log(\gamma_1)\,\log(\gamma_2)$ & 2.076183 & 1.265559 \\

 19 & $\log(\gamma_1)\,\log(\gamma_3)$ & 2.067733 & 1.248415 \\

 20 & $\log(\gamma_1)\,\log(\gamma_4)$ & 2.059412 & 1.223590 \\

 21 & $\log(\gamma_1)\,\log(\gamma_5)$ & 2.065621 & 1.245217 \\

 22 & $\log(\gamma_2)\,\log(\gamma_3)$ & 2.061123 & 1.225565 \\

 23 & $\log(\gamma_2)\,\log(\gamma_4)$ & 2.061075 & 1.235237 \\

 24 & $\log(\gamma_2)\,\log(\gamma_5)$ & 2.075254 & 1.274534 \\

 25 & $\log(\gamma_3)\,\log(\gamma_4)$ & 2.060658 & 1.224150 \\

 26 & $\log(\gamma_3)\,\log(\gamma_5)$ & 2.069393 & 1.253160 \\

 27 & $\log(\gamma_4)\,\log(\gamma_5)$ & 2.067508 & 1.242058 \\

 28 & $[\log(\beta)]^3$ & -15.056142 & 21.505418 \\

 29 & $[\log(\gamma_1)]^3$ & -4.616645 & 5.652102 \\

 30 & $[\log(\gamma_2)]^3$ & -4.616682 & 5.652166 \\

 31 & $[\log(\gamma_3)]^3$ & -4.616695 & 5.652225 \\

 32 & $[\log(\gamma_4)]^3$ & -4.616712 & 5.652248 \\

 33 & $[\log(\gamma_5)]^3$ & -4.616690 & 5.652230 \\

\bottomrule\end{longtable}\endgroup

Each displayed score below is a numeric equation in the standardized variables $z_1,\ldots,z_33$. The intercept and slopes are posterior modes from the independent $N(0,2^2)$ coefficient prior under the Gaussian Laplace approximation. A cutoff with only one observed class cannot produce a fitted finite logistic score, so its constant exceedance probability is stated explicitly.

\noindent\textbf{Cutoff $R_0>1$.}
\[\begin{aligned}
\eta_{R_0>1} &= 12.51521 \\
&\quad + 0.64228z_{1} - 1.11768z_{2} - 1.21598z_{3} \\
&\quad - 2.00079z_{4} - 1.13246z_{5} - 1.27749z_{6} \\
&\quad - 4.72631z_{7} + 1.77743z_{8} + 1.79304z_{9} \\
&\quad + 2.16055z_{10} + 1.55014z_{11} + 0.61652z_{12} \\
&\quad + 1.25946z_{13} + 1.89142z_{14} + 0.50183z_{15} \\
&\quad + 0.35597z_{16} + 1.39300z_{17} - 0.84626z_{18} \\
&\quad - 1.91145z_{19} - 1.62251z_{20} - 1.72862z_{21} \\
&\quad - 1.08751z_{22} - 3.05765z_{23} - 1.98146z_{24} \\
&\quad - 0.92871z_{25} - 2.37432z_{26} - 0.35457z_{27} \\
&\quad + 7.96377z_{28} - 2.93444z_{29} - 3.13763z_{30} \\
&\quad - 2.74948z_{31} - 3.53574z_{32} - 2.67912z_{33}
\end{aligned}\]

\noindent\textbf{Cutoff $R_0>1.5$.}
\[\begin{aligned}
\eta_{R_0>1.5} &= 13.15850 \\
&\quad + 2.27106z_{1} - 2.50741z_{2} - 3.15992z_{3} \\
&\quad - 2.33303z_{4} - 3.28865z_{5} - 3.06234z_{6} \\
&\quad - 5.27479z_{7} + 2.34483z_{8} + 2.26454z_{9} \\
&\quad + 3.06413z_{10} + 2.52472z_{11} + 2.25178z_{12} \\
&\quad + 0.13440z_{13} + 0.32851z_{14} + 0.45818z_{15} \\
&\quad - 1.60475z_{16} - 0.75356z_{17} - 0.80390z_{18} \\
&\quad - 2.50591z_{19} - 2.09967z_{20} - 2.14343z_{21} \\
&\quad - 1.56685z_{22} - 2.82994z_{23} - 2.19180z_{24} \\
&\quad - 1.09932z_{25} - 2.56167z_{26} - 0.57282z_{27} \\
&\quad + 7.31406z_{28} - 3.74916z_{29} - 3.44468z_{30} \\
&\quad - 3.32164z_{31} - 3.30707z_{32} - 2.08511z_{33}
\end{aligned}\]

\noindent\textbf{Cutoff $R_0>2$.}
\[\begin{aligned}
\eta_{R_0>2} &= 11.38801 \\
&\quad + 2.02373z_{1} - 2.51199z_{2} - 3.89814z_{3} \\
&\quad - 2.85787z_{4} - 2.97169z_{5} - 1.48234z_{6} \\
&\quad - 4.91494z_{7} + 3.12022z_{8} + 2.71723z_{9} \\
&\quad + 2.72460z_{10} + 3.31877z_{11} + 1.67249z_{12} \\
&\quad - 2.15424z_{13} - 0.52294z_{14} - 2.60632z_{15} \\
&\quad - 2.48034z_{16} - 0.79363z_{17} - 1.25392z_{18} \\
&\quad - 1.28140z_{19} - 2.06243z_{20} - 1.26790z_{21} \\
&\quad - 1.77366z_{22} - 2.56754z_{23} - 2.18114z_{24} \\
&\quad - 0.29037z_{25} - 1.53690z_{26} - 0.03179z_{27} \\
&\quad + 7.03847z_{28} - 4.50094z_{29} - 4.00095z_{30} \\
&\quad - 4.39083z_{31} - 2.72349z_{32} - 3.36551z_{33}
\end{aligned}\]

\noindent\textbf{Cutoff $R_0>2.5$.}
\[\begin{aligned}
\eta_{R_0>2.5} &= 8.48223 \\
&\quad + 2.89205z_{1} - 3.88419z_{2} - 4.75489z_{3} \\
&\quad - 3.69577z_{4} - 3.69108z_{5} - 3.57534z_{6} \\
&\quad - 3.93972z_{7} + 2.17293z_{8} + 3.76184z_{9} \\
&\quad + 2.99210z_{10} + 1.60146z_{11} + 2.06948z_{12} \\
&\quad - 1.88615z_{13} - 4.21328z_{14} - 2.76488z_{15} \\
&\quad - 2.18406z_{16} - 2.81337z_{17} - 1.00205z_{18} \\
&\quad - 2.01508z_{19} - 3.32995z_{20} - 2.06693z_{21} \\
&\quad - 1.60222z_{22} - 2.01589z_{23} - 2.26424z_{24} \\
&\quad - 0.98328z_{25} - 1.74473z_{26} - 0.55217z_{27} \\
&\quad + 4.79881z_{28} - 5.25037z_{29} - 4.68783z_{30} \\
&\quad - 4.27971z_{31} - 4.68831z_{32} - 3.45110z_{33}
\end{aligned}\]

\noindent\textbf{Cutoff $R_0>3$.}
\[\begin{aligned}
\eta_{R_0>3} &= 5.61181 \\
&\quad + 3.87065z_{1} - 4.66031z_{2} - 4.74500z_{3} \\
&\quad - 4.04285z_{4} - 4.47133z_{5} - 3.59069z_{6} \\
&\quad - 3.76742z_{7} + 3.40818z_{8} + 5.02530z_{9} \\
&\quad + 2.25164z_{10} + 2.53116z_{11} + 1.30341z_{12} \\
&\quad - 3.96450z_{13} - 4.82610z_{14} - 3.26612z_{15} \\
&\quad - 3.03542z_{16} - 1.87901z_{17} - 0.91963z_{18} \\
&\quad - 2.13528z_{19} - 3.23964z_{20} - 2.77611z_{21} \\
&\quad - 1.69870z_{22} - 2.71828z_{23} - 2.33472z_{24} \\
&\quad - 1.07598z_{25} - 1.05468z_{26} - 1.63499z_{27} \\
&\quad + 4.56552z_{28} - 5.44957z_{29} - 5.04605z_{30} \\
&\quad - 5.30484z_{31} - 5.29359z_{32} - 4.53346z_{33}
\end{aligned}\]

\noindent\textbf{Cutoff $R_0>3.5$.}
\[\begin{aligned}
\eta_{R_0>3.5} &= 1.75199 \\
&\quad + 2.72999z_{1} - 4.92553z_{2} - 5.32883z_{3} \\
&\quad - 5.02019z_{4} - 4.75795z_{5} - 2.87610z_{6} \\
&\quad - 3.05213z_{7} + 3.22837z_{8} + 4.19011z_{9} \\
&\quad + 3.47399z_{10} + 2.30064z_{11} + 1.75369z_{12} \\
&\quad - 2.49088z_{13} - 5.94119z_{14} - 5.08196z_{15} \\
&\quad - 5.51481z_{16} - 1.90578z_{17} - 1.30798z_{18} \\
&\quad - 2.79519z_{19} - 2.53446z_{20} - 3.22712z_{21} \\
&\quad - 1.85234z_{22} - 2.24967z_{23} - 1.68155z_{24} \\
&\quad - 0.18614z_{25} - 1.77008z_{26} - 0.40054z_{27} \\
&\quad + 3.66843z_{28} - 4.51964z_{29} - 5.36998z_{30} \\
&\quad - 4.43514z_{31} - 4.31291z_{32} - 4.05087z_{33}
\end{aligned}\]

\noindent\textbf{Cutoff $R_0>4$.}
\[\begin{aligned}
\eta_{R_0>4} &= -1.70037 \\
&\quad + 1.31802z_{1} - 4.64957z_{2} - 4.57716z_{3} \\
&\quad - 4.44406z_{4} - 4.33498z_{5} - 3.78844z_{6} \\
&\quad - 3.20611z_{7} + 3.11333z_{8} + 4.82557z_{9} \\
&\quad + 3.51099z_{10} + 2.32861z_{11} + 0.84142z_{12} \\
&\quad - 5.32277z_{13} - 6.34672z_{14} - 6.01778z_{15} \\
&\quad - 5.63741z_{16} - 1.35272z_{17} - 0.75100z_{18} \\
&\quad - 2.12581z_{19} - 2.37941z_{20} - 2.33421z_{21} \\
&\quad - 0.91673z_{22} - 2.38294z_{23} - 2.31501z_{24} \\
&\quad - 0.39018z_{25} - 2.78547z_{26} - 0.83118z_{27} \\
&\quad + 4.28574z_{28} - 5.51999z_{29} - 4.95221z_{30} \\
&\quad - 5.36548z_{31} - 5.25320z_{32} - 4.42441z_{33}
\end{aligned}\]

\noindent\textbf{Cutoff $R_0>4.5$.}
\[\begin{aligned}
\eta_{R_0>4.5} &= -4.30806 \\
&\quad + 1.60554z_{1} - 4.84990z_{2} - 5.92296z_{3} \\
&\quad - 4.76623z_{4} - 4.01239z_{5} - 3.76835z_{6} \\
&\quad - 2.56639z_{7} + 2.77383z_{8} + 3.71887z_{9} \\
&\quad + 3.08292z_{10} + 0.97013z_{11} + 1.73982z_{12} \\
&\quad - 5.59146z_{13} - 6.21601z_{14} - 6.08861z_{15} \\
&\quad - 3.33100z_{16} - 3.90565z_{17} - 1.41419z_{18} \\
&\quad - 2.24946z_{19} - 2.82522z_{20} - 1.82817z_{21} \\
&\quad - 0.99397z_{22} - 2.48369z_{23} - 3.01681z_{24} \\
&\quad - 1.06189z_{25} - 1.47621z_{26} - 0.27295z_{27} \\
&\quad + 3.05367z_{28} - 6.07521z_{29} - 5.05930z_{30} \\
&\quad - 4.86591z_{31} - 5.63884z_{32} - 3.97435z_{33}
\end{aligned}\]

\noindent\textbf{Cutoff $R_0>5$.}
\[\begin{aligned}
\eta_{R_0>5} &= -7.23438 \\
&\quad - 0.22627z_{1} - 4.52487z_{2} - 5.21027z_{3} \\
&\quad - 4.09335z_{4} - 3.60171z_{5} - 3.00327z_{6} \\
&\quad - 2.85271z_{7} + 3.08168z_{8} + 3.89569z_{9} \\
&\quad + 4.18834z_{10} + 2.32770z_{11} + 0.65824z_{12} \\
&\quad - 4.98818z_{13} - 8.35480z_{14} - 7.38828z_{15} \\
&\quad - 5.20616z_{16} - 2.71832z_{17} - 0.80883z_{18} \\
&\quad - 2.27660z_{19} - 2.63021z_{20} - 2.32667z_{21} \\
&\quad - 0.71861z_{22} - 2.33123z_{23} - 2.24135z_{24} \\
&\quad - 0.45002z_{25} - 1.19584z_{26} + 0.19789z_{27} \\
&\quad + 3.55933z_{28} - 5.35507z_{29} - 6.00406z_{30} \\
&\quad - 4.44781z_{31} - 4.78268z_{32} - 4.61721z_{33}
\end{aligned}\]

\noindent\textbf{Cutoff $R_0>5.5$.}
\[\begin{aligned}
\eta_{R_0>5.5} &= -9.43565 \\
&\quad - 0.96492z_{1} - 4.13178z_{2} - 4.34966z_{3} \\
&\quad - 3.77583z_{4} - 3.74310z_{5} - 2.99356z_{6} \\
&\quad - 2.89122z_{7} + 3.84829z_{8} + 4.10573z_{9} \\
&\quad + 2.81063z_{10} + 1.51374z_{11} + 1.02415z_{12} \\
&\quad - 4.57503z_{13} - 7.73942z_{14} - 6.95656z_{15} \\
&\quad - 7.03814z_{16} - 2.43164z_{17} - 1.60994z_{18} \\
&\quad - 2.15770z_{19} - 2.54742z_{20} - 2.31319z_{21} \\
&\quad - 0.35246z_{22} - 1.45194z_{23} - 2.26288z_{24} \\
&\quad - 0.06726z_{25} - 1.42528z_{26} + 0.44961z_{27} \\
&\quad + 3.43677z_{28} - 4.68135z_{29} - 5.65010z_{30} \\
&\quad - 5.01172z_{31} - 5.04245z_{32} - 3.88027z_{33}
\end{aligned}\]

\noindent\textbf{Cutoff $R_0>6$.}
\[\begin{aligned}
\eta_{R_0>6} &= -11.26236 \\
&\quad - 1.83036z_{1} - 3.38953z_{2} - 4.54989z_{3} \\
&\quad - 3.66939z_{4} - 3.67896z_{5} - 2.51515z_{6} \\
&\quad - 2.42313z_{7} + 2.18337z_{8} + 3.63002z_{9} \\
&\quad + 2.91229z_{10} + 0.81345z_{11} + 1.52214z_{12} \\
&\quad - 4.84355z_{13} - 8.05937z_{14} - 6.77186z_{15} \\
&\quad - 6.90685z_{16} - 3.68868z_{17} - 0.67987z_{18} \\
&\quad - 1.68382z_{19} - 1.83455z_{20} - 2.08749z_{21} \\
&\quad - 0.85677z_{22} - 1.97404z_{23} - 1.98636z_{24} \\
&\quad - 0.24078z_{25} - 1.53921z_{26} + 0.15966z_{27} \\
&\quad + 2.96825z_{28} - 5.61182z_{29} - 5.35340z_{30} \\
&\quad - 4.66616z_{31} - 5.64984z_{32} - 4.32231z_{33}
\end{aligned}\]

\noindent\textbf{Cutoff $R_0>6.5$.}
\[\begin{aligned}
\eta_{R_0>6.5} &= -12.48354 \\
&\quad - 1.26218z_{1} - 3.73856z_{2} - 4.12251z_{3} \\
&\quad - 3.48193z_{4} - 3.06853z_{5} - 2.31534z_{6} \\
&\quad - 1.64363z_{7} + 2.11665z_{8} + 2.97681z_{9} \\
&\quad + 2.36277z_{10} + 0.87006z_{11} + 0.78605z_{12} \\
&\quad - 5.03324z_{13} - 7.09725z_{14} - 7.19586z_{15} \\
&\quad - 8.28784z_{16} - 3.19878z_{17} - 1.29931z_{18} \\
&\quad - 1.77325z_{19} - 1.31115z_{20} - 2.41294z_{21} \\
&\quad - 1.32664z_{22} - 1.68288z_{23} - 1.60588z_{24} \\
&\quad + 0.17781z_{25} - 1.03465z_{26} + 0.32070z_{27} \\
&\quad + 1.61334z_{28} - 5.66842z_{29} - 5.92518z_{30} \\
&\quad - 5.28090z_{31} - 6.10143z_{32} - 4.48358z_{33}
\end{aligned}\]

\noindent\textbf{Cutoff $R_0>7$.}
\[\begin{aligned}
\eta_{R_0>7} &= -12.49967 \\
&\quad - 1.56455z_{1} - 2.96235z_{2} - 3.58202z_{3} \\
&\quad - 3.09278z_{4} - 2.74936z_{5} - 3.49985z_{6} \\
&\quad - 1.16544z_{7} + 1.07228z_{8} + 2.18073z_{9} \\
&\quad + 1.19454z_{10} + 0.69160z_{11} - 1.01287z_{12} \\
&\quad - 4.31542z_{13} - 8.10729z_{14} - 7.15109z_{15} \\
&\quad - 5.59402z_{16} - 3.29468z_{17} - 0.89627z_{18} \\
&\quad - 1.12222z_{19} - 1.45470z_{20} - 1.62564z_{21} \\
&\quad - 0.66379z_{22} - 1.57785z_{23} - 1.53730z_{24} \\
&\quad + 0.02538z_{25} - 1.55673z_{26} - 0.42502z_{27} \\
&\quad + 0.99404z_{28} - 5.31267z_{29} - 6.25334z_{30} \\
&\quad - 5.51646z_{31} - 4.99601z_{32} - 4.91321z_{33}
\end{aligned}\]

\noindent\textbf{Cutoff $R_0>7.5$.}
\[\begin{aligned}
\eta_{R_0>7.5} &= -13.43555 \\
&\quad - 1.15783z_{1} - 2.92710z_{2} - 2.87357z_{3} \\
&\quad - 2.22633z_{4} - 2.88854z_{5} - 2.94824z_{6} \\
&\quad - 0.74054z_{7} + 0.98207z_{8} + 1.93267z_{9} \\
&\quad + 1.15970z_{10} - 0.68245z_{11} - 0.76420z_{12} \\
&\quad - 4.82821z_{13} - 7.79255z_{14} - 6.32347z_{15} \\
&\quad - 5.48503z_{16} - 3.96752z_{17} - 0.48870z_{18} \\
&\quad - 1.40802z_{19} - 1.51210z_{20} - 1.39273z_{21} \\
&\quad + 0.11214z_{22} - 1.38876z_{23} - 1.38064z_{24} \\
&\quad + 0.01467z_{25} - 1.00083z_{26} - 0.21081z_{27} \\
&\quad + 0.16821z_{28} - 5.12220z_{29} - 5.85024z_{30} \\
&\quad - 5.00174z_{31} - 5.60074z_{32} - 4.52811z_{33}
\end{aligned}\]

\noindent\textbf{Cutoff $R_0>8$.}
\[\begin{aligned}
\eta_{R_0>8} &= -14.37453 \\
&\quad - 1.85934z_{1} - 2.83295z_{2} - 2.14814z_{3} \\
&\quad - 1.90673z_{4} - 2.33445z_{5} - 3.02280z_{6} \\
&\quad - 0.45270z_{7} + 0.59991z_{8} + 2.48920z_{9} \\
&\quad + 0.15076z_{10} - 0.72304z_{11} - 0.39343z_{12} \\
&\quad - 5.56767z_{13} - 7.68054z_{14} - 6.27136z_{15} \\
&\quad - 4.88675z_{16} - 4.98964z_{17} - 0.27517z_{18} \\
&\quad - 0.91655z_{19} - 1.68109z_{20} - 1.19454z_{21} \\
&\quad + 0.03588z_{22} - 0.96141z_{23} - 1.37284z_{24} \\
&\quad - 0.07146z_{25} - 0.82120z_{26} + 0.10244z_{27} \\
&\quad - 0.04403z_{28} - 5.45115z_{29} - 5.30055z_{30} \\
&\quad - 5.80328z_{31} - 5.38951z_{32} - 4.03796z_{33}
\end{aligned}\]

\noindent\textbf{Cutoff $R_0>8.5$.}
\[\begin{aligned}
\eta_{R_0>8.5} &= -14.34156 \\
&\quad - 0.26742z_{1} - 2.55961z_{2} - 2.33778z_{3} \\
&\quad - 1.24202z_{4} - 2.22860z_{5} - 3.50679z_{6} \\
&\quad - 0.31646z_{7} - 0.54615z_{8} + 0.49575z_{9} \\
&\quad + 0.30054z_{10} - 0.36087z_{11} - 1.47062z_{12} \\
&\quad - 4.80626z_{13} - 7.29328z_{14} - 5.40763z_{15} \\
&\quad - 4.43293z_{16} - 3.74980z_{17} + 0.10059z_{18} \\
&\quad - 0.81071z_{19} - 1.43717z_{20} - 1.28839z_{21} \\
&\quad + 0.11543z_{22} - 0.97070z_{23} - 1.27657z_{24} \\
&\quad + 0.14194z_{25} - 1.10461z_{26} - 0.23419z_{27} \\
&\quad - 0.75432z_{28} - 5.65294z_{29} - 6.07908z_{30} \\
&\quad - 5.37726z_{31} - 4.53671z_{32} - 4.19009z_{33}
\end{aligned}\]

\noindent\textbf{Cutoff $R_0>9$.}
\[\begin{aligned}
\eta_{R_0>9} &= -14.14523 \\
&\quad - 1.41480z_{1} - 2.52626z_{2} - 2.50660z_{3} \\
&\quad - 2.73001z_{4} - 2.05615z_{5} - 2.89268z_{6} \\
&\quad - 0.01202z_{7} - 0.92531z_{8} + 0.18149z_{9} \\
&\quad - 0.53874z_{10} - 0.98028z_{11} - 1.90383z_{12} \\
&\quad - 4.59287z_{13} - 7.20015z_{14} - 6.83156z_{15} \\
&\quad - 4.12268z_{16} - 2.89400z_{17} - 0.35731z_{18} \\
&\quad - 1.15312z_{19} - 1.14549z_{20} - 0.68511z_{21} \\
&\quad + 0.40394z_{22} - 0.85939z_{23} - 1.25510z_{24} \\
&\quad + 0.23110z_{25} - 1.02752z_{26} + 0.15313z_{27} \\
&\quad - 0.80685z_{28} - 5.46651z_{29} - 5.68838z_{30} \\
&\quad - 4.95505z_{31} - 4.35143z_{32} - 3.92027z_{33}
\end{aligned}\]

\noindent\textbf{Cutoff $R_0>9.5$.}
\[\begin{aligned}
\eta_{R_0>9.5} &= -15.27348 \\
&\quad - 0.65196z_{1} - 2.16938z_{2} - 1.96413z_{3} \\
&\quad - 2.01021z_{4} - 1.93512z_{5} - 1.97319z_{6} \\
&\quad + 0.35354z_{7} - 1.00743z_{8} - 0.01639z_{9} \\
&\quad - 0.48547z_{10} - 0.85409z_{11} - 1.91836z_{12} \\
&\quad - 5.29715z_{13} - 7.24066z_{14} - 6.69925z_{15} \\
&\quad - 3.50544z_{16} - 3.31924z_{17} + 0.03534z_{18} \\
&\quad - 0.87279z_{19} - 1.29444z_{20} - 0.34201z_{21} \\
&\quad + 0.50765z_{22} - 1.10053z_{23} - 1.15857z_{24} \\
&\quad + 0.26210z_{25} - 0.69111z_{26} + 0.21030z_{27} \\
&\quad - 1.53550z_{28} - 5.61612z_{29} - 6.12680z_{30} \\
&\quad - 4.96990z_{31} - 4.23817z_{32} - 4.39709z_{33}
\end{aligned}\]

\noindent\textbf{Cutoff $R_0>10$.}
\[\begin{aligned}
\eta_{R_0>10} &= -14.91433 \\
&\quad - 0.83540z_{1} - 2.56475z_{2} - 1.64606z_{3} \\
&\quad - 1.97317z_{4} - 1.59399z_{5} - 2.06815z_{6} \\
&\quad + 0.79202z_{7} - 1.04031z_{8} - 0.55983z_{9} \\
&\quad - 1.13908z_{10} - 1.23892z_{11} - 1.20850z_{12} \\
&\quad - 5.75083z_{13} - 6.51096z_{14} - 6.91887z_{15} \\
&\quad - 2.99475z_{16} - 3.47356z_{17} + 0.15006z_{18} \\
&\quad - 0.90222z_{19} - 1.44254z_{20} - 0.61219z_{21} \\
&\quad + 0.40405z_{22} - 0.90885z_{23} - 0.90997z_{24} \\
&\quad + 0.54959z_{25} - 0.74273z_{26} - 0.05348z_{27} \\
&\quad - 1.92937z_{28} - 5.37500z_{29} - 6.00028z_{30} \\
&\quad - 5.44799z_{31} - 4.41956z_{32} - 3.61612z_{33}
\end{aligned}\]

\subsubsection{Vector-borne: validation and equations}

\begingroup\scriptsize\setlength{\LTcapwidth}{\linewidth}\setlength{\tabcolsep}{3pt}

\begin{longtable}{@{}lrrrrrrr@{}}

\caption{Class-specific held-out metrics for Vector-borne with log cubic $h=0.5$.}\label{tab:supp-class-log_cubic_h0p5-vector-borne}\\

\toprule Bracket & Support & Recall & Spec. & Precision & F1 & NPV & OvR AUC \\ \midrule\endfirsthead

\toprule Bracket & Support & Recall & Spec. & Precision & F1 & NPV & OvR AUC \\ \midrule\endhead

 \texttt{$<$1} & 380 & 0.961 & 1.000 & 0.995 & 0.977 & 0.998 & 1.000 \\

 \texttt{1-1.5} & 585 & 0.976 & 0.997 & 0.966 & 0.971 & 0.998 & 1.000 \\

 \texttt{1.5-2} & 629 & 0.952 & 0.997 & 0.965 & 0.958 & 0.996 & 1.000 \\

 \texttt{2-2.5} & 722 & 0.978 & 0.994 & 0.944 & 0.961 & 0.998 & 1.000 \\

 \texttt{2.5-3} & 658 & 0.968 & 0.997 & 0.962 & 0.965 & 0.997 & 1.000 \\

 \texttt{3-3.5} & 569 & 0.940 & 0.997 & 0.966 & 0.953 & 0.995 & 0.999 \\

 \texttt{3.5-4} & 612 & 0.959 & 0.996 & 0.950 & 0.954 & 0.997 & 0.999 \\

 \texttt{4-4.5} & 532 & 0.955 & 0.996 & 0.948 & 0.951 & 0.997 & 0.999 \\

 \texttt{4.5-5} & 402 & 0.918 & 0.997 & 0.951 & 0.934 & 0.996 & 0.999 \\

 \texttt{5-5.5} & 436 & 0.959 & 0.996 & 0.933 & 0.946 & 0.998 & 0.999 \\

 \texttt{5.5-6} & 331 & 0.931 & 0.997 & 0.939 & 0.935 & 0.997 & 0.999 \\

 \texttt{6-6.5} & 278 & 0.917 & 0.996 & 0.882 & 0.899 & 0.997 & 0.999 \\

 \texttt{6.5-7} & 226 & 0.827 & 0.995 & 0.839 & 0.833 & 0.995 & 0.997 \\

 \texttt{7-7.5} & 193 & 0.813 & 0.996 & 0.818 & 0.816 & 0.995 & 0.997 \\

 \texttt{7.5-8} & 170 & 0.806 & 0.997 & 0.846 & 0.825 & 0.996 & 0.997 \\

 \texttt{8-8.5} & 152 & 0.849 & 0.996 & 0.787 & 0.816 & 0.997 & 0.997 \\

 \texttt{8.5-9} & 141 & 0.695 & 0.995 & 0.726 & 0.710 & 0.995 & 0.995 \\

 \texttt{9-9.5} & 136 & 0.721 & 0.995 & 0.695 & 0.708 & 0.995 & 0.990 \\

 \texttt{9.5-10} & 101 & 0.347 & 0.998 & 0.729 & 0.470 & 0.992 & 0.980 \\

 \texttt{$>$10} & 747 & 0.996 & 0.994 & 0.948 & 0.971 & 1.000 & 1.000 \\

\bottomrule\end{longtable}\endgroup

\paragraph{Ordered-threshold feature map and equations.}

All $25$ fitted predictors are standardized using $z_j=(\phi_j(\boldsymbol\theta)-\bar\phi_j)/s_j$, where the values below use the complete training sample for the printed (full-data) equations. Held-out predictions separately recompute means, standard deviations, and fits within each training fold. Here $\phi_j$ is the feature shown in the table. The score for every estimable cutoff is $q_c=\operatorname{logit}^{-1}(\eta_c)$; positive scores increase the fitted exceedance probability, not the numerical reproduction number.

\begingroup\scriptsize\setlength{\LTcapwidth}{\linewidth}

\begin{longtable}{@{}r p{0.44\linewidth} r r@{}}

\caption{Standardization constants and ordered features for Vector-borne (compact log-cubic basis).}\label{tab:supp-features-vector-borne-fine}\\

\toprule Index & Raw feature $\phi_j$ & Training mean & Training SD \\ \midrule\endfirsthead

\toprule Index & Raw feature $\phi_j$ & Training mean & Training SD \\ \midrule\endhead

 1 & $\log(\beta_{hv})$ & -0.477252 & 0.766431 \\

 2 & $\log(\beta_{vh})$ & -0.477251 & 0.766429 \\

 3 & $\log(\gamma_h)$ & -1.284924 & 0.622634 \\

 4 & $\log(\mu_v)$ & -1.284926 & 0.622639 \\

 5 & $\log(m)$ & 1.017660 & 0.622636 \\

 6 & $[\log(\beta_{hv})]^2$ & 0.815113 & 1.543980 \\

 7 & $[\log(\beta_{vh})]^2$ & 0.815108 & 1.543964 \\

 8 & $[\log(\gamma_h)]^2$ & 2.038655 & 1.966211 \\

 9 & $[\log(\mu_v)]^2$ & 2.038666 & 1.966240 \\

 10 & $[\log(m)]^2$ & 1.423258 & 1.026810 \\

 11 & $\log(\beta_{hv})\,\log(\beta_{vh})$ & 0.233862 & 0.815595 \\

 12 & $\log(\beta_{hv})\,\log(\gamma_h)$ & 0.618415 & 1.138219 \\

 13 & $\log(\beta_{hv})\,\log(\mu_v)$ & 0.609387 & 1.125993 \\

 14 & $\log(\beta_{hv})\,\log(m)$ & -0.485875 & 0.957396 \\

 15 & $\log(\beta_{vh})\,\log(\gamma_h)$ & 0.615623 & 1.133774 \\

 16 & $\log(\beta_{vh})\,\log(\mu_v)$ & 0.608023 & 1.116597 \\

 17 & $\log(\beta_{vh})\,\log(m)$ & -0.491223 & 0.970204 \\

 18 & $\log(\gamma_h)\,\log(\mu_v)$ & 1.654114 & 1.201854 \\

 19 & $\log(\gamma_h)\,\log(m)$ & -1.306969 & 1.088962 \\

 20 & $\log(\mu_v)\,\log(m)$ & -1.303662 & 1.091837 \\

 21 & $[\log(\beta_{hv})]^3$ & -1.449315 & 3.799209 \\

 22 & $[\log(\beta_{vh})]^3$ & -1.449298 & 3.799143 \\

 23 & $[\log(\gamma_h)]^3$ & -3.817289 & 5.386489 \\

 24 & $[\log(\mu_v)]^3$ & -3.817333 & 5.386617 \\

 25 & $[\log(m)]^3$ & 2.035700 & 1.837440 \\

\bottomrule\end{longtable}\endgroup

Each displayed score below is a numeric equation in the standardized variables $z_1,\ldots,z_25$. The intercept and slopes are posterior modes from the independent $N(0,2^2)$ coefficient prior under the Gaussian Laplace approximation. A cutoff with only one observed class cannot produce a fitted finite logistic score, so its constant exceedance probability is stated explicitly.

\noindent\textbf{Cutoff $R_0>1$.}
\[\begin{aligned}
\eta_{R_0>1} &= 12.58180 \\
&\quad + 3.76674z_{1} + 2.95946z_{2} - 2.48592z_{3} \\
&\quad - 2.44337z_{4} + 1.98039z_{5} - 3.17523z_{6} \\
&\quad - 4.37511z_{7} + 0.72020z_{8} + 0.53916z_{9} \\
&\quad - 0.22135z_{10} - 4.64952z_{11} + 3.31329z_{12} \\
&\quad + 3.18691z_{13} - 3.66019z_{14} + 2.80152z_{15} \\
&\quad + 3.28911z_{16} - 3.54709z_{17} - 0.31832z_{18} \\
&\quad + 0.72519z_{19} + 0.61941z_{20} + 0.95890z_{21} \\
&\quad + 0.18897z_{22} + 0.83136z_{23} + 0.77037z_{24} \\
&\quad - 0.04318z_{25}
\end{aligned}\]

\noindent\textbf{Cutoff $R_0>1.5$.}
\[\begin{aligned}
\eta_{R_0>1.5} &= 15.27889 \\
&\quad + 6.25165z_{1} + 6.18270z_{2} - 2.82233z_{3} \\
&\quad - 2.77603z_{4} + 3.63616z_{5} - 5.47959z_{6} \\
&\quad - 5.78509z_{7} + 1.54887z_{8} + 1.50449z_{9} \\
&\quad - 0.59649z_{10} - 5.72458z_{11} + 2.83877z_{12} \\
&\quad + 3.01848z_{13} - 4.13840z_{14} + 2.91747z_{15} \\
&\quad + 2.99915z_{16} - 3.88349z_{17} + 1.98993z_{18} \\
&\quad - 0.57864z_{19} - 0.04950z_{20} + 0.02994z_{21} \\
&\quad - 0.33946z_{22} - 0.42331z_{23} - 0.50959z_{24} \\
&\quad - 0.06392z_{25}
\end{aligned}\]

\noindent\textbf{Cutoff $R_0>2$.}
\[\begin{aligned}
\eta_{R_0>2} &= 17.00614 \\
&\quad + 8.67935z_{1} + 8.76509z_{2} - 3.05290z_{3} \\
&\quad - 3.76122z_{4} + 4.08005z_{5} - 4.12764z_{6} \\
&\quad - 3.52422z_{7} + 2.24736z_{8} + 2.25165z_{9} \\
&\quad + 0.01992z_{10} - 4.65834z_{11} + 1.67003z_{12} \\
&\quad + 1.20209z_{13} - 2.15806z_{14} + 1.39975z_{15} \\
&\quad + 0.94855z_{16} - 2.30814z_{17} + 4.03130z_{18} \\
&\quad - 3.05326z_{19} - 2.38547z_{20} + 0.48629z_{21} \\
&\quad + 0.91328z_{22} - 0.53997z_{23} - 0.78527z_{24} \\
&\quad + 0.28541z_{25}
\end{aligned}\]

\noindent\textbf{Cutoff $R_0>2.5$.}
\[\begin{aligned}
\eta_{R_0>2.5} &= 15.68741 \\
&\quad + 9.90765z_{1} + 9.90896z_{2} - 3.84995z_{3} \\
&\quad - 4.12919z_{4} + 3.99169z_{5} - 2.79052z_{6} \\
&\quad - 2.68775z_{7} + 2.91041z_{8} + 2.82918z_{9} \\
&\quad + 1.41956z_{10} - 3.15846z_{11} - 0.50293z_{12} \\
&\quad - 0.61349z_{13} - 0.24181z_{14} - 0.91136z_{15} \\
&\quad - 0.66772z_{16} - 0.36007z_{17} + 5.60346z_{18} \\
&\quad - 4.56298z_{19} - 4.31544z_{20} + 0.36137z_{21} \\
&\quad + 0.10659z_{22} - 1.62050z_{23} - 1.24802z_{24} \\
&\quad + 0.44421z_{25}
\end{aligned}\]

\noindent\textbf{Cutoff $R_0>3$.}
\[\begin{aligned}
\eta_{R_0>3} &= 11.60394 \\
&\quad + 10.06829z_{1} + 10.14600z_{2} - 4.09994z_{3} \\
&\quad - 4.23973z_{4} + 3.66850z_{5} - 2.18898z_{6} \\
&\quad - 2.31998z_{7} + 3.72098z_{8} + 3.29322z_{9} \\
&\quad + 1.90831z_{10} - 1.82450z_{11} - 2.02551z_{12} \\
&\quad - 1.84828z_{13} + 1.82710z_{14} - 2.04640z_{15} \\
&\quad - 1.82625z_{16} + 1.81054z_{17} + 6.38966z_{18} \\
&\quad - 6.11004z_{19} - 5.85917z_{20} - 0.40496z_{21} \\
&\quad - 0.31529z_{22} - 1.24724z_{23} - 1.53952z_{24} \\
&\quad + 0.77412z_{25}
\end{aligned}\]

\noindent\textbf{Cutoff $R_0>3.5$.}
\[\begin{aligned}
\eta_{R_0>3.5} &= 6.52853 \\
&\quad + 10.39218z_{1} + 10.63516z_{2} - 4.40742z_{3} \\
&\quad - 4.35649z_{4} + 3.04706z_{5} - 1.04907z_{6} \\
&\quad - 0.80627z_{7} + 3.38052z_{8} + 3.77094z_{9} \\
&\quad + 2.95186z_{10} - 0.64619z_{11} - 3.19373z_{12} \\
&\quad - 3.36098z_{13} + 2.82258z_{14} - 2.90593z_{15} \\
&\quad - 3.31754z_{16} + 2.83028z_{17} + 7.63324z_{18} \\
&\quad - 6.95510z_{19} - 7.28711z_{20} - 0.75729z_{21} \\
&\quad + 0.14168z_{22} - 1.62862z_{23} - 1.41959z_{24} \\
&\quad + 0.20665z_{25}
\end{aligned}\]

\noindent\textbf{Cutoff $R_0>4$.}
\[\begin{aligned}
\eta_{R_0>4} &= 0.72651 \\
&\quad + 9.35875z_{1} + 10.30398z_{2} - 4.88403z_{3} \\
&\quad - 4.55395z_{4} + 2.80223z_{5} - 0.75019z_{6} \\
&\quad + 0.39495z_{7} + 3.50150z_{8} + 4.23818z_{9} \\
&\quad + 3.76875z_{10} + 0.50222z_{11} - 4.31774z_{12} \\
&\quad - 4.28655z_{13} + 4.21435z_{14} - 4.35555z_{15} \\
&\quad - 4.07792z_{16} + 3.82932z_{17} + 7.47302z_{18} \\
&\quad - 7.46248z_{19} - 7.57760z_{20} - 1.28964z_{21} \\
&\quad - 0.30021z_{22} - 1.77212z_{23} - 1.36856z_{24} \\
&\quad + 0.11029z_{25}
\end{aligned}\]

\noindent\textbf{Cutoff $R_0>4.5$.}
\[\begin{aligned}
\eta_{R_0>4.5} &= -4.55077 \\
&\quad + 8.13805z_{1} + 8.63301z_{2} - 4.24822z_{3} \\
&\quad - 4.24134z_{4} + 2.61597z_{5} + 0.59553z_{6} \\
&\quad + 0.78848z_{7} + 3.79319z_{8} + 3.53863z_{9} \\
&\quad + 3.43943z_{10} + 1.63250z_{11} - 4.80259z_{12} \\
&\quad - 4.64382z_{13} + 4.65808z_{14} - 4.21835z_{15} \\
&\quad - 4.43947z_{16} + 4.66673z_{17} + 7.23096z_{18} \\
&\quad - 7.15784z_{19} - 7.11453z_{20} - 0.88622z_{21} \\
&\quad - 0.44448z_{22} - 1.39946z_{23} - 1.65832z_{24} \\
&\quad + 0.14447z_{25}
\end{aligned}\]

\noindent\textbf{Cutoff $R_0>5$.}
\[\begin{aligned}
\eta_{R_0>5} &= -8.71551 \\
&\quad + 7.09572z_{1} + 7.00706z_{2} - 4.08638z_{3} \\
&\quad - 4.04331z_{4} + 2.24710z_{5} + 1.56299z_{6} \\
&\quad + 1.06802z_{7} + 2.88494z_{8} + 3.39587z_{9} \\
&\quad + 2.51527z_{10} + 2.30389z_{11} - 5.10401z_{12} \\
&\quad - 4.80154z_{13} + 4.68765z_{14} - 4.90796z_{15} \\
&\quad - 4.76251z_{16} + 4.99581z_{17} + 7.47645z_{18} \\
&\quad - 7.29582z_{19} - 6.82012z_{20} - 0.22644z_{21} \\
&\quad - 0.69390z_{22} - 1.47954z_{23} - 1.19946z_{24} \\
&\quad + 0.80324z_{25}
\end{aligned}\]

\noindent\textbf{Cutoff $R_0>5.5$.}
\[\begin{aligned}
\eta_{R_0>5.5} &= -12.28136 \\
&\quad + 5.62889z_{1} + 5.37265z_{2} - 3.92778z_{3} \\
&\quad - 3.85579z_{4} + 1.89579z_{5} + 2.12870z_{6} \\
&\quad + 1.88428z_{7} + 2.39573z_{8} + 2.87979z_{9} \\
&\quad + 2.84328z_{10} + 2.69298z_{11} - 5.30039z_{12} \\
&\quad - 4.93264z_{13} + 5.00603z_{14} - 5.55028z_{15} \\
&\quad - 4.83380z_{16} + 4.72827z_{17} + 6.95539z_{18} \\
&\quad - 6.49361z_{19} - 6.59118z_{20} - 0.09692z_{21} \\
&\quad - 0.07930z_{22} - 1.96167z_{23} - 1.17219z_{24} \\
&\quad + 0.41417z_{25}
\end{aligned}\]

\noindent\textbf{Cutoff $R_0>6$.}
\[\begin{aligned}
\eta_{R_0>6} &= -13.59959 \\
&\quad + 3.82989z_{1} + 4.06960z_{2} - 3.44591z_{3} \\
&\quad - 3.51223z_{4} + 1.62609z_{5} + 1.33979z_{6} \\
&\quad + 1.94796z_{7} + 2.74087z_{8} + 2.81439z_{9} \\
&\quad + 2.58040z_{10} + 2.93364z_{11} - 5.09045z_{12} \\
&\quad - 4.87283z_{13} + 4.71757z_{14} - 4.96498z_{15} \\
&\quad - 5.07826z_{16} + 4.67680z_{17} + 5.67008z_{18} \\
&\quad - 5.82149z_{19} - 5.64395z_{20} - 1.82316z_{21} \\
&\quad - 1.24788z_{22} - 1.27913z_{23} - 1.35359z_{24} \\
&\quad + 0.37821z_{25}
\end{aligned}\]

\noindent\textbf{Cutoff $R_0>6.5$.}
\[\begin{aligned}
\eta_{R_0>6.5} &= -15.00736 \\
&\quad + 2.76260z_{1} + 3.01001z_{2} - 3.01606z_{3} \\
&\quad - 2.82681z_{4} + 1.13160z_{5} + 2.51828z_{6} \\
&\quad + 2.18244z_{7} + 2.36277z_{8} + 2.52883z_{9} \\
&\quad + 1.04583z_{10} + 2.59038z_{11} - 4.58342z_{12} \\
&\quad - 4.41120z_{13} + 5.11685z_{14} - 4.52759z_{15} \\
&\quad - 4.92150z_{16} + 4.72311z_{17} + 5.14182z_{18} \\
&\quad - 5.47117z_{19} - 5.42696z_{20} + 0.02130z_{21} \\
&\quad - 1.92210z_{22} - 0.93898z_{23} - 1.11313z_{24} \\
&\quad + 1.35254z_{25}
\end{aligned}\]

\noindent\textbf{Cutoff $R_0>7$.}
\[\begin{aligned}
\eta_{R_0>7} &= -15.95586 \\
&\quad + 2.31135z_{1} + 1.70059z_{2} - 2.60545z_{3} \\
&\quad - 2.25383z_{4} + 0.97758z_{5} + 3.19711z_{6} \\
&\quad + 2.20819z_{7} + 2.36342z_{8} + 1.89075z_{9} \\
&\quad + 2.02618z_{10} + 2.94423z_{11} - 4.22624z_{12} \\
&\quad - 4.44497z_{13} + 4.52456z_{14} - 4.51124z_{15} \\
&\quad - 4.69668z_{16} + 4.25336z_{17} + 4.92576z_{18} \\
&\quad - 4.79710z_{19} - 5.12492z_{20} + 0.15311z_{21} \\
&\quad - 0.49920z_{22} - 0.69981z_{23} - 1.51620z_{24} \\
&\quad + 0.13797z_{25}
\end{aligned}\]

\noindent\textbf{Cutoff $R_0>7.5$.}
\[\begin{aligned}
\eta_{R_0>7.5} &= -16.67762 \\
&\quad + 1.09427z_{1} + 0.71126z_{2} - 2.27695z_{3} \\
&\quad - 2.13740z_{4} + 0.84337z_{5} + 2.50574z_{6} \\
&\quad + 1.48709z_{7} + 2.10975z_{8} + 2.32529z_{9} \\
&\quad + 1.80949z_{10} + 2.91007z_{11} - 4.42637z_{12} \\
&\quad - 4.19987z_{13} + 4.61296z_{14} - 4.48366z_{15} \\
&\quad - 4.25838z_{16} + 4.63093z_{17} + 4.45977z_{18} \\
&\quad - 4.35161z_{19} - 4.48737z_{20} - 0.43454z_{21} \\
&\quad - 1.72218z_{22} - 1.12403z_{23} - 0.78580z_{24} \\
&\quad + 0.57562z_{25}
\end{aligned}\]

\noindent\textbf{Cutoff $R_0>8$.}
\[\begin{aligned}
\eta_{R_0>8} &= -17.03230 \\
&\quad + 0.30425z_{1} + 0.57243z_{2} - 1.91663z_{3} \\
&\quad - 2.05740z_{4} + 0.72434z_{5} + 1.19596z_{6} \\
&\quad + 1.80072z_{7} + 2.26204z_{8} + 2.09642z_{9} \\
&\quad + 1.66203z_{10} + 3.50711z_{11} - 4.19175z_{12} \\
&\quad - 4.26720z_{13} + 4.45971z_{14} - 4.38427z_{15} \\
&\quad - 4.73326z_{16} + 3.60257z_{17} + 4.07803z_{18} \\
&\quad - 4.16958z_{19} - 4.24462z_{20} - 1.88479z_{21} \\
&\quad - 1.77178z_{22} - 0.80441z_{23} - 0.82323z_{24} \\
&\quad + 0.27710z_{25}
\end{aligned}\]

\noindent\textbf{Cutoff $R_0>8.5$.}
\[\begin{aligned}
\eta_{R_0>8.5} &= -17.46332 \\
&\quad + 0.58275z_{1} + 0.74379z_{2} - 1.46686z_{3} \\
&\quad - 1.39764z_{4} + 0.41310z_{5} + 2.11253z_{6} \\
&\quad + 2.50561z_{7} + 1.97100z_{8} + 1.57056z_{9} \\
&\quad + 1.38951z_{10} + 3.27393z_{11} - 4.01901z_{12} \\
&\quad - 4.72283z_{13} + 3.37198z_{14} - 4.26463z_{15} \\
&\quad - 4.55176z_{16} + 3.25617z_{17} + 4.41219z_{18} \\
&\quad - 4.12698z_{19} - 4.16084z_{20} - 1.61020z_{21} \\
&\quad - 1.79279z_{22} - 0.74456z_{23} - 1.34186z_{24} \\
&\quad + 0.24905z_{25}
\end{aligned}\]

\noindent\textbf{Cutoff $R_0>9$.}
\[\begin{aligned}
\eta_{R_0>9} &= -17.94417 \\
&\quad - 0.12396z_{1} - 0.61072z_{2} - 1.35744z_{3} \\
&\quad - 0.97609z_{4} + 0.63294z_{5} + 2.28245z_{6} \\
&\quad + 1.61557z_{7} + 1.53221z_{8} + 1.63305z_{9} \\
&\quad + 1.08912z_{10} + 3.02232z_{11} - 4.69128z_{12} \\
&\quad - 3.61167z_{13} + 3.35702z_{14} - 4.17016z_{15} \\
&\quad - 3.99977z_{16} + 3.90641z_{17} + 4.25475z_{18} \\
&\quad - 3.41639z_{19} - 3.80761z_{20} - 0.89754z_{21} \\
&\quad - 1.26041z_{22} - 1.21868z_{23} - 1.14719z_{24} \\
&\quad + 0.78086z_{25}
\end{aligned}\]

\noindent\textbf{Cutoff $R_0>9.5$.}
\[\begin{aligned}
\eta_{R_0>9.5} &= -18.12105 \\
&\quad - 0.91466z_{1} - 0.53873z_{2} - 0.70747z_{3} \\
&\quad - 0.83324z_{4} + 0.50301z_{5} + 1.68710z_{6} \\
&\quad + 2.53503z_{7} + 0.76397z_{8} + 1.97325z_{9} \\
&\quad + 1.43028z_{10} + 2.77299z_{11} - 4.09665z_{12} \\
&\quad - 3.92236z_{13} + 3.23646z_{14} - 4.47736z_{15} \\
&\quad - 3.85262z_{16} + 3.19528z_{17} + 4.34650z_{18} \\
&\quad - 4.07710z_{19} - 3.22384z_{20} - 0.78387z_{21} \\
&\quad - 1.02109z_{22} - 1.42438z_{23} - 0.66097z_{24} \\
&\quad + 0.02923z_{25}
\end{aligned}\]

\noindent\textbf{Cutoff $R_0>10$.}
\[\begin{aligned}
\eta_{R_0>10} &= -18.25117 \\
&\quad - 0.36383z_{1} - 0.38071z_{2} - 0.23800z_{3} \\
&\quad - 0.24304z_{4} + 0.29985z_{5} + 2.49099z_{6} \\
&\quad + 2.12711z_{7} + 1.33973z_{8} + 1.34714z_{9} \\
&\quad + 0.52678z_{10} + 3.07155z_{11} - 3.99309z_{12} \\
&\quad - 3.99357z_{13} + 2.91135z_{14} - 3.50006z_{15} \\
&\quad - 4.00640z_{16} + 2.90362z_{17} + 4.34379z_{18} \\
&\quad - 3.66913z_{19} - 3.63895z_{20} - 0.43970z_{21} \\
&\quad - 0.92267z_{22} - 1.05875z_{23} - 1.06672z_{24} \\
&\quad + 0.57019z_{25}
\end{aligned}\]

\section{Methodological scope and robustness considerations}\label{subsec:reviewer-questions}

The following table answers several methodological questions that are important for interpreting REEBLRA correctly.

\begin{longtable}{@{}p{0.22\linewidth}p{0.39\linewidth}p{0.32\linewidth}@{}} 
\toprule
\textbf{Question} & \textbf{REEBLRA response} & \textbf{Residual limitation / recommended check} \\
\midrule
\endfirsthead
\toprule
\textbf{Question} & \textbf{REEBLRA response} & \textbf{Residual limitation / recommended check} \\
\midrule
\endhead
Why not fit $R$ directly? & The method targets the decision surface $R=c$, so it need not reproduce the full response magnitude. Ordered cutoffs recover coarse magnitude. & Direct regression remains a legitimate comparator and may be preferable when accurate numerical $R$ itself is the scientific endpoint; REEBLRA makes no universal superiority claim. \\
Why is $c=1$ emphasized? & Under standard invasion conditions, one is the canonical epidemic threshold. & Other cutoffs are scientifically legitimate; conclusions are cutoff-dependent. The app supports user-defined cutoffs and ordered brackets. \\
Why Bayesian? & Proper priors provide regularization, finite posterior coefficients under separation, and posterior uncertainty. & Results may be prior-sensitive when classes are sparse or separation is extreme. Use MCMC and prior-sensitivity analysis for inferential claims. \\
What about perfect separation? & Bayesian coefficients remain finite under proper priors, unlike ordinary unpenalized MLE. & Report class support and coefficient uncertainty; Firth/bias-reduced logistic regression is a useful sensitivity analysis. \\
What about multicollinearity? & VIF and correlation diagnostics are reported, and shrinkage reduces numerical instability. & Priors do not restore identifiability. Reparameterize, aggregate, redesign ranges, or use structured shrinkage when coefficient interpretation is unstable. \\
Could the feature basis leak mechanistic knowledge? & The implemented bases are generic transformations (linear, log-linear, log-quadratic, compact log-cubic) and do not encode the analytic NGM formula. & If basis choice is tuned using performance, it must be selected within a nested validation procedure; pre-specification is preferable. \\
Why logistic regression rather than a more flexible classifier? & Logistic equations are transparent, differentiable, exportable, and permit direct sign/magnitude comparison with mechanistic gradients. & Trees, kernels, neural networks, or Gaussian processes may fit more complex boundaries but generally sacrifice the compact symbolic equation that is central to REEBLRA. \\
How many simulations are required? & There is no universal number; adequacy depends on parameter dimension, threshold curvature, class support, and feature complexity. LHS improves space filling relative to naive random sampling. & Report class counts and held-out stability; increase simulations when folds are sparse or metrics/coefficients are unstable. Learning-curve analysis is a useful extension. \\
How is overfitting controlled? & Feature construction and standardization are rebuilt within each training fold, and held-out predictions are pooled for cross-validation metrics. & Small samples, many features, or adaptive model selection require stronger regularization and nested resampling. \\
Why use successive binary thresholds rather than a single ordinal model? & Separate threshold equations retain an explicit symbolic rule for every cutoff; cumulative probabilities are then projected to a non-increasing sequence before bracket probabilities are formed. & This is an established threshold-logit strategy, not a new ordinal-statistical family. The present fits are not jointly constrained; compare with proportional/partial proportional odds and report ordering violations.\\
Are high AUC values sufficient? & No. REEBLRA also reports balanced accuracy, sensitivity, specificity, precision, F1, NPV, Cohen's $\kappa$, Brier score, confusion matrices, and class support. & Discrimination can remain high when threshold calibration is imperfect; probability calibration should be examined when probabilities themselves are used for decisions. \\
Do threshold-spanning simulation ranges represent epidemiological priors? & No. They are experimental-design ranges chosen to populate both sides of the decision boundary. & Performance is conditional on the sampled domain and should not be interpreted as population prevalence or real-world parameter frequency. \\
What if analytic $R_0$ is unknown? & The target may be a growth-based numerical \texttt{R0estimate} or another externally computed reproduction-number summary. & Different estimation methods and their underlying assumptions may produce different reproduction-number estimates, potentially affecting threshold classifications and fitted REEBLRA equations. The estimation method should therefore be clearly documented, and robustness across alternative estimators should be assessed when feasible. \\
Can coefficients be interpreted causally? & No. Coefficients describe the fitted threshold surface and may be compared with mechanistic derivatives. & They are not causal effects, elasticities, or proof that a parameter mechanism is identifiable. \\
Does the method extrapolate? & It is intended for interpolation and threshold screening within the sampled domain. & No guarantee is made outside the simulated parameter region, under structural model change, or across different data-generating mechanisms. \\
What about stochastic simulators? & The framework can conceptually accept simulation-derived labels or probabilities. & The current built-in examples are deterministic ODEs; stochastic models require repeated simulations and explicit treatment of label uncertainty. \\
Does REEBLRA replace the NGM? & No. It provides a simulation-calibrated, interpretable surrogate and a structural cross-check. & When a valid analytic next-generation derivation is tractable, it remains the mechanistic reference standard. \\
\bottomrule
\end{longtable}

\section{Recommended reporting checklist}\label{app:checklist}
Studies using REEBLRA should report the data mode; built-in equations or all custom state and auxiliary equations; selected nonnegative states; parameter ranges; LHS size and seed; initial conditions; time horizon and output step; target infectious or derived output; growth-window and conversion rule; handling of non-positive growth; threshold source and cutoff; bracket step; Bayesian priors and inference algorithm; validation scheme; class counts; declared predictor feature basis; standardized and unstandardized fitted-feature equations; raw and row-normalized coefficient summaries; collinearity diagnostics; accuracy, balanced accuracy, AUC, sensitivity/recall, specificity, precision, F1, NPV, Cohen's $\kappa$, Brier score where applicable, one-vs-rest multiclass metrics, and class support; and the domain and limitations of the simulation design. 

Also report whether threshold labels are analytic or estimated from simulated trajectories; the exact estimator and regularization used; pre-projection ordering violations; uncertainty or calibration diagnostics; threshold-neighborhood error; feature count relative to support; and pointwise gradient-cosine statistics when computed.

\end{appendices}

\bibliographystyle{unsrtnat}
\bibliography{REEBLRA_references}

\end{document}